\documentclass[11pt, a4paper]{article}
\usepackage[margin=1in]{geometry}

\usepackage[T1]{fontenc}    % use 8-bit T1 fonts
\usepackage{hyperref}       % hyperlinks
\usepackage{url}            % simple URL typesetting
\usepackage{booktabs}       % professional-quality tables
\usepackage{amsfonts}       % blackboard math symbols
\usepackage{nicefrac}       % compact symbols for 1/2, etc.
\usepackage{microtype}      % microtypography
\usepackage{xcolor}         % colors
 
\usepackage{amsmath}
\usepackage{amsthm}
\usepackage{amssymb}
\usepackage[ruled,vlined]{algorithm2e}
\usepackage{caption}
\usepackage{subcaption}
\usepackage{wrapfig}
\usepackage{hyperref}
\usepackage{bm}
\usepackage{tikzit}
\tikzstyle{basic}=[fill=white, draw=black, shape=circle]
\tikzstyle{square}=[fill=white, draw=black, shape=rectangle]
\tikzstyle{big dashed}=[fill=white, draw=black, shape=circle, minimum width=1cm, dashed]
\tikzstyle{vertical ellipse dashed}=[fill=none, draw=blue, minimum width=0.75cm, minimum height=3cm, ellipse, dashed, tikzit shape=rectangle, tikzit draw=blue, tikzit fill=white]
\tikzstyle{small vertical ellipse dashed}=[fill=none, draw=blue, shape=circle, tikzit fill=white, tikzit draw=blue, dashed, minimum width=0.75cm, minimum height=1.5cm, tikzit shape=rectangle, ellipse]
\tikzstyle{tiny vertical ellipse dashed}=[fill=none, draw=blue, shape=circle, tikzit fill=white, ellipse, dashed, minimum width=0.75cm, minimum height=1cm, tikzit shape=rectangle]
\tikzstyle{red}=[fill=red, draw=black, shape=circle]
\tikzstyle{green}=[fill={rgb,255: red,0; green,128; blue,128}, draw=black, shape=circle]
\tikzstyle{blue}=[fill=blue, draw=black, shape=circle]
\tikzstyle{huge dashed}=[fill=white, draw=black, shape=circle, dashed, minimum width=2cm]
\tikzstyle{medium}=[fill=white, draw=black, shape=circle, minimum width=1cm]
\tikzstyle{pale green}=[fill={rgb,255: red,173; green,231; blue,0}, draw=black, shape=circle, minimum width=1cm]
\tikzstyle{horizontal ellipse dashed}=[fill=white, draw=black, tikzit draw=magenta, tikzit shape=rectangle, minimum width=3cm, minimum height=0.75cm, ellipse, dashed]
\tikzstyle{minsize}=[fill=white, draw=black, shape=circle, minimum width=0.75cm]
\tikzstyle{horizontal ellipse green}=[fill={rgb,255: red,191; green,255; blue,0}, draw=black, tikzit draw={rgb,255: red,191; green,255; blue,0}, tikzit shape=rectangle, minimum width=3cm, minimum height=0.75cm, ellipse, dashed]
\tikzstyle{horizontal ellipse blue}=[fill={rgb,255: red,107; green,203; blue,255}, draw=black, tikzit draw=blue, tikzit shape=rectangle, minimum width=3cm, minimum height=0.75cm, ellipse, dashed]
\tikzstyle{smallblack}=[fill=black, draw=black, shape=circle, inner sep=0 pt, minimum size=3 pt]
\tikzstyle{smallSquare}=[fill=white, draw=black, shape=rectangle, inner sep=0 pt, minimum size=6 pt]
\tikzstyle{smallCircle}=[fill=white, draw=black, shape=circle, inner sep=0 pt, minimum size=6 pt]
\tikzstyle{big vertical ellipse dashed}=[fill=none, draw=blue, shape=circle, tikzit shape=rectangle, ellipse, dashed, minimum width=0.95cm, minimum height=3.7cm]
\tikzstyle{smallred}=[fill=red, draw=red, shape=circle, inner sep=0 pt, minimum size=3 pt]

\tikzstyle{directed}=[->]
\tikzstyle{undirected}=[-, line width=1pt]
\tikzstyle{directed red}=[draw=red, ->, line width=1pt]
\tikzstyle{directed green}=[draw={rgb,255: red,0; green,128; blue,128}, ->, line width=1pt]
\tikzstyle{directed blue}=[draw=blue, ->, line width=1pt]
\tikzstyle{directed purple}=[draw={rgb,255: red,128; green,0; blue,128}, ->, line width=1pt]
\tikzstyle{undirected red}=[-, draw=red, line width=1pt]
\tikzstyle{undirected green}=[-, draw={rgb,255: red,0; green,107; blue,61}, line width=1pt]
\tikzstyle{undirected blue}=[-, draw=blue, line width=1pt]
\tikzstyle{undirected purple}=[-, draw={rgb,255: red,128; green,0; blue,128}, line width=1pt]
\tikzstyle{undirected dashed}=[-, line width=1pt, dashed]
\tikzstyle{orange dashed}=[-, draw={rgb,255: red,255; green,128; blue,0}, dashed, line width=1.5pt]
\tikzstyle{directed dash}=[->, dashed]
\tikzstyle{blue dashed}=[-, draw=blue, dashed, line width=1pt]
\tikzstyle{green dashed}=[-, draw={rgb,255: red,0; green,162; blue,0}, dashed, line width=1pt]
\tikzstyle{blue filled}=[-, fill={blue!20}, draw=blue, line width=1pt, opacity=0.5, tikzit fill=white]
\tikzstyle{red filled}=[-, fill={red!20}, line width=1pt, draw=red, opacity=0.5, tikzit fill=white]
\tikzstyle{green filled}=[-, line width=1pt, draw={rgb,255: red,0; green,107; blue,61}, opacity=0.5, tikzit fill=white, fill={rgb,255: red,149; green,255; blue,179}]
\tikzstyle{orange filled}=[-, fill={orange!20}, draw=orange, line width=1pt, opacity=0.5, tikzit fill=white]
\tikzstyle{undirected dashed}=[-, draw=black, dashed, line width=1pt]

\usepackage{multicol,multirow}

\newtheorem{theorem}{Theorem} [section]

\newtheorem{lemma}{Lemma} [section]
\newtheorem{remark}{Remark} [section]

\newtheorem{corollary}{Corollary} [section]

\newcommand{\transpose}{\intercal}                    % Transpose
\newcommand{\inner}[2]{\langle #1 , #2 \rangle} % <a,b>
\newcommand{\norm}[1]{\left\| #1\right\|}                  % ||a||
\newcommand{\abs}[1]{\left\lvert#1\right\rvert}
\DeclareMathOperator*{\argmin}{arg\,min}        % arg min
\DeclareMathOperator*{\argmax}{arg\,max}        % arg max
\newcommand{\tr}{\mathrm{tr}}                   % trace

\newcommand{\sign}{\mathrm{sign}}
\newcommand{\poly}{\mathrm{poly}}
\renewcommand{\P}{\mathsf{P}}
\newcommand{\NP}{\mathsf{NP}}

\newcommand{\vol}{\mathrm{vol}}

\newcommand{\bigo}[1]{O\!\left(#1\right)}

\newcommand{\E}{\mathbb{E}}

\newcommand{\R}{\mathbb{R}}

\newcommand{\union}{\cup}
\newcommand{\intersect}{\cap}
\newcommand{\cardinality}[1]{\abs{#1}}
\newcommand{\setcomplement}[1]{\overline{#1}}

\definecolor{indiagreen}{rgb}{0.07, 0.53, 0.03}

\newcommand{\twopartdefow}[3]
{
	\left\{
		\begin{array}{ll}
			#1 & \mbox{if } #2 \\
			#3 & \mbox{otherwise}
		\end{array}
	\right.
}

\newcommand{\threepartdefow}[5]
{
	\left\{
		\begin{array}{lll}
			#1 & \mbox{if } #2 \\
			#3 & \mbox{if } #4 \\
			#5 & \mbox{otherwise}
		\end{array}
	\right.
}
\newcommand{\allnotation}[1]{#1}

\renewcommand{\vec}[1]{{\allnotation{\bm{#1}}}}
\newcommand{\vecf}{\vec{f}}
\newcommand{\vecg}{\vec{g}}
\newcommand{\vecx}{\vec{x}}
\newcommand{\vecy}{\vec{y}}

\newcommand{\vecu}{\vec{u}}
\newcommand{\vecv}{\vec{v}}
\newcommand{\vecr}{\vec{r}}
\newcommand{\constvec}{\vec{1}}

\newcommand{\set}[1]{{\allnotation{#1}}}
\newcommand{\setv}{\set{V}}
\newcommand{\sete}{\set{E}}
\newcommand{\sets}{\set{S}}
\newcommand{\setl}{\set{L}}
\newcommand{\setr}{\set{R}}
\newcommand{\seta}{\set{A}}
\newcommand{\setb}{\set{B}}
\newcommand{\setc}{\set{C}}
\newcommand{\setu}{\set{U}}
\newcommand{\setp}{\set{P}}
\newcommand{\vertexset}{\setv}
\newcommand{\edgeset}{\sete}

\newcommand{\graph}[1]{{\allnotation{#1}}}
\newcommand{\graphg}{\graph{G}}
\newcommand{\graphh}{\graph{H}}

\newcommand{\mat}[1]{{\allnotation{\mathbf{#1}}}}
\newcommand{\mata}{\mat{A}}

\newcommand{\lap}{\mat{L}}

\newcommand{\signlap}{\mat{J}}
\newcommand{\signlapn}{\mat{Z}}
\newcommand{\degm}{\mat{D}}
\newcommand{\degmhalf}{\degm^{\allnotation{\frac{1}{2}}}}
\newcommand{\degmhalfneg}{\degm^{\allnotation{-\frac{1}{2}}}}
\newcommand{\adj}{\mat{A}}

\newcommand{\identity}{\mat{I}}

\renewcommand{\deg}{{\allnotation{d}}}
\newcommand{\weight}{{\allnotation{w}}}
\newcommand{\bipart}{{\allnotation{\beta}}}

\newcommand{\discrep}{{\allnotation{\Delta}}}

\newcommand{\firstdef}[1]{\emph{#1}}

\newcommand{\iverson}[1]{\pmb{\left[\vphantom{#1}\right.} #1 \pmb{\left.\vphantom{#1}\right]}}

\title{Finding Bipartite Components in Hypergraphs\footnote{A preliminary version of this work appeared at   NeurIPS~2021. This work is supported by a Langmuir PhD Scholarship, and an EPSRC Early Career Fellowship~(EP/T00729X/1).}}

\author{%
Peter Macgregor\\ University of Edinburgh
\and
He Sun\\ 
University of Edinburgh}
\date{}

\numberwithin{equation}{section}

\begin{document}

\maketitle

\begin{abstract}
Hypergraphs are important  objects to model ternary or higher-order  relations of objects, and have a number of  applications in analysing many complex datasets occurring in practice. In this work we study a new heat diffusion process in hypergraphs, and employ this process to design a polynomial-time algorithm that approximately finds bipartite components in a hypergraph.  We theoretically prove the performance of our proposed algorithm, and compare it against the previous state-of-the-art through extensive experimental analysis on both synthetic and real-world datasets. We find that our new algorithm consistently and significantly outperforms the previous state-of-the-art across a wide range of hypergraphs.
\end{abstract}

% We define some of the notation that is used commonly throughout this chapter.
\newcommand{\lur}{\setl \union \setr}   % Union of sets L and R
\newcommand{\rank}{\mathrm{rank}}                   % The rank of a given hyperedge.
\newcommand{\cals}{\mathcal{S}}         % Sets used in proof of LP
\newcommand{\cali}{\mathcal{I}}
\newcommand{\cale}{\mathcal{E}}
\newcommand{\cut}{\mathrm{cut}}
\newcommand{\dxdt}{\frac{\mathrm{d} x}{\mathrm{d} t}}
\newcommand{\dfdt}{\frac{\mathrm{d} \vecf_t}{\mathrm{d} t}}
\newcommand{\signlaph}{\signlap_\graphh}
\newcommand{\signlapg}{\signlap_\graphg}
\newcommand{\sfe}{\sets_{\vecf}(e)}
\newcommand{\sfte}{\sets_{\vecf_t}(e)}
\newcommand{\ife}{\set{I}_{\vecf}(e)}
\newcommand{\ifte}{\set{I}_{\vecf_t}(e)}
\newcommand{\diffalgname}{\textsc{FindBipartiteComponents}}
\newcommand{\diffalgshortname}{\textsc{FBC}}
\newcommand{\diffalgapproxname}{\textsc{FBCApprox}}
\newcommand{\diffalgapproxshortname}{\textsc{FBCA}}

% Expressions for the discrepancy of a given edge.
\newcommand{\deltaesquare}[1]{\left(\max_{u \in e} #1(u) + \min_{v \in e}#1(v)\right)^2}
\newcommand{\deltaeabs}[1]{\abs{\max_{u \in e} #1(u) + \min_{v \in e} #1(v)}}

\section{Introduction}
Spectral methods study the efficient matrix representation of graphs and datasets, and apply the algebraic properties of these matrices to design efficient algorithms. Over the last three decades, spectral methods have become one of the most powerful techniques in machine learning, and have had comprehensive applications in a wide range of domains, including clustering~\cite{ngSpectralClusteringAnalysis2001,vonluxburgTutorialSpectralClustering2007}, image and video segmentation~\cite{shiNormalizedCutsImage2000}, and network analysis~\cite{searySpectralMethodsAnalyzing2003}, among many others. While the success of this line of research is based on   our rich understanding of      Laplacian operators of graphs, there has been a sequence of  very recent work studying \emph{non-linear} Laplacian operators for more complex objects~(i.e., hypergraphs) and employing these non-linear operators to design hypergraph algorithms with better performance. 

\subsection{Our contribution}

In this work, we study the non-linear Laplacian-type operators for hypergraphs, and employ such an operator to design a  polynomial-time algorithm for finding bipartite components in hypergraphs.  The main contribution of our work is as follows.

First of all, we introduce and study a non-linear Laplacian-type operator  $\signlap_\graphh$ for any hypergraph $\graphh$.
While  we'll formally define the operator $\signlap_\graphh$ in Section~\ref{sec:diff_and_alg}, one can informally think about $\signlap_\graphh$ as a variant of the standard non-linear hypergraph Laplacian $\lap_\graphh$ studied in~\cite{chanSpectralPropertiesHypergraph2018, liSubmodularHypergraphsPLaplacians2018, takaiHypergraphClusteringBased2020}, where the variation  is needed to study the other end of the spectrum of $\lap_\graphh$.
We present a polynomial-time algorithm that  finds an eigenvalue $\lambda$ and its associated eigenvector of $\signlap_\graphh$.
The algorithm is based on the following  heat diffusion process:
% \he{Peter: should we write $f$ or $\varphi$? we need to be consistent with Figure 1, and section 3}
% \peter{I'm beginning to think that we should just use $f$ everywhere. The only argument against it, as far as I can see, is that the rate of change $(df)/(dt) = - D^{-1} J f$ is (slightly) less intuitive than the one for $\varphi$ which looks like the random walk laplacian update.}
starting from an arbitrary vector $\vecf_0 \in \R^n$ that describes the initial heat distribution among the vertices, we use $\vecf_0$ to  construct some graph\footnote{Note that the word `graph' always refers to a non-hyper graph. Similarly, we always use $\lap_\graphh$ to refer to the \emph{non-linear} hypergraph Laplacian operator, and use $\lap_\graphg$ as the standard graph Laplacian.}
$\graphg_0$, and use the diffusion process in $\graphg_0$ to represent the  one in the original hypergraph $\graphh$ and update $\vecf_t$; this process continues until the time at which $\graphg_0$ cannot be used to appropriately simulate the diffusion process in $\graphh$ any more.
At this point, we use the currently maintained $\vecf_t$ to construct another graph $\graphg_t$ to simulate the diffusion process in $\graphh$, and update $\vecf_t$.
This process continues until the vector $\vecf_t$ converges; see Figure~\ref{fig:illustration} for illustration.
We theoretically prove that this heat diffusion process is unique, well-defined, and our maintained vector $\vecf_t$ converges to an eigenvector of $\signlap_\graphh$. While this result is quite interesting on its own and forms the basis of our second result, our analysis shows that, for certain hypergraphs $\graphh$, both the operator $\signlap_\graphh$ and $\lap_\graphh$ could have $\omega(1)$ eigenvectors.
This result answers an open question  in \cite{chanSpectralPropertiesHypergraph2018}, which asks whether $\lap_\graphh$ could have more than $2$ eigenvectors\footnote{We underline that, while the operator $\lap_\graphg$ of a graph $\graphg$ has $n$ eigenvalues, the number of eigenvalues of $\lap_\graphh$ is unknown because of its non-linearity.}.

Secondly, we present a polynomial-time algorithm that, given a hypergraph $\graphh=(\vertexset_\graphh, \edgeset_\graphh, \weight)$ as input, 
finds disjoint subsets $\setl, \setr \subset \vertexset_\graphh$ that are highly connected with each other.
The key to our algorithm is a Cheeger-type inequality for hypergraphs that relates the spectrum of $\signlap_\graphh$ and the bipartiteness ratio of $\graphh$, an analog of the bipartiteness ratio studied in \cite{trevisanMaxCutSmallest2012} for graphs. Both the design and analysis of our algorithm is inspired  by \cite{trevisanMaxCutSmallest2012}, however our analysis is much more involved  because of the non-linear operator $\signlap_\graphh$ and hyperedges of different ranks. Our second result alone answers an open question posed by \cite{yoshidaCheegerInequalitiesSubmodular2019},  which asks whether there is a hypergraph operator which satisfies a Cheeger-type inequality for bipartiteness.

The significance of our work is further demonstrated by  extensive experimental studies of our   algorithms on both synthetic and real-world datasets.
In particular, on the well-known Penn Treebank corpus that contains $49,208$ sentences and over 1 million words, 
our \emph{purely unsupervised} algorithm is able to identify a significant fraction of  verbs  from non-verbs in its two output clusters.
Hence, we believe that  our work  could potentially have many applications in unsupervised learning for hypergraphs.
Using the publicly available code of our implementation, we welcome the reader to explore further applications of our work in even more diverse datasets. 

\begin{figure}[t]
    \centering
    \scalebox{0.8}{
    \input{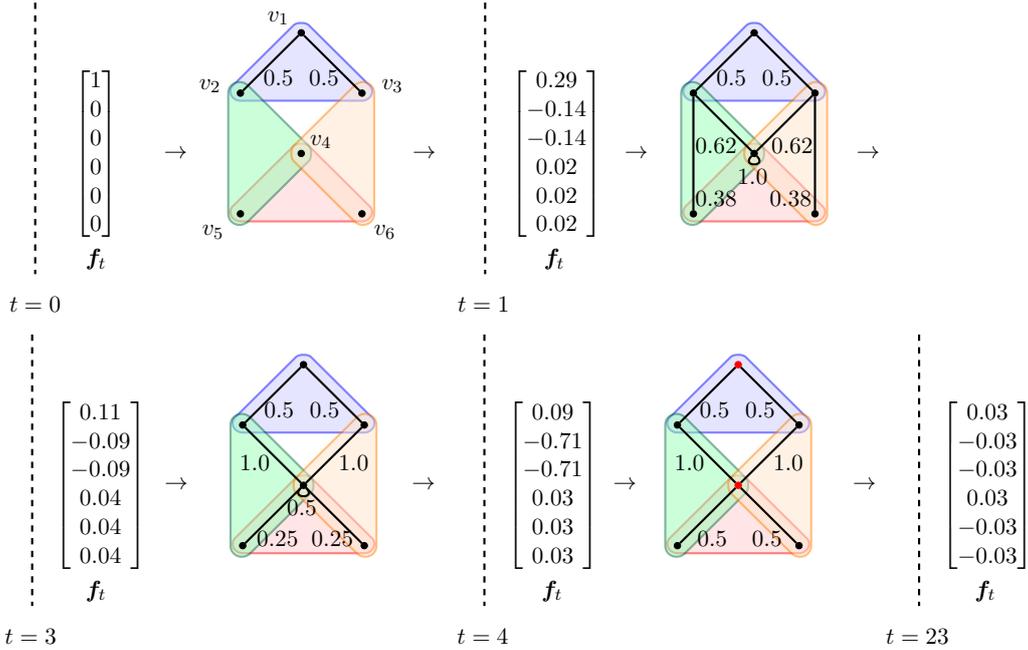}
    }
\caption{\small{Illustration of our proposed diffusion process. In each time step, we construct a graph $\graphg$ based on the current vector $\vecf_t$, and update $\vecf_t$ with the $\signlap_\graphg$ operator.
Notice that the graph $\graphg$ changes throughout the execution of the algorithm, and that the final $\vecf_t$ vector can be used to partition the vertices of $\graphh$ into two well-connected sets (all the edges are adjacent to both sets), by splitting according to positive and negative entries.
This specific example with $\vecf_t$ values is generated by the publicly available implementation of our algorithm.}
\label{fig:illustration}}
\end{figure}

\subsection{Related work}
The spectral theory of hypergraphs using non-linear operators is introduced in~\cite{chanSpectralPropertiesHypergraph2018} and generalised in~\cite{yoshidaCheegerInequalitiesSubmodular2019}. The operator they describe is applied for hypergraph clustering applications in~\cite{liSubmodularHypergraphsPLaplacians2018, takaiHypergraphClusteringBased2020}.
There are many approaches for finding clusters in hypergraphs by constructing a graph which approximates the hypergraph and using a graph clustering algorithm directly~\cite{chitraRandomWalksHypergraphs2019, liInhomogeneousHypergraphClustering2017, zhouLearningHypergraphsClustering2006}.
Another approach for hypergraph clustering is based on tensor spectral decomposition~\cite{ghoshdastidarConsistencySpectralPartitioning2014, leordeanuEfficientHypergraphClustering2012}.
\cite{macgregorLocalAlgorithmsFinding2021, mooreActiveLearningNode2011, zhuHomophilyGraphNeural2020} consider the problem of finding densely connected clusters in graphs.
Heat diffusion processes are used for clustering graphs in~\cite{chungLocalGraphPartitioning2009, klosterHeatKernelBased2014}.
\cite{fountoulakisPNormFlowDiffusion2020} studies a different, flow-based diffusion process for finding clusters in graphs, and \cite{fountoulakisLocalHyperflowDiffusion2021} generalises this to hypergraphs. We note that 
all of these methods solve a different problem to ours, and cannot be compared directly.
Our algorithm is related to the hypergraph max cut problem, and the state-of-the-art approximation algorithm is given by~\cite{zhangImprovedApproximationsMax2004}.
% There are no known non-SDP based approximation algorithms for the hypergraph max cut problem. 
\cite{trevisanMaxCutSmallest2012} introduces graph bipartiteness and gives %a non-SDP
an
approximation algorithm for the graph max cut problem.
To the best of our knowledge, we are the first to generalise this notion of bipartiteness to hypergraphs.
Finally, we note that there have been recent improvements in the time complexity for solving linear programs~\cite{cohenSolvingLinearPrograms2021, vandenbrandDeterministicLinearProgram2020} although we do not take these into account in our analysis since the goal of this paper is not to obtain the fastest algorithm possible.

\subsection{Organisation}
The paper is organised as follows. In Section~\ref{sec:notation}, we introduce the necessary preliminaries and notation. Section~\ref{sec:baselines} describes the limitations of existing methods for finding bipartite components in hypergraphs and in Section~\ref{sec:diff_and_alg} we introduce the new algorithm.
We prove the main theoretical result in Section~\ref{sec:thm1proof}.
Finally, in Section~\ref{sec:experiments} we include some experimental results demonstrating the applicability of the new hypergraph operator.

\section{Preliminaries} \label{sec:notation}

\paragraph{Graphs.}
Throughout the paper,
`graph' always refers to a non-hyper graph.
We always use $\graphg=(\vertexset_\graphg,\edgeset_\graphg, \weight)$ to express a graph, in which every edge $e \in \edgeset_\graphg$ consists of two vertices in $\vertexset_\graphg$ and we let $n = \cardinality{\vertexset_\graphg}$.
The degree of any vertex $u \in \vertexset_\graphg$ is defined by  $\deg_\graphg(u) \triangleq \sum_{v \in \vertexset_\graphg} \weight(u,v)$,
and for any $\sets \subseteq \vertexset_\graphg$ the volume of $\sets$ is defined by $\vol_\graphg(\sets) \triangleq \sum_{u\in \sets } \deg_\graphg(u)$. Following \cite{trevisanMaxCutSmallest2012}, the \firstdef{bipartiteness ratio} of any disjoint sets $\setl,\setr \subset \vertexset_\graphg$ is defined by
% \[
% \beta_G(L,R) \triangleq 1 - \frac{2\cdot w(L,R)}{\vol_G(L\cup R)},
% \]
 
\[
\bipart_\graphg(\setl, \setr) \triangleq \frac{2 \weight(\setl, \setl) + 2 \weight(\setr, \setr) + \weight(\setl \union \setr, \setcomplement{\setl \union \setr})}{\vol_\graphg(\setl \union \setr)}
\]
where $\weight(\seta, \setb) = \sum_{(u, v) \in \seta \times \setb} \weight(u, v)$, 
and we further define
$
\bipart_\graphg \triangleq \min_{\sets \subset \vertexset} \bipart_\graphg(\sets, \vertexset \setminus \sets)$.  
Notice that a low $\bipart_\graphg$-value means that there is a dense cut between $\setl$ and $\setr$, and there is a sparse cut between $\setl \union \setr$ and $\vertexset \setminus (\setl \union \setr)$. In particular, $\bipart_\graphg=0$ implies that $(\setl, \setr)$ forms a bipartite component of $\graphg$. 
 We use $\degm_\graphg$ to denote the $n \times n$ diagonal matrix whose entries are $(\degm_\graphg)_{uu} = \deg_\graphg(u)$, for all $u \in \vertexset_\graphg$.
 Moreover, we use $\adj_\graphg$ to denote the $n \times n$ adjacency matrix whose entries are $(\adj_\graphg)_{uv} = \weight(u,v)$, for all $u, v \in \vertexset_\graphg$.
 The Laplacian matrix is defined by $\lap_\graphg \triangleq \degm_\graphg - \adj_\graphg$.
 In addition, we define the signless Laplacian operator $\signlap_\graphg \triangleq \degm_\graphg + \adj_\graphg$, and its normalised counterpart $\signlapn_\graphg \triangleq \degmhalfneg_\graphg \signlap_\graphg \degmhalfneg_\graphg$.  
For any real and symmetric matrix $\mata$, the eigenvalues of $\mata$ are denoted by $\lambda_1(\mata) \leq \dots \leq \lambda_n(\mata)$, and  the eigenvector associated with $\lambda_i(\mata)$ is denoted by $\vecf_i(\mata)$ for $1 \leq i \leq n$. 
% It is shown in \cite{Trevisan2012} that $\lambda_1(\mathcal{J}_G)/2 \leq \beta_G \leq\sqrt{2\cdot \lambda_1 (\mathcal{J}_G)}$ holds for any $2$-graph. 

%\begin{lemma}[\cite{Trevisan2012}] \label{lem:trev_graphs}
 %   It holds for any $2$-graph $G$ that
  %  $\lambda_1(\mathcal{J}_G)/2 \leq \beta_G \leq\sqrt{2\cdot \lambda_1 (\mathcal{J}_G)}$.
%\end{lemma}

\paragraph{Hypergraphs.}
Let $\graphh=(\vertexset_\graphh,\edgeset_\graphh,\weight)$ be a hypergraph with $n=\cardinality{\vertexset_\graphh}$ vertices and weight function $\weight:\edgeset_\graphh \mapsto \R^+$. For any vertex $v\in \vertexset_\graphh$, the degree of $v$ is defined by $\deg_\graphh(v) \triangleq \sum_{e\in \edgeset_\graphh} \weight(e)\cdot \iverson{v\in e}$, where $\iverson{X}=1$ if event $X$ holds and $\iverson{X}=0 $ otherwise.
The \firstdef{rank} of edge $e\in \edgeset_\graphh$ is written as $\rank(e)$ and is defined to be the total number of vertices in $e$. For any $\seta, \setb \subset \vertexset_\graphh$, the cut value between $\seta$ and $\setb$ is defined by
\[
\weight(\seta, \setb) \triangleq
\sum_{e\in \edgeset_\graphh} \weight(e) \cdot \iverson{e \intersect \seta \neq \emptyset \land e \intersect \setb \neq \emptyset}.
\]
Sometimes, we are required to analyse the weights of edges that intersect some vertex sets and not others. To this end, we define for any $\seta, \setb, \setc \subseteq \vertexset_\graphh$ that 
\[
\weight(\seta, \setb~|~\setc) \triangleq \sum_{e \in \edgeset_\graphh} \weight(e)\cdot \iverson{e \intersect \seta \neq \emptyset \land e \intersect \setb \neq \emptyset \land e \intersect C=\emptyset},
\]
and we sometimes write $\weight(\seta~|~\setc) \triangleq \weight(\seta, \seta~|~\setc)$ for simplicity.
Generalising the notion of the bipartiteness ratio of a graph, the \firstdef{bipartiteness ratio} of sets $\setl,\setr$ in a hypergraph $\graphh$ is defined by 
    \[
        \bipart_\graphh(\setl, \setr) \triangleq \frac{2\weight(\setl | \setcomplement{\setl}) + 2\weight(\setr | \setcomplement{\setr}) + \weight(\setl, \setcomplement{\lur} | \setr) + w(\setr, \setcomplement{\lur}| \setl)}{\vol(\lur)},
    \]
    and we define $$\bipart_H\triangleq \min_{\sets \subset \vertexset_\graphh}\bipart_\graphh(\sets, \vertexset_\graphh \setminus \sets).$$
For any hypergraph $H$ and $\vecf \in \R^n$, we define the discrepancy of an edge $e \in \edgeset_\graphh$ with respect to $\vecf$ as
\[
\discrep_{\vecf} (e) \triangleq \max_{u\in e} \vecf(u) + \min_{v\in e} \vecf(v).
\] 
Moreover, the weighted discrepancy of $e \in \edgeset_\graphh$ is
\[
    c_{\vecf}(e) \triangleq \weight(e) \abs{\discrep_{\vecf}(e)},
\]
and for any set $\sets \subseteq \edgeset_\graphh$  we have
\[
    c_{\vecf}(\sets) \triangleq \sum_{e \in \sets} c_{\vecf}(e).
\]
The discrepancy of $\graphh$ with respect to $\vecf$ is defined by
\[
D(\vecf) \triangleq\frac{\sum_{e\in \edgeset_\graphh} \weight(e) \cdot \discrep_{\vecf}(e)^2 }{\sum_{v\in \vertexset_\graphh} \deg_\graphh(v)\cdot \vecf(v)^2}.
\]
For any
$\vecf, \vecg \in \R^n$, the weighted inner product between $\vecf$ and $\vecg$ is defined by 
$$\inner{\vecf}{\vecg}_w \triangleq \vecf^\transpose \degm_\graphh \vecg,$$
where $\degm_\graphh \in \R^{n \times n}$ is the diagonal matrix consisting of the degrees of all the vertices of $\graphh$. The weighted norm of $\vecf$ is defined by $\norm{\vecf}_w^2 \triangleq  \inner{\vecf}{\vecf}_w$.

For any non-linear operator $\mat{J}: \R^n \mapsto \R^n$, we say that $(\lambda, \vecf)$ is an eigen-pair if and only if $\mat{J} \vecf = \lambda \vecf$ and note that in general, a non-linear operator can have any number of eigenvalues and eigenvectors.
Additionally, we define the weighted Rayleigh quotient of the operator $\signlap$ and vector $\vecf$ to be
\[
    R_\signlap(\vecf) \triangleq \frac{\vecf^\transpose \signlap \vecf}{\norm{\vecf}_\weight^2}.
\]
It is  important to remember that throughout the paper we always use the letter $\graphh$ to represent a hypergraph, and $\graphg$ to represent a graph. 

\section{Baseline Algorithms} \label{sec:baselines}
Before we come to the definition of the new algorithm, let us first consider the baseline algorithms for finding bipartite components in hypergraphs.
One natural idea is to
construct a graph $\graphg$ from the original hypergraph $\graphh$ and apply a graph algorithm on $\graphg$ to find a good approximation of the cut structure of $\graphh$.
However, any simple graph reduction would introduce a factor of $r$, which relates to the rank of hyperedges in $\graphh$, into the approximation guarantee.
To see this, consider the following two natural graph reductions:
\begin{enumerate}
\item \firstdef{Random Reduction}: From $\graphh=(\vertexset,\edgeset_\graphh)$, we construct $\graphg=(\vertexset, \edgeset_\graphg)$ in which every hyperedge $e\in  \edgeset_\graphh$ is replaced by an edge connecting two randomly chosen vertices in $e$.
    \item \firstdef{Clique Reduction}: From $\graphh=(\vertexset,\edgeset_\graphh)$, we construct $\graphg=(\vertexset, \edgeset_\graphg)$ in which 
every hyperedge $e\in \edgeset_\graphh$ of rank $\rank(e)$ is replaced by a clique of $\rank(e)$ vertices in $\graphg$;
\end{enumerate}

To discuss the drawback of both reductions, we study the following two $r$-uniform  hypergraphs $\graphh_1$ and $\graphh_2$, in which all the edges connect the disjoint vertex sets  $\setl$ and $\setr$ and are constructed as follows:
\begin{enumerate}
    \item in $\graphh_1$, every edge contains exactly one vertex from $\setl$, and $r-1$ vertices from $\setr$;
    \item in $\graphh_2$, every edge contains exactly $r/2$ vertices from $\setl$ and $r/2$ vertices from $\setr$.
\end{enumerate}
As such, we have that $\weight_{\graphh_1}(\setl,\setr) = \cardinality{\edgeset_{\graphh_1}}$, and $\weight_{\graphh_2}(\setl,\setr) = \cardinality{\edgeset_{\graphh_2}}$.
See  Figure~\ref{fig:example_hypergraphs} for illustration.

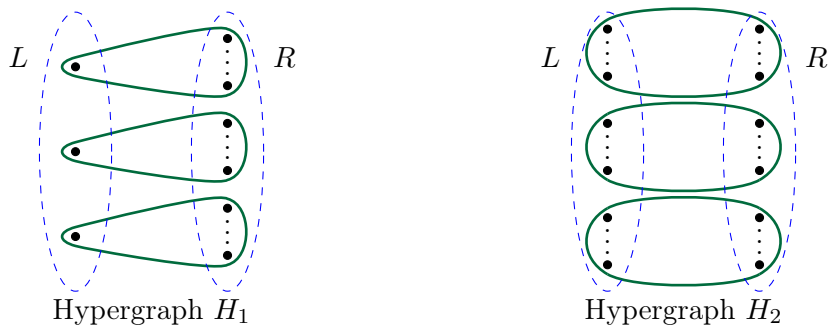
\begin{figure}[h]
    \centering
    \begin{tikzpicture}
	\begin{pgfonlayer}{nodelayer}
		\node [style=smallblack] (4) at (-5, 3) {};
		\node [style=smallblack] (6) at (-5, 1.75) {};
		\node [style=none] (15) at (-5, 1.5) {};
		\node [style=none] (17) at (-5, 3.25) {};
		\node [style=none] (18) at (-9, 2.5) {};
		\node [style=none] (19) at (-9, 2) {};
		\node [style=smallblack] (20) at (-9, 2.25) {};
		\node [style=none] (29) at (-10.5, 2.5) {$\setl$};
		\node [style=none] (30) at (-3.5, 2.5) {$\setr$};
		\node [style=none] (31) at (-7, -4.25) {Hypergraph $\graphh_1$};
		\node [style=smallblack] (34) at (9, 3.25) {};
		\node [style=smallblack] (35) at (5, 2) {};
		\node [style=smallblack] (36) at (9, 2) {};
		\node [style=smallblack] (49) at (5, 3.25) {};
		\node [style=none] (58) at (3.5, 2.5) {$\setl$};
		\node [style=none] (59) at (10.5, 2.5) {$\setr$};
		\node [style=none] (60) at (7, -4.25) {Hypergraph $\graphh_2$};
		\node [style=none] (61) at (9, 1.75) {};
		\node [style=none] (62) at (9, 3.5) {};
		\node [style=none] (63) at (5, 3.5) {};
		\node [style=none] (64) at (5, 1.75) {};
		\node [style=none] (81) at (-5, 2.575) {$\vdots$};
		\node [style=smallblack] (82) at (-5, 0.75) {};
		\node [style=smallblack] (83) at (-5, -0.5) {};
		\node [style=none] (84) at (-5, -0.75) {};
		\node [style=none] (85) at (-5, 1) {};
		\node [style=none] (86) at (-9, 0.25) {};
		\node [style=none] (87) at (-9, -0.25) {};
		\node [style=smallblack] (88) at (-9, 0) {};
		\node [style=none] (89) at (-5, 0.325) {$\vdots$};
		\node [style=smallblack] (90) at (-5, -1.5) {};
		\node [style=smallblack] (91) at (-5, -2.75) {};
		\node [style=none] (92) at (-5, -3) {};
		\node [style=none] (93) at (-5, -1.25) {};
		\node [style=none] (94) at (-9, -2) {};
		\node [style=none] (95) at (-9, -2.5) {};
		\node [style=smallblack] (96) at (-9, -2.25) {};
		\node [style=none] (97) at (-5, -1.925) {$\vdots$};
		\node [style=big vertical ellipse dashed] (100) at (-9, 0) {};
		\node [style=big vertical ellipse dashed] (101) at (-5, 0) {};
		\node [style=none] (102) at (5, 2.825) {$\vdots$};
		\node [style=none] (103) at (9, 2.825) {$\vdots$};
		\node [style=smallblack] (104) at (9, 0.75) {};
		\node [style=smallblack] (105) at (5, -0.5) {};
		\node [style=smallblack] (106) at (9, -0.5) {};
		\node [style=smallblack] (107) at (5, 0.75) {};
		\node [style=none] (108) at (9, -0.75) {};
		\node [style=none] (109) at (9, 1) {};
		\node [style=none] (110) at (5, 1) {};
		\node [style=none] (111) at (5, -0.75) {};
		\node [style=none] (112) at (5, 0.325) {$\vdots$};
		\node [style=none] (113) at (9, 0.325) {$\vdots$};
		\node [style=smallblack] (114) at (9, -1.75) {};
		\node [style=smallblack] (115) at (5, -3) {};
		\node [style=smallblack] (116) at (9, -3) {};
		\node [style=smallblack] (117) at (5, -1.75) {};
		\node [style=none] (118) at (9, -3.25) {};
		\node [style=none] (119) at (9, -1.5) {};
		\node [style=none] (120) at (5, -1.5) {};
		\node [style=none] (121) at (5, -3.25) {};
		\node [style=none] (122) at (5, -2.175) {$\vdots$};
		\node [style=none] (123) at (9, -2.175) {$\vdots$};
		\node [style=big vertical ellipse dashed] (124) at (9, 0) {};
		\node [style=big vertical ellipse dashed] (125) at (5, 0) {};
	\end{pgfonlayer}
	\begin{pgfonlayer}{edgelayer}
		\draw [style=undirected green] (19.center)
			 to [bend left=75, looseness=2.50] (18.center)
			 to [in=150, out=15, looseness=0.25] (17.center)
			 to [bend left=75] (15.center)
			 to [in=-15, out=-150, looseness=0.25] cycle;
		\draw [style=undirected green] (62.center)
			 to [bend left=60, looseness=1.25] (61.center)
			 to [bend left, looseness=0.50] (64.center)
			 to [bend left=60, looseness=1.25] (63.center)
			 to [bend left, looseness=0.50] cycle;
		\draw [style=undirected green] (87.center)
			 to [bend left=75, looseness=2.50] (86.center)
			 to [in=150, out=15, looseness=0.25] (85.center)
			 to [bend left=75] (84.center)
			 to [in=-15, out=-150, looseness=0.25] cycle;
		\draw [style=undirected green] (95.center)
			 to [bend left=75, looseness=2.50] (94.center)
			 to [in=150, out=15, looseness=0.25] (93.center)
			 to [bend left=75] (92.center)
			 to [in=-15, out=-150, looseness=0.25] cycle;
		\draw [style=undirected green, bend left=60, looseness=1.25] (109.center) to (108.center);
		\draw [style=undirected green, bend left, looseness=0.50] (108.center) to (111.center);
		\draw [style=undirected green, bend left=60, looseness=1.25] (111.center) to (110.center);
		\draw [style=undirected green, bend left, looseness=0.50] (110.center) to (109.center);
		\draw [style=undirected green, bend left=60, looseness=1.25] (119.center) to (118.center);
		\draw [style=undirected green, bend left, looseness=0.50] (118.center) to (121.center);
		\draw [style=undirected green, bend left=60, looseness=1.25] (121.center) to (120.center);
		\draw [style=undirected green, bend left, looseness=0.50] (120.center) to (119.center);
	\end{pgfonlayer}
\end{tikzpicture}
    \caption{Example hypergraphs illustrating the problem with simple reductions.
    \label{fig:example_hypergraphs}}
\end{figure}

Now we analyse the size of the cut $(\setl, \setr)$ in the reduced graphs.
\begin{itemize}
    \item Let $\weight_{\graphg}(\setl,\setr)$ be the cut value of $(\setl,\setr)$ in the \emph{random} graph $\graphg$ constructed from $\graphh$ by the random reduction. We have for $\graphh_1$ that $\E[\weight_{\graphg}(\setl,\setr) ] = \Theta(1/r)\cdot \weight_{\graphh_1}(\setl,\setr)$  and for $\graphh_2$ that $\E[\weight_{\graphg}(\setl,\setr) ] = \Theta(1)\cdot \weight_{\graphh_2}(\setl,\setr)$. 
    \item 
Similarly, by setting $\weight_{\graphg}(\setl,\setr)$ to be the cut value of $(\setl,\setr)$ in the reduced graph $\graphg$ constructed from $\graphh$ by the clique reduction, we have for $\graphh_1$ that
$\weight_{\graphg}(\setl,\setr)= \Theta(r)\cdot \weight_{\graphh_1}(\setl,\setr)$
and for $\graphh_2$ that
$\weight_{\graphg}(\setl,\setr) = \Theta(r^2)\cdot \weight_{\graphh_2}(\setl,\setr)$.
\end{itemize}
Since these two approximation ratios differ by a factor of $r$ in the two reductions, no matter how we weight the edges in the reduction, there are always some hypergraphs in which some cuts cannot be approximated better than a factor of $O(r)$. This suggests that   reducing an $r$-uniform hypergraph $\graphh$ to a graph with some simple reduction would always introduce a factor of $r$ into the approximation guarantee.
This is one of the main reasons to develop spectral theory for hypergraphs through heat diffusion processes~\cite{chanSpectralPropertiesHypergraph2018, takaiHypergraphClusteringBased2020, yoshidaCheegerInequalitiesSubmodular2019}.

\section{Diffusion process and the algorithm } \label{sec:diff_and_alg}
    In this section, we introduce a new diffusion process in hypergraphs and use it to design a polynomial-time algorithm  for finding bipartite components in hypergraphs.
    We first study graphs to give some intuition, and then   generalise to hypergraphs and describe our algorithm.
    Finally, we sketch some of the detailed analysis which proves that the diffusion process is well defined.

\subsection{The diffusion process in graphs} \label{sec:2graph-diffusion}
To discuss the intuition behind our designed diffusion process, let us look at the case of graphs. Let $\graphg=(\vertexset, \edgeset , \weight)$ be a graph, and we have for  any $\vecx \in \R^n$ that
\[
\frac{\vecx^\transpose \signlapn_\graphg \vecx}{\vecx^\transpose \vecx} = 
\frac{\vecx^\transpose ( \identity + \degmhalfneg_\graphg \adj_\graphg \degmhalfneg_\graphg) \vecx }{\vecx^\transpose \vecx }.
\]
By setting $\vecx = \degmhalf_\graphg \vecy$, we have that
\begin{equation}\label{eq:2graphcalculation}
\frac{\vecx^\transpose \signlapn_\graphg \vecx}{\vecx^\transpose \vecx}
= \frac{\vecy^\transpose \degmhalf_\graphg \signlapn_\graphg \degmhalf_\graphg \vecy }{ \vecy^\transpose \degm_\graphg \vecy}
= \frac{\vecy^\transpose (\degm_\graphg + \adj_\graphg) \vecy }{\vecy^\transpose \degm_\graphg  \vecy }
= \frac{\sum_{\{u,v\}\in \edgeset_\graphg} \weight(u,v) \cdot (\vecy(u) + \vecy(v))^2}{\sum_{u\in \vertexset_\graphg} \deg_\graphg(u) \cdot \vecy(u)^2}.
\end{equation}
% \[
% \frac{x^{\rot}Jx}{x^{\rot} x} = \frac{y^{\rot} D_G^{1/2} J D_G^{1/2} y }{ y^{\rot} D_G y} = \frac{y^{\rot} (D_G + A_G) y }{y^{\rot} D_G  y } = \frac{\sum_{\{u,v\}\in E_G} (y(u) + y(v))^2}{\sum_{u\in V_G} d_G(u) \cdot y_u^2}.
% \]
   It is easy to see  that   $\lambda_{1}(\signlapn_\graphg)=0$ if $\graphg$ is bipartite,  and it is known that  $\lambda_{1}(\signlapn_\graphg)$ and its corresponding eigenvector $\vecf_1(\signlapn_\graphg)$ are closely related to two densely connected components of $\graphg$~\cite{trevisanMaxCutSmallest2012}.
   Moreover, similar to the heat equation for graph Laplacians $\lap_\graphg$, suppose $\degm_\graphg \vecf_t \in \R^n$ is some measure on the vertices of $\graphg$, then a diffusion process defined by the differential equation
\begin{equation} \label{eq:2graph_pde}
\frac{\mathrm{d} \vecf_t}{\mathrm{d} t} = -\degm_\graphg^{-1} \signlap_\graphg \vecf_t\
\end{equation}
 will converge to the eigenvector of $\degm_\graphg^{-1} \signlap_\graphg$ with minimum eigenvalue and
can be employed to find two densely connected components of the underlying graph.\footnote{For the reader familiar with the heat diffusion process of graphs~(e.g.,  \cite{chungHeatKernelPagerank2007, klosterHeatKernelBased2014}), we remark that the above-defined process essentially employs the operator $\signlap_\graphg$ to
replace the Laplacian $\lap_\graphg$ when defining the heat diffusion: through $\signlap_\graphg$, the heat diffusion can be used to find two densely connected components of $\graphg$.}

\subsection{The hypergraph diffusion and our algorithm} \label{sec:algorithm}
 Now we study  whether one can construct a new hypergraph operator $\signlap_\graphh$ which generalises the diffusion in graphs to hypergraphs.
First of all, we focus on a fixed time $t$ with measure vector $\degm_\graphh \vecf_t \in \R^n$ and ask whether we can follow \eqref{eq:2graph_pde} and define the rate of change
\[
    \frac{\mathrm{d} \vecf_t}{\mathrm{d} t} = - \degm_\graphh^{-1} \signlap_\graphh \vecf_t
\]
so that the diffusion can proceed for an infinitesimal time step.
Our intuition is that the rate of change due to some edge $e \in \edgeset_\graphh$ should involve only the vertices in $e$ with the maximum or minimum value in the normalised measure $\vecf_t$.
To formalise this, for any edge $e \in \edgeset_\graphh$, we define
\[
    \sets_{\vecf}(e) \triangleq \{ v \in e : \vecf_t(v) = \max_{u \in e} \vecf_t(u) \} \mbox{\ \ \ and\ \ \ }
    \set{I}_{\vecf}(e) \triangleq \{ v \in e : \vecf_t(v) = \min_{u \in e} \vecf_t(u) \}.
\]
That is, for any edge $e$ and normalised measure $\vecf_t$, $\sets_{\vecf}(e) \subseteq e$ consists of the vertices $v$ adjacent to $e$ whose $\vecf_t(v)$ values are maximum and $\set{I}_{\vecf}(e) \subseteq e$ consists of the vertices $v$ adjacent to $e$ whose $\vecf_t(v)$ values are minimum.
See Figure~\ref{fig:example_hyperedge} for an example.
Then, applying the $\signlap_\graphh$ operator to a vector $\vecf_t$ should be equivalent to applying the operator $\signlap_\graphg$ for some graph $\graphg$ which we construct by splitting the weight of each hyperedge $e \in \edgeset_\graphh$ between the edges in $\sets_{\vecf}(e) \times \set{I}_{\vecf}(e)$.
Similar to the case for graphs and \eqref{eq:2graphcalculation}, for any $\vecx = \degmhalf_\graphh \vecf_t$ this will give us the quadratic form
\begin{align*}
\frac{\vecx^\transpose \degmhalfneg_\graphh \signlap_\graphh \degmhalfneg_\graphh \vecx}{\vecx^\transpose \vecx}
= \frac{\vecf_t^\transpose \signlap_\graphg \vecf_t }{\vecf_t^\transpose \degm_\graphh  \vecf_t } & 
= \frac{\sum_{\{u,v\}\in \edgeset_\graphg} \weight_\graphg(u, v) \cdot (\vecf_t(u) + \vecf_t(v))^2}{\sum_{u\in \vertexset_\graphg} \deg_\graphh(u) \cdot \vecf_t(u)^2} \\
& = \frac{\sum_{e \in \edgeset_\graphh} \weight_\graphh(e) (\max_{u \in e} \vecf_t(u) + \min_{v \in e} \vecf_t(v))^2}{\sum_{u \in \vertexset_\graphh} \deg_\graphh(u)\cdot \vecf_t(u)^2},
\end{align*}
where   $\weight_\graphg(u, v)$ is the weight of the edge $\{u, v\}$ in $\graphg$, and $\weight_\graphh(e)$ is the weight of the edge $e$ in $\graphh$.
We will show in the proof of Theorem~\ref{thm:mainalg} that $\signlap_\graphh$ has an eigenvalue of $0$ if the hypergraph is $2$-colourable\footnote{Hypergraph $\graphh$ is $2$-colourable if there are disjoint sets $\setl, \setr \subset \vertexset_\graphh$ such that every edge intersects $\setl$ and $\setr$.}, and that the spectrum of $\signlap_\graphh$ is closely related to the hypergraph bipartiteness.

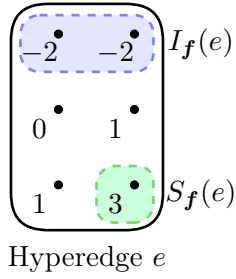
\begin{figure} 
    \centering
    \begin{tikzpicture}
	\begin{pgfonlayer}{nodelayer}
		\node [style=smallblack] (0) at (-1, 2) {};
		\node [style=smallblack] (1) at (1, 2) {};
		\node [style=smallblack] (2) at (-1, 0) {};
		\node [style=smallblack] (3) at (1, 0) {};
		\node [style=smallblack] (4) at (-1, -2) {};
		\node [style=smallblack] (5) at (1, -2) {};
		\node [style=none] (8) at (-2.25, -2.2) {};
		\node [style=none] (9) at (-1.25, -3.2) {};
		\node [style=none] (10) at (0.75, -3.2) {};
		\node [style=none] (11) at (1.75, -2.2) {};
		\node [style=none] (12) at (1.75, 1.8) {};
		\node [style=none] (13) at (0.75, 2.8) {};
		\node [style=none] (14) at (-1.25, 2.8) {};
		\node [style=none] (15) at (-2.25, 1.8) {};
		\node [style=none] (16) at (-1.5, 1.5) {$-2$};
		\node [style=none] (17) at (0.5, 1.5) {$-2$};
		\node [style=none] (18) at (-1.5, -0.5) {$0$};
		\node [style=none] (19) at (0.5, -0.5) {$1$};
		\node [style=none] (20) at (0.5, -2.5) {$3$};
		\node [style=none] (21) at (-1.5, -2.5) {$1$};
		\node [style=none] (22) at (-2, 2) {};
		\node [style=none] (23) at (-2, 1.5) {};
		\node [style=none] (24) at (-1.5, 1) {};
		\node [style=none] (25) at (1, 1) {};
		\node [style=none] (26) at (1.5, 1.5) {};
		\node [style=none] (27) at (1.5, 2) {};
		\node [style=none] (28) at (1, 2.5) {};
		\node [style=none] (29) at (-1.5, 2.5) {};
		\node [style=none] (30) at (1, -1.5) {};
		\node [style=none] (31) at (1.5, -2) {};
		\node [style=none] (32) at (1.5, -2.5) {};
		\node [style=none] (33) at (1, -3) {};
		\node [style=none] (34) at (0.5, -3) {};
		\node [style=none] (35) at (0, -2.5) {};
		\node [style=none] (36) at (0, -2) {};
		\node [style=none] (37) at (0.5, -1.5) {};
		\node [style=none] (38) at (-0.25, -3.95) {Hyperedge $e$};
		\node [style=none] (39) at (2.75, 1.75) {$\ife$};
		\node [style=none] (40) at (2.75, -2.25) {$\sfe$};
	\end{pgfonlayer}
	\begin{pgfonlayer}{edgelayer}
		\draw [style=undirected] (9.center)
			 to (10.center)
			 to [bend left=315, looseness=1.25] (11.center)
			 to (12.center)
			 to [bend right=45, looseness=1.25] (13.center)
			 to (14.center)
			 to [bend right=45, looseness=1.25] (15.center)
			 to (8.center)
			 to [bend right=45, looseness=1.25] cycle;
		\draw [style=blue dashed, fill={blue!20}, opacity=0.5] (29.center)
			 to (28.center)
			 to [bend left=45] (27.center)
			 to (26.center)
			 to [bend left=45] (25.center)
			 to (24.center)
			 to [bend right=315] (23.center)
			 to (22.center)
			 to [bend left=45] cycle;
		\draw [style=green dashed, fill={rgb,255: red,149; green,255; blue,179}, opacity=0.5] (33.center)
			 to (34.center)
			 to [bend right=315, looseness=1.25] (35.center)
			 to (36.center)
			 to [bend left=45] (37.center)
			 to (30.center)
			 to [bend left=45, looseness=1.25] (31.center)
			 to (32.center)
			 to [bend left=45] cycle;
	\end{pgfonlayer}
\end{tikzpicture}
    \caption{\small{Illustration of 
      $\sfe$ and $\ife$. Vertices are labelled with their value in $\vecf_t$.
    \label{fig:example_hyperedge}}}
\end{figure}
 For this reason, we would expect that the diffusion process based on the operator $\signlap_\graphh$ can be used to find sets with small hypergraph bipartiteness.
However, one needs to be very  cautious here as,   by the nature of the diffusion process,  the values $\vecf_t(v)$ of all the vertices $v$  change over time and, as a result, the sets $\sets_{\vecf}(e)$ and $\set{I}_{\vecf}(e)$ that consist of the vertices with the maximum and minimum $\vecf_t$-value
might change after an \emph{infinitesimal} time step; this will prevent the process from continuing. We will   discuss this  issue in detail through the so-called \firstdef{Diffusion Continuity 
Condition} in   Section~\ref{sec:natural_assumption}.
In essence,  the diffusion continuity condition ensures that one  can always    construct a graph $\graphg$ by allocating the weight of each hyperedge $e$ to the edges in $\sets_{\vecf}(e) \times \set{I}_{\vecf}(e)$ such that the sets $\sets_{\vecf}(e)$ and $\set{I}_{\vecf}(e)$ will not change in infinitesimal time although $\vecf_t$ changes according to $(\mathrm{d} \vecf_t)/(\mathrm{d} t) = - \degm_\graphh^{-1} \signlap_\graphg \vecf_t$.
We will also  present an efficient procedure in  Section~\ref{sec:natural_assumption}  to compute the weights of edges in $\sfe \times \ife$. All of these guarantee that (i) every graph that corresponds to the hypergraph diffusion process at any time step   can be efficiently constructed;
(ii) with this sequence of constructed graphs, the diffusion process defined by $\signlaph$ is able to continue until the heat distribution converges.
With this,  we summarise the main idea of our presented algorithm   as follows: 
\begin{itemize}\itemsep -2pt
    \item First of all, we introduce some arbitrary $\vecf_0 \in \R^n$ as the initial diffusion vector, and a step size parameter $\epsilon>0$ to discretise
    the diffusion process. At each step, the algorithm constructs the graph $\graphg$ guaranteed by the diffusion continuity condition,  and updates  $\vecf_t \in \R^n$ according to the rate of change $(\mathrm{d} \vecf_t)/(\mathrm{d}t) = - \degm_\graphh^{-1} \signlapg \vecf_t$. The algorithm terminates when  
     $\vecf_t$ has converged, i.e., the ratio between the current Rayleigh quotient $(\vecf_t^\transpose \signlapg \vecf_t)/(\vecf_t^\transpose \degm_\graphh \vecf_t)$
     and the one in the previous time step is bounded by some predefined constant.
    \item Secondly, similar to many previous spectral graph clustering algorithms~(e.g. \cite{andersenLocalGraphPartitioning2006, takaiHypergraphClusteringBased2020, trevisanMaxCutSmallest2012}),
    the algorithm constructs  the sweep sets defined by $\vecf_t$ and returns the two sets with minimum
    $\bipart_\graphh$-value among all the constructed sweep sets. 
    Specifically, for every $1\leq i\leq n$, the algorithm constructs    $\setl_j = \{v_i : \abs{\vecf_t(v_i)} \geq \abs{\vecf_t(v_j)} \land \vecf_t(v_i) < 0\}$ and $\setr_j = \{v_i : \abs{\vecf_t(v_i)} \geq \abs{\vecf_t(v_j)} \land \vecf_t(v_i) \geq 0\}$.
    Then, between the $n$ pairs $(\setl_j, \setr_j)$, the algorithm returns the one with the minimum $\bipart_\graphh$-value. 
\end{itemize}
  See Algorithm~\ref{algo:main} for the formal description,
and its performance is summarised in  Theorem~\ref{thm:mainalg}.  

\begin{theorem}[Main Result] \label{thm:mainalg}
Given a hypergraph $\graphh = (\vertexset_\graphh, \edgeset_\graphh, \weight)$ and parameter $\epsilon>0$, the following holds:
\begin{enumerate}\itemsep -2pt
    \item There is an algorithm that finds an eigen-pair ($\lambda, \vecf$) of the operator $\signlaph$ such that $\lambda \leq \lambda_1(\signlapg)$, where $\graphg$ is the clique reduction of $\graphh$ and the inequality is strict if $\min_{e \in \edgeset_\graphh} \rank(e) > 2$. The algorithm runs in $\poly(\cardinality{\vertexset_\graphh}, \cardinality{\edgeset_\graphh} ,1/\epsilon)$ time.
    \item Given an eigen-pair $(\lambda, \vecf)$ of the operator $\signlaph$, there is an algorithm that constructs the two-sided sweep sets  defined  on $\vecf$, and  finds sets $\setl$ and $\setr$ such that $\bipart_\graphh(\setl, \setr) \leq \sqrt{2 \lambda}$. The algorithm runs in $\poly(\cardinality{\vertexset_\graphh}, \cardinality{\edgeset_\graphh})$ time.\\
\end{enumerate}
\end{theorem}

\begin{algorithm} \SetAlgoLined
\SetKwInOut{Input}{Input}
\SetKwInOut{Output}{Output}
\Input{Hypergraph $\graphh$, starting vector $\vecf_0\in \R^n$, step size $\epsilon>0$}
\Output{Sets $\setl$ and $\setr$} 
$t := 0$ \\
\While{$\vecf_t$ has not converged}{
    Use $\vecf_t$ to construct graph $\graphg$ satisfying the diffusion continuity condition
     \\
    $\vecf_{t + \epsilon} := \vecf_t - \epsilon \degm_\graphh^{-1} \signlapg \vecf_t $\\
    $t := t + \epsilon$ 
} 
Set $j := \argmin_{1 \leq i \leq n}\bipart_\graphh(\setl_i, \setr_i)$ \\
\Return $(\setl_j, \setr_j)$
\caption{\diffalgname}\label{algo:main}
\end{algorithm}

\subsubsection{Computing the minimum eigenvector of \texorpdfstring{$\signlaph$}{J\_H} is \texorpdfstring{$\NP$}{NP}-hard} \label{sec:hyp-np-hard}
The second statement of Theorem~\ref{thm:mainalg} answers an open question posed by Yoshida~\cite{yoshidaCheegerInequalitiesSubmodular2019} by showing that our constructed hypergraph operator satisfies a Cheeger-type inequality for hypergraph bipartiteness.
However, the following theorem shows that there is no polynomial-time algorithm to compute the minimum eigenvector of \emph{any such operator} unless $\P=\NP$.
This means that the problem we consider in this work is fundamentally more difficult than the equivalent problem for graphs, as well as the problem of finding a sparse cut in a hypergraph.

\begin{theorem} \label{thm:min_eigvec_np}
     Given any operator that satisfies a Cheeger-type inequality for hypergraph bipartiteness, there's no polynomial-time algorithm that computes a multiplicative-factor approximation of the minimum eigenvalue or eigenvector unless $\P=\NP$.
\end{theorem}
\begin{proof}
Our proof is based on considering the following  \textsc{Max Set Splitting} problem:  For any given hypergraph, the problem looks for a partition $\setl,\setr$ with $\setl \union \setr=\vertexset_\graphh$ and $\setl \intersect \setr=\emptyset$, such that it holds for any $e\in \edgeset_\graphh$ that  $e\cap \setl \neq\emptyset$ and $e\cap \setr \neq\emptyset$. This problem is known to be $\NP$-complete~\cite{gareyComputersIntractability1979}.
This is also referred to as \textsc{Hypergraph 2-Colorability} and we can consider the problem of coloring every vertex in the hypergraph either red or blue such that every edge contains at least one vertex of each color.

We will assume that there is some operator $\signlap$ which satisfies the Cheeger-type inequality
\begin{equation} \label{eq:cheeg_ineq_proof}
    \frac{\lambda_1(\signlap)}{2} \leq \bipart_H \leq \sqrt{2 \lambda_1(\signlap)}
\end{equation}
and that there is an algorithm which can compute the minimum eigen-pair of the operator in polynomial time.
We will show that this would allow us to solve the \textsc{Max Set Splitting} problem in polynomial time.

Given a 2-colorable hypergraph $\graphh$ with coloring $(\setl, \setr)$, we will use the eigenvector of the operator $\signlaph$ to find a valid coloring.
By definition, we have that $\bipart(\setl, \setr) = 0$ and $\gamma_1 = 0$ by \eqref{eq:cheeg_ineq_proof}.
Furthermore, by the second statement of Theorem~\ref{thm:mainalg}, we can compute disjoint sets $\setl', \setr'$ such that $\bipart(\setl', \setr') = 0$. Note that in general $\setl'$ and $\setr'$ will be different from $\setl$ and $\setr$.

Now, let $\sete' = \{e \in \edgeset_\graphh : e \intersect (\setl' \union \setr') \neq \emptyset\}$.
Then, by the definition of bipartiteness, for all $e \in \sete'$ we have $e \intersect \setl' \neq \emptyset$ and $e \intersect \setr' \neq \emptyset$.
Therefore, if we color every vertex in $\setl'$ blue and every vertex in $\setr'$ red then every $e \in \sete'$ will be satisfied. That is, they will contain at least one vertex of each color.

It remains to color the vertices in $\vertexset \setminus (\setl' \union \setr')$ such that the edges in $\edgeset_\graphh \setminus \sete'$ are satisfied.
The hypergraph $\graphh' = (\vertexset \setminus (\setl' \union \setr'), \edgeset_\graphh \setminus \sete')$ must also be 2-colorable and so we can recursively apply our algorithm until every vertex is coloured.
This algorithm will run in polynomial time since there are at most $\bigo{n}$ iterations and each iteration runs in polynomial time by our assumption.
\end{proof}

\subsection{Dealing with the  diffusion continuity condition} \label{sec:natural_assumption}
It remains for us to discuss the diffusion continuity condition, which guarantees that $\sfe$ and $\ife$ 
will not change in infinitesimal time and the diffusion process
will eventually converge to some stable distribution.
Formally, let $\vecf_t$ be the normalised measure on the vertices of $\graphh$,
and let
\begin{equation*}
\vecr \triangleq \dfdt = -\degm_\graphh^{-1} \signlaph \vecf_t
\end{equation*}
be the derivative of $\vecf_t$, which describes the rate of change   for every vertex at the current time $t$.
We write $\vecr(v)$ for any $v\in \vertexset_\graphh$ as 
$
\vecr(v) = \sum_{e\in \edgeset_\graphh} \vecr_e(v)$,  
where $\vecr_e(v)$ is the contribution of edge $e$ towards the rate of change of $\vecf_t(v)$.
Now we discuss three rules that we expect the diffusion process to satisfy, and later prove that these three rules uniquely define the rate of change $\vecr$.

First of all, as we mentioned in Section~\ref{sec:algorithm},
we expect that only the vertices in $\sfe \cup \ife$ will participate in the diffusion process,
i.e., $\vecr_e(u)=0$ unless $u\in \sfe \cup \ife$. Moreover, any vertex $u$  participating in the diffusion process must satisfy the following: 
\begin{center}
\begin{minipage}{0.99\textwidth}
\centering
\begin{itemize}\itemsep -2pt
\item Rule~(0a): if $\abs{\vecr_e(u)}>0$ and $u\in \sfe$, then  $\vecr(u) = \max_{v \in \sfe} \{\vecr(v)\}$. % $r(u)\geq r(v)$ for any $v\in S(e)$.
\item Rule~(0b): if $\abs{\vecr_e(u)}>0$ and $u\in \ife$, then $\vecr(u) = \min_{v \in \ife} \{\vecr(v)\}$. % $r(u)\leq r(v)$ for any $v\in I(e)$.
\end{itemize}
\end{minipage}
\end{center}
To explain Rule~(0), notice that for an infinitesimal time, $\vecf_t(u)$ will be increased  according to $\left(\mathrm{d}\vecf_t/\mathrm{d}t\right)(u) = \vecr(u)$. Hence, by Rule~(0) we know that, if $u\in \sfe$ (resp.\ $u \in \ife$) participates in the diffusion process in edge $e$, then in an infinitesimal time $\vecf(u)$ will
remain the maximum (resp.\ minimum) among the vertices in $e$.
Such a rule is necessary to ensure that the vertices involved in the diffusion in edge $e$ do not change in infinitesimal time, and the diffusion process is able to continue.

Our next rule states that the total rate of change of the measure due to edge $e$ is equal to $-\weight(e) \cdot \discrep_{\vecf}(e)$:
\begin{center}
\begin{minipage}{0.99\textwidth}
\centering
\begin{itemize}
\item Rule~(1): $\sum_{v\in \sfe} \deg(v) \vecr_e(v)  = \sum_{v\in \ife} \deg(v) \vecr_e(v) = -\weight(e)\cdot \discrep_{\vecf}(e)$ for all $e \in \edgeset_\graphh$.
\end{itemize}
\end{minipage}
\end{center}
 
This rule is a generalisation from the operator $\signlapg$ in graphs.
In particular, since $\degm_\graphg^{-1} \signlapg \vecf_t(u) =   \sum_{\{u, v\}\in \edgeset_\graphg} \weight_\graphg(u, v) (\vecf_t(u) + \vecf_t(v))/\deg_\graphg(u)$, the rate of change of $\vecf_t(u)$ due to the edge $\{u, v\} \in \edgeset_\graphg$ is $- \weight_\graphg(u, v) (\vecf_t(u) + \vecf_t(v))/\deg_\graphg(u)$.
Rule (1) states that in the hypergraph case the rate of change of the vertices in $\sfe$ and $\ife$ together behave like the rate of change of $u$ and $v$ in the graph case.
 
One might have expected that Rules~(0) and (1) together would define a unique process. Unfortunately, this isn't the case.
For example, let us define an unweighted hypergraph $\graphh= (\vertexset_\graphh, \edgeset_\graphh)$, where $\vertexset_\graphh=\{u,v,w\}$ and $\edgeset_\graphh = \{ \{ u,v,w\} \}$. By setting the measure to be $\vecf_t=(1,1,-2)^\transpose$ and $e=\{u,v,w\}$, we have $\discrep_{\vecf_t}(e)=-1$, and $\vecr(u) + \vecr(v) = \vecr(w)= 1$ by Rule (1). In such a scenario, either $\{u,w\}$ or $\{v,w\}$ can participate in the diffusion and satisfy Rule (0), which makes the process not uniquely defined.
To overcome this, we introduce the following stronger rule to replace Rule~(0): 
\begin{center}
\begin{minipage}{0.99\textwidth}
\centering
\begin{itemize}\itemsep -2pt
\item Rule~(2a): Assume that $\abs{\vecr_e(u)}>0$ and $u\in \sfe$. 
\begin{itemize}
    \item If $\discrep_{\vecf}(e)>0$, then $\vecr(u) = \max_{v\in \sfe) }\{\vecr(v) \}$;
    \item If $\discrep_{\vecf}(e)<0$, then $\vecr(u)= \vecr(v)$ for all $v\in \sfe$.
\end{itemize}
\item Rule~(2b): Assume that $\abs{\vecr_e(u)}>0$ and $u\in \ife$:
\begin{itemize}
    \item If $\discrep_{\vecf}(e)<0$, then $\vecr(u) = \min_{v\in \ife)} \{\vecr(v)\}$;
    \item If $\discrep_{\vecf}(e)>0$, then $\vecr(u)= \vecr(v)$ for all $v\in \ife$.
\end{itemize}
\end{itemize}
\end{minipage}
\end{center}
Notice that the first conditions of Rules~(2a) and (2b) correspond to Rules (0a) and (0b) respectively; the second conditions are introduced for purely technical reasons: they state that, if the discrepancy of $e$ is negative  (resp.\ positive), then all the vertices $u\in \sfe$ (resp.\ $u \in \ife$) will have the same value of $\vecr(u)$.
Lemma~\ref{lem:rulesimplydiffusion} shows that there is a unique $\vecr\in\R^n$ that satisfies Rules~(1) and (2), and $\vecr$ can be computed in polynomial time. 
Therefore, our two rules uniquely define a diffusion process, and we can use the computed $\vecr$ to simulate the continuous diffusion process with a discretised version.\footnote{Note that the graph $\graphg$ used for the diffusion at time $t$ can be easily computed from the $\{\vecr_e(v)\}$ values, although in practice this is not actually needed since the $\vecr(u)$ values can be used to update the diffusion directly.}

\begin{lemma}
     \label{lem:rulesimplydiffusion}
    For any given $\vecf_t\in\R^n$, there is a unique $\vecr=\mathrm{d}\vecf_t/\mathrm{d}t$ and associated $\{\vecr_e(v)\}_{e\in \edgeset, v\in \vertexset}$ that satisfy  Rule~(1) and (2), and  $\vecr$ can be computed in polynomial time by linear programming. 
\end{lemma}
 
\begin{remark}
The rules we define and the proof of Lemma~\ref{lem:rulesimplydiffusion} are more involved than those used in~\cite{chanSpectralPropertiesHypergraph2018} to define the hypergraph Laplacian operator.
In particular, in contrast to~\cite{chanSpectralPropertiesHypergraph2018}, in our case the discrepancy $\discrep_{\vecf}(e)$ within a hyperedge $e$ can be either positive or negative.
This results in the four different cases in Rule~(2) which must be carefully considered throughout the proof of Lemma~\ref{lem:rulesimplydiffusion}.
\end{remark}

\subsection{Proof of Lemma~\ref{lem:rulesimplydiffusion}} \label{sec:thm2proof}
In this section we prove Lemma~\ref{lem:rulesimplydiffusion}.
First, we will construct a linear program which can compute the rate of change $\vecr$ satisfying the rules of the diffusion process.
We then give a complete analysis of the new linear program which establishes Lemma~\ref{lem:rulesimplydiffusion}.

\subsubsection{Computing \texorpdfstring{$\vecr$}{r} by a linear program} Now we present an algorithm that computes the vector $\vecr=\mathrm{d}\vecf_t/\mathrm{d}t$ for any $\vecf_t\in\R^n$. 
Without loss of generality, let us fix $\vecf_t\in\R^n$, and define a set of equivalence classes $\mathcal{U}$ on $\vertexset$ such that vertices $u,v\in \vertexset_\graphh$ are in the same equivalence class if $\vecf_t(u)=\vecf_t(v)$. 
Next we study every equivalence class $\setu \in\mathcal{U}$ in turn, and will set the $\vecr$-value of the vertices in $\setu$ recursively.
 In each iteration, we fix the $\vecr$-value of some subset $\setp \subseteq \setu$ and recurse on $\setu \setminus \setp$.
As we'll prove later, it's important to highlight that the recursive procedure ensures that the $\vecr$-values assigned to the vertices always \emph{decrease} after each recursion.
Notice that it suffices to consider the edges $e$ in which $(\sfte \cup \ifte)\cap \setu \neq\emptyset$, since the diffusion process induced by other edges $e$ will have no impact on   $\vecr(u)$ for any $u\in \setu$. Hence, we introduce the sets
\[
\mathcal{S}_{\setu} \triangleq \{e\in \edgeset_\setu: \sfte \intersect \setu \neq\emptyset \}, \qquad \mathcal{I}_{\setu} \triangleq \{e\in \edgeset_\setu: \ifte \intersect \setu \neq\emptyset\},
\]
where $\edgeset_\setu$ consists of the edges adjacent to some vertex in $\setu$. To work with the four cases listed in Rule~(2a) and (2b), we define
\begin{align}
\mathcal{S}_{\setu}^+& \triangleq \{ e\in \cals_\setu: \discrep_{\vecf_t}(e)<0\},\nonumber\\
\cals_{\setu}^- & \triangleq \{ e\in\cals_\setu: \discrep_{\vecf_t}(e)>0\},\nonumber\\
\cali_\setu^+ & \triangleq \{ e\in \cali_\setu: \discrep_{\vecf_t}(e)<0\},\nonumber\\
\cali_\setu^- & \triangleq \{ e\in \cali_\setu: \discrep_{\vecf_t}(e)>0\}.\nonumber
\end{align}
Our objective is to find some  $\setp \subseteq \setu$ and assign the same $\vecr$-value to every vertex in $\setp$.
To this end, for any $\setp \subseteq \setu$ we define
\newcommand{\spp}[2]{\cals_{\set{#1}, \set{#2}}^+}
\newcommand{\spm}[2]{\cals_{\set{#1}, \set{#2}}^-}
\newcommand{\ipp}[2]{\cali_{\set{#1}, \set{#2}}^+}
\newcommand{\ipm}[2]{\cali_{\set{#1}, \set{#2}}^-}
\begin{align*}
\spp{U}{P} & \triangleq \left\{e\in\cals_{\setu}^+: \sfte \subseteq \setp \right\},\\
\ipp{U}{P} & \triangleq 
\left\{e\in\cali_{\setu}^+: \ifte \subseteq \setp \right\}, \\
\spm{U}{P} & \triangleq \left\{ e\in \cals_{\setu}^-: \sfte \cap \setp \neq\emptyset \right\},\\
\ipm{U}{P} &\triangleq \{e \in \cali_{\setu}^-: \ifte \subseteq \setp \}.
\end{align*}
These are the edges which will contribute to the rate of change of the vertices in $\setp$.
Before continuing the analysis, we briefly explain the intuition behind these four definitions:
\begin{enumerate}
    \item For $\spp{U}{P}$, since every $e\in \cals_{\setu}^+$ satisfies $\discrep_{\vecf_t}(e)<0$ and all the vertices in $\sfte$ must have the same value by Rule~(2a), all such $e\in \spp{U}{P}$ must satisfy that $\sfte \subseteq \setp$, since the unassigned vertices will receive lower values of $\vecr$ in the remaining part of the recursion process.
    \item For $\ipp{U}{P}$, since every $e\in\cali_{\setu}^+$ satisfies $\discrep_{\vecf_t}(e)<0$, Rule (2b) implies that if $\vecr_e(u) \neq 0$ then $\vecr(u)\leq \vecr(v)$ for all $v\in \ifte$. Since unassigned vertices will receive lower values of $\vecr$ later, such $e\in\ipp{U}{P}$ must satisfy $\ifte \subseteq \setp$.
    \item For $\spm{U}{P}$,  since every $e\in\cals_{\setu}^-$ satisfies $\discrep_{\vecf_t}(e)>0$, by Rule~(2a) it suffices that some vertex in $\sfte$ receives the assignment in the current iteration, i.e., every such $e$ must satisfy $\sfte \cap \setp \neq\emptyset$.
    \item The case for $\ipm{U}{P}$ is the same as $\spp{U}{P}$. 
\end{enumerate}

As we expect all the vertices $u\in \sfte$ to have the same $\vecr$-value for every $e$ as long as $\discrep_{\vecf_t}(e)<0$ by Rule~(2a) and at the moment we are only considering the assignment of the vertices in $\setp$,  we expect that 
\begin{equation}\label{eq:condition1}
\left\{ e\in \cals_{\setu}^+\setminus \spp{U}{P}: \sfte \cap \setp \neq \emptyset \right\} =\emptyset,
\end{equation}
and this will ensure that, as long as $\discrep_{\vecf_t}(e)<0$ and some $u\in \sfte$ gets its $\vecr$-value, then all the other vertices in $\sfte$ would be assigned the same value as $u$.
Similarly, by Rule~(2b), we expect all the vertices $u\in \ifte$ to have the same $\vecr$-value for every $e$ as long as $\discrep_{\vecf_t}(e)>0$, and so we expect that
\begin{equation}\label{eq:condition2}
\left\{ 
e\in \cali_{\setu}^- \setminus \ipm{U}{P}: \ifte \cap \setp \neq\emptyset
\right\} =\emptyset.
\end{equation}
We will set the $\vecr$-value by dividing the total discrepancy of the edges in $\ipp{U}{P} \union \spp{U}{P} \union \ipm{U}{P} \union \spm{U}{P}$ between the vertices in $\setp$. As such, 
we would like to find some $\setp \subseteq \setu$ that maximises the value of 
\[
\frac{1}{\vol(\setp)}\cdot \left(\sum_{e\in \spp{U}{P}\cup\ipp{U}{P}} c_{\vecf_t}(e)  - \sum_{e\in \spm{U}{P} \cup \ipm{U}{P}} c_{\vecf_t}(e)\right).
\]

Taking all of these requirements into account, we will show that, for any equivalence class $\setu$, we can find the desired set $\setp$ by solving the following linear program:
\begin{alignat}{2}
   & \text{maximise } & & c(\vecx) = \sum_{e \in \cals_{\setu}^+\cup \cali_{\setu}^+ } c_{\vecf_t}(e)\cdot  x_e  - \sum_{e \in \cals_{\setu}^-\cup\cali_{\setu}^-} c_{\vecf_t}(e)\cdot x_e \label{eq:lp} \\
   & \text{subject to }& \quad & \sum_{v \in \setu}
   \begin{aligned}[t]
                \deg(v) y_v & = 1 \\[3ex]
                x_e & = y_u & \quad e & \in \cals_{\setu}^+, u \in \sfte, \\
                x_e & \leq y_u & \quad e & \in \cali_{\setu}^+, u \in \ifte, \\
                x_e & \geq y_u & \quad e & \in \cals_{\setu}^-, u \in \sfte, \\
                x_e & = y_u & \quad e & \in \cali_{\setu}^-, u \in \ifte, \\
                x_e, y_v & \geq 0 & & \forall  v\in \setu,  e\in \edgeset_{\setu}.
   \end{aligned}\nonumber
\end{alignat}
Since the linear program only gives partial assignment to the vertices' $\vecr$-values, we solve the same linear program on the reduced instance given by the set $\setu \setminus \setp$. The formal description of our algorithm is given in Algorithm~\ref{algo:computechangerate}.
% \textcolor{orange}{to define $c(P)$ for any $P$}
\begin{algorithm} \SetAlgoLined
\SetKwInOut{Input}{Input}
\SetKwInOut{Output}{Output}
\Input{vertex set $\setu \subseteq \vertexset$, and edge set $\edgeset_{\setu}$ }
\Output{Values of $\{\vecr(v)\}_{v\in \setu}$}
Construct sets $\cals_{\setu}^+$, $\cals_{\setu}^-$, $\cali_{\setu}^+$, and $\cali_{\setu}^-$\\
Solve the linear program defined by \eqref{eq:lp}, and define $\setp :=\{v\in \setu: y(v) >0\}$\\
Construct sets $\spp{U}{P}$, $\spm{U}{P}$, $\ipp{U}{P}$, and $\ipm{U}{P}$\\
Set $C(\setp) := c_{\vecf_t}\left( \spp{U}{P} \right) + c_{\vecf_t} \left(\ipp{U}{P} \right) - c_{\vecf_t} \left(\spm{U}{P}  \right) - c_{\vecf_t} \left( \ipm{U}{P} \right)$\\
Set $\delta(\setp): = C(\setp) / \vol(\setp)$\\
Set $\vecr(u) :=\delta(\setp)$ for every $u\in \setp$\\
 \textsc{ComputeChangeRate}$\left(\setu \setminus \setp, \edgeset_\setu \setminus \left(\spp{U}{P} \cup \ipp{U}{P} \cup \spm{U}{P} \cup \ipm{U}{P} \right)\right)$
 \caption{\textsc{ComputeChangeRate}$(\setu , \edgeset_\setu)$}\label{algo:computechangerate}
\end{algorithm}

\subsubsection{Analysis of the linear program}
Now we analyse Algorithm~\ref{algo:computechangerate}, and the properties of the $\vecr$-values it computes.
 Specifically, we will show the following facts which will together allow us to establish Lemma~\ref{lem:rulesimplydiffusion}.
\begin{enumerate}
    \item Algorithm~\ref{algo:computechangerate} always produces a unique vector $\vecr$, no matter which optimal result is returned when computing the linear program~\eqref{eq:lp}.
    \item If there is any vector $\vecr$ which is consistent with Rules (1) and (2), then it must be equal to the output of Algorithm~\ref{algo:computechangerate}.
    \item The vector $\vecr$ produced by Algorithm~\ref{algo:computechangerate} is consistent with Rules (1) and (2).
\end{enumerate}
 
\paragraph{The output of Algorithm~\ref{algo:computechangerate} is unique.}
First of all, for any $\setp \subseteq \setu$ that satisfies \eqref{eq:condition1} and \eqref{eq:condition2}, we define vectors $\vecx_\setp$ and $\vecy_\setp$ by
\[
    \vecx_\setp(e) = \twopartdefow{\frac{1}{\vol(\setp)}}{e \in \spp{U}{P}  \union \ipp{U}{P} \union \spm{U}{P} \union \ipm{U}{P}}{0},
\]
\[
    \vecy_\setp(v) = \twopartdefow{\frac{1}{\vol(\setp)}}{v \in \setp}{0},
\]
and $z_\setp=\left(\vecx_\setp, \vecy_\setp \right)$.
It is easy to verify that $\left(\vecx_\setp, \vecy_\setp \right)$ is a feasible solution to \eqref{eq:lp} with the objective value $c\left(\vecx_\setp \right)=\delta(\setp)$. We will prove that \textsc{ComputeChangeRate}~(Algorithm~\ref{algo:computechangerate}) computes a unique vector $\vecr$ regardless of how ties are broken when computing the subsets $\setp$.

For any feasible solution $z=(\vecx,\vecy)$, we say that a non-empty set $\set{Q}$ is a \emph{level set} of $z$ if there is some $t>0$ such that $\set{Q}=\left\{u\in \setu: y_u\geq t\right\}$. We'll first show that any non-empty level set of an optimal solution $z$ also corresponds to an optimal solution.

\begin{lemma}
    Suppose that $z^{\star}= \left(\vecx^{\star}, \vecy^{\star} \right)$ is an optimal solution of the linear program \eqref{eq:lp}. Then, any non-empty level set $\set{Q}$ of $z^{\star}$ corresponds to an optimal solution of \eqref{eq:lp} as well.
\end{lemma}
\begin{proof}
Let $\setp = \{v \in \vertexset_\graphh : \vecy^\star(v) > 0 \}$.
The proof is by case distinction. We first look at the case in which all the vertices in $\setp$ have the same value of  $\vecy^{\star}(v)$ for any $v\in \setp$.
Then, it must be that $z^\star = z_\setp$ and
every non-empty level set $\set{Q}$ of $z^{\star}$ equals to $\setp$, and so the statement holds trivially.

Secondly, we assume that the vertices in $\setp$ have at least two different $\vecy^{\star}$-values. We define $\alpha=\min\left\{ y_v^{\star}: v\in \setp \right\}$, and have
\[
\alpha\cdot\vol(\setp) < \sum_{u\in \setu} \deg(u)\cdot y_u^{\star}=1.
\]
We introduce $\widehat{z}=\left( \widehat{\vecx}, \widehat{\vecy} \right)$ defined by
 \begin{align*}
        \widehat{\vecx}(e) & = \twopartdefow{\frac{\vecx^{\star}(e) - \alpha}{1 - \alpha \vol(\setp)}}{\vecx^{\star}(e) \geq 0}{0}, 
\end{align*}
and
\begin{align*}
        \widehat{\vecy}(v) & = \twopartdefow{\frac{\vecy^{\star}(v) - \alpha}{1 - \alpha \vol(\setp)}}{v \in \setp}{0},
    \end{align*}
which implies that 
\begin{equation}\label{eq:linearcombination1}
\vecx^{\star} = \left( 1-\alpha\vol(\setp) \right)\widehat{\vecx} + \alpha\cdot \constvec_\setp = \left( 1-\alpha\vol(\setp) \right)\widehat{\vecx} + \alpha\cdot \vol(\setp)\cdot \vecx_\setp,
\end{equation}
and
\begin{equation} \label{eq:linearcombination2}
\vecy^{\star} = \left( 1-\alpha\vol(\setp) \right)\widehat{\vecy} + \alpha\cdot \constvec_\setp = \left( 1-\alpha\vol(\setp) \right)\widehat{\vecy} + \alpha\cdot \vol(\setp)\cdot \vecy_\setp,
\end{equation}
where $\constvec_\setp$ is the indicator vector of the set $\setp$.
Notice that 
 $\widehat{z}$ preserves the relative ordering of the vertices and edges with respect to $\vecx^{\star}$ and $\vecy^{\star}$, and all the constraints in \eqref{eq:lp} hold for $\widehat{z}$. These imply that $\widehat{z}$ is a feasible solution to \eqref{eq:lp} as well. Moreover, it's not difficult to see that $\widehat{z}$ is an optimal solution of \eqref{eq:lp}, since otherwise by the linearity of \eqref{eq:linearcombination1} and \eqref{eq:linearcombination2}, $z_\setp$ would have a higher objective value than $z^{\star}$, contradicting the fact that $z^{\star}$ is an optimal solution. Hence, the non-empty level set  defined by $\widehat{z}$ corresponds to an optimal solution.
Finally, by applying the second case inductively, we prove the claimed statement of the lemma.
\end{proof}

By applying the lemma above and the linearity of the objective function of \eqref{eq:lp}, we obtain the following corollary.
\begin{corollary} \label{cor:lpmaximal}
    The following statements hold:
    \begin{itemize}
        \item Suppose that $\setp_1$ and $\setp_2$ are optimal subsets of $\setu$. Then, $\setp_1\cup \setp_2$, as well as $\setp_1\cap \setp_2$ satisfying $\setp_1\cap \setp_2\neq\emptyset$,  is  an optimal subset of $\setu$. 
        \item The optimal set of maximum size is unique, and contains all optimal subsets.
    \end{itemize}
\end{corollary}

Now we are ready to show that the procedure \textsc{ComputeChangeRate}~(Algorithm~\ref{algo:computechangerate}) and the linear program \eqref{eq:lp} together will always give us the same set of $\vecr$-values regardless of which optimal solution of \eqref{eq:lp} is used for the recursive construction of the entire vector $\vecr$.

\begin{lemma}
    Let $(\setu, \edgeset_\setu)$ be the input to \textsc{ComputeChangeRate}, and $\setp \subset \setu$ be the set returned by \eqref{eq:lp}. Moreover, let $\left(\setu'=\setu \setminus \setp, \edgeset_{\setu'}\right)$ be the input to the recursive call \textsc{ComputeChangeRate}$(\setu', \edgeset_{\setu'})$. Then, it holds for any $\setp' \subseteq \setu'$ that $\delta(\setp')\leq \delta(\setp)$, where the equality holds iff $\delta(\setp \cup \setp') = \delta(\setp)$.
\end{lemma}
\begin{proof}
By the definition of the function $c$ and sets $\cals^+, \cals^-, \cali^+, \cali^-$, it holds that
\[
c\left( \spp{U'}{P'} \right) = c\left( \spp{U}{ P\cup P'} \right) - c\left( \spp{U}{P} \right), 
\]
and the same equality holds for sets $\cals^-, \cali^+$ and $\cali^-$.
% By introducing  \[
% \delta_M = \delta(P) = \frac{c\left(\cals_P^+\right) + c\left(\cali_P^+\right) - c\left(\cals_P^-\right) - c\left(\cali_P^-\right)}{\vol(P)},
% \]
We have that 
    \begin{align*}
        \delta(\setp') & = \frac{c\left(\spp{U'}{P'}\right) + c\left(\ipp{U'}{P'}\right) - c\left(\spm{U'}{P'}\right) - c\left(\ipm{U'}{P'}\right)}{\vol(\setp')} \\
        & = \frac{\delta\left(\setp \union \setp'\right)\cdot \vol(\setp \union \setp') - \delta(\setp) \cdot \vol(\setp)}{\vol(\setp \union \setp') - \vol(\setp)}.
    \end{align*}
    Therefore, it holds  for any operator $\bowtie \in \{<, =, >\}$ that
    \begin{alignat*}{2}
        & & \delta(\setp') & \bowtie \delta(\setp) \\
        \iff & \quad & \frac{\delta(\setp \union \setp')\cdot \vol(\setp \union \setp') - \delta(\setp)\cdot \vol(\setp)}{\vol(\setp \union \setp') - \vol(\setp)} & \bowtie \delta(\setp) \\
        \iff & & \delta(\setp \union \setp') & \bowtie \delta(\setp),
    \end{alignat*}
    which implies that  $\delta(\setp') \leq \delta(\setp)$ iff $\delta(\setp \union \setp') \leq \delta(\setp)$ with equality iff $\delta(\setp \union \setp') = \delta(\setp)$.
    Since $\setp$ is optimal, it cannot be the case that $\delta\left( \setp \cup \setp' \right)>\delta(\setp)$, and therefore the lemma follows.
 \end{proof}
Combining everything together, we have the following result which summaries the properties of $\vecr$ computed by Algorithm~\ref{algo:computechangerate} and \eqref{eq:lp}  and establishes the first fact promised at the beginning of this section.
\begin{lemma}
    For any input instance $(\setu, \edgeset_\setu)$,  Algorithm~\ref{algo:computechangerate} always returns the same output $\vecr\in\mathbb{R}^{\cardinality{\setu}}$ no matter which optimal sets  are returned by solving the linear program~\eqref{eq:lp}. In particular, Algorithm~\ref{algo:computechangerate} always finds the unique optimal set $\setp \subseteq \setu$ of maximum size and assigns $\vecr(u) = \delta(\setp)$ to every $u\in \setp$. After removing the computed $\setp \subset \setu$, the computed $\vecr(v)=\delta(\setp')$ for some $\setp'\subseteq \setu \setminus \setp$ and  any $v\in \setp'$ is always strictly less than $\vecr(u)=\delta(\setp)$ for any $u\in \setp$.
\end{lemma}

\paragraph{Any $\vecr$ satisfying Rules (1) and (2) is computed by Algorithm~\ref{algo:computechangerate}.}
 Next we show that if there is any vector $\vecr$ which satisfies Rules (1) and (2), it must be equal to the output of Algorithm~\ref{algo:computechangerate}.
\begin{lemma} \label{lem:alg_computes_r}
For any hypergraph $\graphh = (\vertexset_\graphh, \edgeset_\graphh, \weight)$ and $\vecf_t \in \R^n$, if there is a vector $\vecr = \mathrm{d}\vecf_t/\mathrm{d}t$ with an associated $\{\vecr_e(v)\}_{e \in \edgeset_\graphh, v \in \vertexset_\graphh}$ which satisfies Rules~(1) and (2), then $\vecr$ is equal to the output of Algorithm~\ref{algo:computechangerate}.
\end{lemma} 
\begin{proof}
We will focus our attention on a single equivalence class $\setu \subset \vertexset$ where for any $u, v \in \setu$, $\vecf(u) = \vecf(v)$.
    Recall that for each $e \in \edgeset_\setu$, $c_{\vecf_t}(e) = \weight(e) \abs{\discrep_{\vecf_t}(e)}$, which is the rate of flow due to $e$ into $\setu$ (if $e \in \cals_{\setu}^+ \union \cali_{\setu}^+$) or out of $\setu$ (if $e \in \cals_{\setu}^- \union \cali_{\setu}^-$).
    Let $\vecr\in\R^n$ be the vector supposed to satisfy Rules~(1) and (2). We assume that $\setu \subseteq \vertexset$ is an arbitrary equivalence class, and define $$\set{T} \triangleq \left\{ u\in \setu: \vecr(u) = \max_{v\in \setu} \vecr(v)\right\}.$$
    Let us study which properties $\vecr$ must satisfy according to Rules~(1) and (2).
\begin{itemize}
    \item Assume that $e\in\cals_\setu^-$, i.e., it holds that $\sfte \cap U \neq\emptyset$ and $\discrep_{\vecf_t}(e)>0$. To satisfy Rule~(2a), it suffices to have that $c_{\vecf_t}(e)=\weight(e) \cdot  \discrep_{\vecf_t}(e) = - \sum_{v \in \sfe} \deg(v) \vecr_e(v) = - \sum_{v\in 
    \set{T}} \deg(v) \vecr_e(v)$ if $\sfte \cap \set{T} \neq\emptyset$, and $\vecr_e(v) = 0$ for all $v \in \set{T}$ otherwise. 
    \item Assume that $e\in\cals_\setu^+$, i.e., it holds that $\sfte \cap \setu \neq\emptyset$ and $\discrep_{\vecf_t}(e)<0$. To satisfy Rule~(2a), it suffices to have $\sfte \subseteq \set{T}$, or $\sfte \cap \set{T}=\emptyset$.
    \item Assume that $e\in \cali_\setu^+$, i.e., it holds that $\ife \cap \setu \neq\emptyset$ and $\discrep_{\vecf_t}(e)<0$. To satisfy Rule~(2b), it suffices to have that $c_{\vecf_t}(e) = \sum_{v\in \ife} \deg(v) \vecr_e(v) = \sum_{v\in \set{T}} \deg(v) \vecr_e(v)$ if $\ife \subseteq \set{T}$, and $\vecr_e(v) = 0$ for all $v \in \set{T}$ otherwise. 
    \item Assume that $e\in\cali_\setu^-$, i.e., it holds that $\ife \cap \setu \neq\emptyset$ and $\discrep_{\vecf_t}(e)>0$. To satisfy Rule~(2b), it suffices to have $\ife \subseteq \set{T}$, or $\ife \cap \set{T} =\emptyset$.
\end{itemize}
    Notice that the four conditions above needed to satisfy Rule~(2) naturally reflect our definitions of the sets $\spp{U}{P},\ipp{U}{P}, \spm{U}{P}$, and $\ipm{U}{P}$ and for all $u \in \set{T}$, it must be that $\vecr(u) = \delta(\set{T})$.
    
We will show that the output set $\setp$ returned by solving the linear program~\eqref{eq:lp}
 is the set $\set{T}$. To prove this, notice that on one hand, by  Corollary~\ref{cor:lpmaximal}, the linear program gives us the unique maximal optimal subset $\setp \subseteq \setu$, and every $v\in \setp$ satisfies that $\vecr(v)=\delta(\setp)\leq \vecr(u) = \delta(\set{T})$ for any $u\in \set{T}$ as every vertex in $\set{T}$ has the maximum $\vecr$-value. On the other side, we have that $\delta(\set{T}) \leq \delta(\setp)$ since  $\setp$ is the set returned by the linear program, and therefore  $\set{T}=\setp$. We can apply this argument recursively, and this proves that Algorithm~\ref{algo:computechangerate} must return the vector $\vecr$.
\end{proof}

\paragraph{The output of Algorithm~\ref{algo:computechangerate} satisfies Rules (1) and (2).}
 Now we show that the output of Algorithm~\ref{algo:computechangerate} does indeed satisfy Rules (1) and (2) which, together with Lemma~\ref{lem:alg_computes_r}, implies that there is exactly one such vector which satisfies the rules.
\begin{lemma} \label{lem:algo_satisfies_rules}
    For any hypergraph $\graphh = (\vertexset_\graphh, \edgeset_\graphh, \weight)$ and vector $\vecf_t \in \R^n$, the vector $\vecr$ constructed by Algorithm~\ref{algo:computechangerate} has corresponding $\{\vecr_e(v)\}_{e \in \edgeset_\graphh, v \in \vertexset_\graphh}$ which satisfies Rules~(1) and (2).
    Moreover, the $\{\vecr_e(v)\}_{e \in \edgeset_\graphh, v \in \vertexset_\graphh}$ values can be found in polynomial time using the vector $\vecr$.
\end{lemma} 
\begin{proof} 
\newcommand{\setetp}{\mathcal{E}_{\set{T}}^+}
\newcommand{\setetpp}{\mathcal{E}_{\set{T}'}^+}
\newcommand{\setetm}{\mathcal{E}_{\set{T}}^-}
\newcommand{\setetpm}{\mathcal{E}_{\set{T}'}^-}
We will focus on a single iteration of the algorithm, in which $r(v)$ is assigned for the vertices in some set $\set{T} \subset \vertexset_\graphh$.
We use the notation \[
\setetp = \ipp{U}{T} \union \spp{U}{T}, \qquad \setetm = \ipm{U}{T} \union \spm{U}{T}
\]
and will show that the values of $\vecr_e(v)$ for $e \in \setetp \union \setetm$ can be found and satisfy Rules (1) and (2). Therefore, by applying this argument to each recursive call of the algorithm, we establish the lemma.
Given the set $\set{T}$, construct the following undirected flow graph, which is illustrated in  Figure~\ref{fig:maxflow}.
\begin{itemize}
    \item The vertex set is $\setetp \union \setetm \union \set{T} \union \{s, t\}$.
    \item For all $e \in \setetp$, there is an edge $(s, e)$ with capacity $c_{\vecf_t}(e)$.
    \item For all $e \in \setetm$, there is an edge $(e, t)$ with capacity $c_{\vecf_t}(e)$.
    \item For all $v \in \set{T}$, if $\delta(\set{T}) \geq 0$, there is an edge $(v, t)$ with capacity $\deg(v) \delta(\set{T})$.
    Otherwise, there is an edge $(s, v)$ with capacity $\deg(v)\abs{\delta(\set{T})}$.
    \item For each $e \in \setetp \union \setetm$, and each $v \in \set{T} \intersect \left(\sfte \union \ifte \right)$, there is an edge $(e, v)$ with  capacity $\infty$.
\end{itemize}

\begin{figure}[h]
    \centering
    \begin{subfigure}{0.49\textwidth}
    \centering
        \begin{tikzpicture}
	\begin{pgfonlayer}{nodelayer}
		\node [style=smallblack] (0) at (-5, 0) {};
		\node [style=square] (1) at (-2, 1.5) {};
		\node [style=square] (2) at (-2, 2.5) {};
		\node [style=square] (4) at (-2, -0.5) {};
		\node [style=none] (5) at (-5.35, -0.25) {s};
		\node [style=basic] (7) at (2, -0.5) {};
		\node [style=basic] (8) at (2, -1.5) {};
		\node [style=basic] (9) at (2, -2.5) {};
		\node [style=square] (12) at (2, 1.5) {};
		\node [style=square] (13) at (2, 2.5) {};
		\node [style=square] (15) at (2, 0.5) {};
		\node [style=none] (17) at (2, -3.5) {$T$};
		\node [style=smallblack] (18) at (5, 0) {};
		\node [style=none] (19) at (5.3, -0.25) {t};
		\node [style=square] (29) at (-2, 0.5) {};
		\node [style=square] (30) at (-2, -1.5) {};
		\node [style=square] (31) at (-2, -2.5) {};
		\node [style=vertical ellipse dashed] (32) at (-2, 0) {};
		\node [style=none] (33) at (-2, 3.5) {};
		\node [style=none] (34) at (-2, 3.5) {$\mathcal{E}_T^+$};
		\node [style=small vertical ellipse dashed] (36) at (2, 1.5) {};
		\node [style=none] (37) at (2, 3.5) {$\mathcal{E}_T^-$};
		\node [style=small vertical ellipse dashed] (38) at (2, -1.5) {};
		\node [style=none] (40) at (0, 0) {$\vdots$};
		\node [style=none] (42) at (-4, 1.5) {\tiny $c(e)$};
		\node [style=none] (44) at (4.125, 1.5) {\tiny $c(e)$};
		\node [style=none] (46) at (4.125, -1.75) {\tiny $d(v)\delta(T)$};
		\node [style=none] (47) at (0, 1.5) {\tiny $\infty$};
	\end{pgfonlayer}
	\begin{pgfonlayer}{edgelayer}
		\draw (13) to (18);
		\draw (12) to (18);
		\draw (15) to (18);
		\draw (7) to (18);
		\draw (8) to (18);
		\draw (9) to (18);
		\draw (0) to (2);
		\draw (0) to (1);
		\draw (0) to (4);
		\draw (4) to (9);
		\draw (0) to (29);
		\draw (0) to (30);
		\draw (0) to (31);
		\draw (31) to (8);
		\draw [bend left=60, looseness=1.75] (8) to (15);
		\draw (2) to (7);
	\end{pgfonlayer}
\end{tikzpicture}
    \end{subfigure}
    \begin{subfigure}{0.49\textwidth}
    \centering
        \begin{tikzpicture}
	\begin{pgfonlayer}{nodelayer}
		\node [style=smallblack] (0) at (-5, -0.5) {};
		\node [style=square] (1) at (-2, 2) {};
		\node [style=square] (2) at (-2, 3) {};
		\node [style=square] (4) at (-2, -1) {};
		\node [style=none] (5) at (-5.375, -0.75) {s};
		\node [style=basic] (7) at (2, 1) {};
		\node [style=basic] (8) at (2, -2) {};
		\node [style=basic] (9) at (2, -3) {};
		\node [style=square] (12) at (2, 2) {};
		\node [style=square] (13) at (2, 3) {};
		\node [style=square] (15) at (2, -1) {};
		\node [style=none] (17) at (2, -4) {$T'$};
		\node [style=smallblack] (18) at (5, -0.5) {};
		\node [style=none] (19) at (5.3, -0.75) {t};
		\node [style=square] (29) at (-2, 1) {};
		\node [style=square] (30) at (-2, -2) {};
		\node [style=square] (31) at (-2, -3) {};
		\node [style=none] (34) at (-2, -0.25) {$\mathcal{E}_{T'}^+$};
		\node [style=none] (37) at (2, -0.25) {$\mathcal{E}_{T'}^-$};
		\node [style=none] (40) at (0, 2) {$\vdots$};
		\node [style=none] (42) at (-3.925, -1.925) {\tiny $c(e)$};
		\node [style=none] (44) at (3.7, -0.45) {\tiny $c(e)$};
		\node [style=none] (46) at (4.325, -2) {\tiny $d(v)\delta(T)$};
		\node [style=tiny vertical ellipse dashed] (47) at (2, -2.5) {};
		\node [style=none] (48) at (-3, 3.25) {};
		\node [style=none] (49) at (-3, 1) {};
		\node [style=none] (50) at (-2.5, 0.5) {};
		\node [style=none] (51) at (2.5, 0.5) {};
		\node [style=none] (52) at (3, 0) {};
		\node [style=none] (53) at (3, -3.25) {};
		\node [style=none] (54) at (0, -1.575) {\tiny $\infty$};
	\end{pgfonlayer}
	\begin{pgfonlayer}{edgelayer}
		\draw (13) to (18);
		\draw (12) to (18);
		\draw (15) to (18);
		\draw (7) to (18);
		\draw (8) to (18);
		\draw (9) to (18);
		\draw (0) to (2);
		\draw (0) to (1);
		\draw (0) to (4);
		\draw (4) to (9);
		\draw (0) to (29);
		\draw (0) to (30);
		\draw (0) to (31);
		\draw (31) to (8);
		\draw [bend left=60, looseness=1.75] (8) to (15);
		\draw [style=blue dashed] (48.center)
			 to (49.center)
			 to [bend right=45] (50.center)
			 to (51.center)
			 to [bend left=45] (52.center)
			 to (53.center);
	\end{pgfonlayer}
\end{tikzpicture}
    \end{subfigure}
    \caption{ \textbf{Left}: an illustration of the constructed max-flow graph, when $\delta(\set{T}) \geq 0$.
    The minimum cut is given by $\{s\}$. \textbf{Right}: a cut induced by $\set{T}' \subset \set{T}$. We can assume that every $e \in \setetp \union \setetm$ connected to $\set{T}'$ is on the same side of the cut as $\set{T}'$. Otherwise, there would be an edge with infinite capacity crossing the cut.}
    \label{fig:maxflow}
\end{figure}
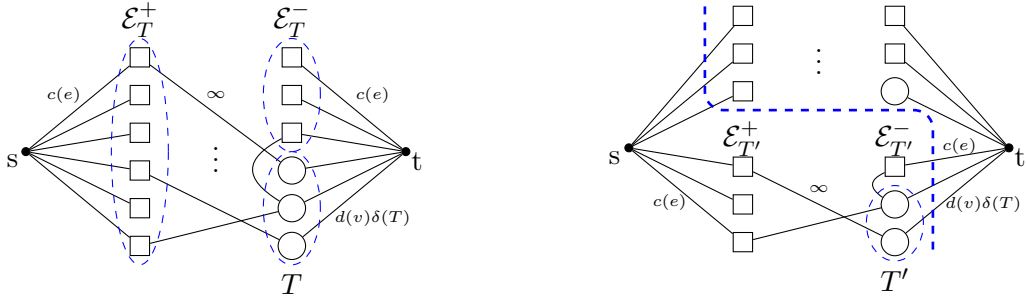

We use $\cut(\seta)$ to denote the weight of the cut defined by the set $\seta$ in this constructed graph and note that $\cut(\{s\}) = \cut(\{t\})$ since
\begin{align*}
    \cut(\{s\}) - \cut(\{t\}) = \sum_{e \in \setetp} c_{\vecf}(e) - \sum_{e \in \setetm} c_{\vecf}(e) - \vol(\set{T}) \delta(\set{T}) = 0
\end{align*}
by the definition of $\delta(\set{T})$.

Now, suppose that the maximum flow value on this graph is $\cut(\{s\})$, and let the corresponding flow from $u$ to $v$ be given by $\Theta(u, v) = - \Theta(v, u)$.
Then, set $\deg(v) \vecr_e(v) = \Theta(e, v)$ for any $e \in \setetp \union \setetm$ and $v \in \set{T} \intersect e$.
This configuration of the values $\vecr_e(v)$ would be compatible with the vector $\vecr$ computed by Algorithm~\ref{algo:computechangerate} and would satisfy the rules of the diffusion process for the following reasons:
\begin{itemize}
    \item For all $v \in \set{T}$, the edge $(v, t)$ or $(s, v)$ is saturated and so $\sum_{e \in \edgeset} \deg(v) \vecr_e(v) = \deg(v) \delta(\set{T}) = \deg(v) \vecr(v)$.
    \item For any $e \in \setetp \union \setetm$, the edge $(s, e)$ or $(e, t)$ is saturated and so we have $\sum_{v \in \set{T} \intersect e} \deg(v) \vecr_v(e) = - \weight(e) \discrep(e)$. Since $\setetp \union \setetm$ is removed in the recursive step of the algorithm, $\vecr_e(v) = 0$ for all $v \in \setu \setminus \set{T}$ and so $\sum_{v \in \setu \intersect e} \deg(v) \vecr_e(v) = - \weight(e) \discrep(e)$. This establishes Rule~(1) since $\setu \intersect e$ is equal to either $\sfe$ or $\ife$.
    \item For edges in $\spp{U}{T}$ (resp.\ $\ipp{U}{T}$, $\ipm{U}{T}$), since $\sfte$ (resp.\ $\ifte$, $\ifte$) is a subset of $\set{T}$ and every $v \in \set{T}$ has the same value of $\vecr(v)$, Rule~(2) is satisfied.
    For edges in $\spm{U}{T}$, for any $v \not \in \set{T}$, we have $\vecr_e(v) = 0$ and $\vecr(v) < \delta(\set{T})$ which satisfies Rule~(2).
\end{itemize}

We will now show that every cut separating $s$ and $t$ has weight at least $\cut(\{s\})$ which will establish that the maximum flow on this graph is $\cut(\{s\})$ by the max-flow min-cut theorem.

Consider some arbitrary cut given by $\set{X} = \{s\} \union \set{T}' \union \setetpp \union \setetpm$ where $\set{T}'$ (resp.\ $\setetpp$, $\setetpm$) is a subset of $\set{T}$ (resp.\ $\setetp$, $\setetm$).
Figure~\ref{fig:maxflow} illustrates this cut.
Since all of the edges not connected to $s$ or $t$ have infinite capacity, we can assume that no such edge crosses the cut which implies that
\begin{itemize}
    \item For all $e \in \setetpp$, $e \intersect (\set{T} \setminus \set{T}') = \emptyset$.
    \item For all $e \in \setetpm$, $e \intersect (\set{T} \setminus \set{T}') = \emptyset$.
    \item For all $e \in (\setetp \setminus \setetpp)$, $e \intersect \set{T}' = \emptyset$.
    \item For all $e \in (\cale_T^- \setminus \setetpm)$, $e \intersect \set{T}' = \emptyset$.
\end{itemize}
These conditions, along with the definition of $\setetp$ and $\setetm$, allow us to assume that $\setetpp = \ipp{U}{T'} \union \spp{U}{T'}$ and $\setetpm = \ipm{U}{T'} \union \spm{U}{T'}$.
The size of this arbitrary cut is
\[
    \cut(\set{X}) = \cut(\{s\}) - \sum_{e \in \setetpp} c(e) + \sum_{e \in \setetpm} c(e) + \sum_{v \in \set{T}'} \deg(v) \delta(\set{T}).
\]
Since $\set{T}$ maximises the objective function $\delta$, we have
\[
    \sum_{e \in \setetpp} c(e) - \sum_{e \in \setetpm} c(e) = \vol(\set{T}') \delta(\set{T}') \leq \vol(\set{T}') \delta(\set{T}) = \sum_{v \in \set{T}'} \deg(v) \delta(\set{T})
\]
and can conclude that $\cut(\set{X}) \geq \cut(\{s\})$ which completes the proof.
\end{proof}

We can now combine the results in Lemmas~\ref{lem:alg_computes_r} and \ref{lem:algo_satisfies_rules} to prove Lemma~\ref{lem:rulesimplydiffusion}.
\begin{proof}[Proof of Lemma~\ref{lem:rulesimplydiffusion}.]
    Lemma~\ref{lem:alg_computes_r} and Lemma~\ref{lem:algo_satisfies_rules} together imply that there is a unique vector $\vecr$ and corresponding $\{\vecr_e(v)\}_{e \in \edgeset_\graphh, v \in \vertexset_\graphh}$ which satisfies Rules~(1) and (2).
    Lemma~\ref{lem:alg_computes_r} further shows that Algorithm~\ref{algo:computechangerate} computes this vector $\vecr$, and the proof of Lemma~\ref{lem:algo_satisfies_rules} gives a polynomial-time method for computing the $\{\vecr_e(v)\}$ values by solving a sequence of max-flow problems.
\end{proof}

\section{Analysis of Algorithm~\ref{algo:main}} \label{sec:thm1proof}
In this section we will 
analyse the \diffalgname\ algorithm and
prove Theorem~\ref{thm:mainalg}.
This section is split into two subsections which correspond to the two statements in Theorem~\ref{thm:mainalg}.
First, we show that the diffusion process converges to an eigenvector of $\signlaph$.
We then show that this allows us to find sets $\setl, \setr \subset \vertexset_\graphh$ with low hypergraph bipartiteness.

\subsection{Convergence of the diffusion process}
We show in this section that the diffusion process determined by the operator $\signlaph$ converges in polynomial time to an eigenvector of $\signlaph$.

\begin{theorem} \label{thm:convergence}
For any $\epsilon > 0$, there is some $t = \bigo{1 / \epsilon^3}$ such that for any starting vector $\vecf_0$, there is an interval $[c, c + 2\epsilon]$ such that
\[
    \frac{1}{\norm{\vecf_t}_\weight}\sum_{\vec{u}_i : \lambda_i \in [c, c + 2\epsilon]} \inner{\vecf_t}{\vec{u}_i}_\weight \geq 1 - \epsilon
\]
where $(\vec{u}_i, \lambda_i)$ are the eigen-pairs of $\degm_\graphh^{-1} \signlap_t = \degm_\graphh^{-1}(\degm_{\graphg_t} + \adj_{\graphg_t})$ and $\graphg_t$ is the graph constructed to represent the diffusion operator $\signlaph$ at time $t$.
\end{theorem}
By taking $\epsilon$ to be some small constant, this shows that the vector $\vecf_t$ converges to an eigenvector of the hypergraph operator in polynomial time.  

\begin{proof}
 We will prove this by showing that the Rayleigh quotient $R_{\signlaph}(\vecf_t)$ is always decreasing at a rate of at least $\epsilon^3$ whenever the conclusion of Theorem \ref{thm:convergence} does not hold. Since $R_{\signlaph}(\vecf_t)$ can only decrease by a constant amount, the theorem will follow.

    \newcommand{\sumj}{\sum_{j = 1}^n}
    \newcommand{\ajs}{\alpha_j^2}
    \newcommand{\halfdelta}{\frac{\delta}{2}}
    First, we derive an expression for the rate of change of the Rayleigh quotient $R_{\signlaph}(\vecf_t)$.
    Let $\vecx_t \triangleq \degmhalf_\graphh \vecf_{t}$, and let $\signlapn_t = \degmhalfneg_\graphh (\degm_{\graphg_t} + \adj_{\graphg_t})\degmhalfneg_\graphh$.  Then, we have
     \[
        R_{\signlap_\graphh}(\vecf_t) = \frac{\vecf_t^\transpose (\degm_{\graphg_t} + \adj_{\graphg_t}) \vecf_t}{\vecf_t^\transpose \degm_\graphh \vecf_t} = \frac{\vecx_t^\transpose \signlapn_t \vecx_t}{\vecx_t^\transpose \vecx_t}.
    \]
    For every eigenvector $\vecu_j$ of $\degm_\graphh^{-1} \signlaph$, the vector $\vecv_j = \degmhalf_\graphh \vecu_j$ is an eigenvector of $\signlapn$ with the same eigenvalue $\lambda_i$.
 Since $\signlapn$ is symmetric, the vectors $\vecv_1$ to $\vecv_n$ are orthogonal. 
    Additionally, notice that
    \[
    \inner{\vecx_t}{\vecv_j} = \vecf^\transpose_t \degm_\graphh \vecu_j = \inner{\vecf_{t}}{\vecu_j}_\weight
    \]
    and we define $\alpha_j = \inner{\vecf_t}{\vecu_j}_w$ so we can write $\vecx_t = \sum_{j = 1}^n \alpha_j \vecv_j$ and $\vecf_{t} = \degmhalfneg_\graphh \vecx_t = \sum_{j = 1}^n \alpha_j \vecu_j$.
    Now, we have that
    \begin{align*}
        \frac{\mathrm{d}}{\mathrm{d}t} \inner{\vecx_t}{\signlapn_{t} \vecx_t} & = \inner{\frac{\mathrm{d}}{\mathrm{d}t} \vecx_t}{\signlapn_t \vecx_t} + \inner{\vecx_t}{\frac{\mathrm{d}}{\mathrm{d}t} \signlapn_t \vecx_t} \\
        & = - \vecx_t^\transpose \signlapn_t^2 \vecx_t - \vecx_t^\transpose \signlapn_t^2 \vecx_t \\
        & = - 2 \sum_{j = 1}^n \alpha_j^2 \lambda_j^2.
    \end{align*}
    Additionally, we have
    \[
        \frac{\mathrm{d}}{\mathrm{d}t} \vecx_t^\transpose \vecx_t = - \vecx_t^\transpose \signlapn_t \vecx_t - \vecx_t^\transpose \signlapn_t \vecx_t 
        = - 2 \vecx_t^\transpose \signlapn_t \vecx_t.
    \]
    Recalling that $\vecx_t^\transpose \vecx_t = \sum_{j = 1}^n \alpha_j^2$, this gives 
    \newcommand{\dddelta}{\frac{\mathrm{d}}{\mathrm{d} \delta}}
    \begin{align}
        \frac{\mathrm{d}}{\mathrm{d} t} R(\vecf_t) & = \frac{1}{(\vecx_t^\transpose \vecx_t)^2} \left[\left(\frac{\mathrm{d}}{\mathrm{d} t} (\vecx_t^\transpose \signlapn_t \vecx_t) \right) (\vecx_t^\transpose \vecx_t) - \left(\frac{\mathrm{d}}{\mathrm{d}t} (\vecx_t^\transpose \vecx_t) \right) (\vecx_t^\transpose \signlapn_t \vecx_t) \right] \nonumber \\
        & = \frac{1}{\vecx_t^\transpose \vecx_t} \left(\frac{\mathrm{d}}{\mathrm{d}t} (\vecx_t^\transpose \signlapn_t \vecx_t)\right) + 2 R(\vecf_t)^2 \nonumber \\
        & = 2 \left[ R(\vecf_t)^2 - \frac{1}{\sum_{j = 1}^n \alpha_j^2} \sumj \ajs \lambda_j^2\right] \label{eq:ddt_gt}
    \end{align}
    We will now show that at any time $t$, if the conclusion of the theorem does not hold, then
    \begin{equation} \label{eq:gt_decreases}
        \frac{\mathrm{d}}{\mathrm{d} t} R(\vecf_t) \leq - \epsilon^3.
    \end{equation}
    Assuming that this holds, and using the fact that for all $t$ it is the case that $0 \leq R(\vecf_t) \leq 2$, when $t = 2 / \epsilon^3$, either $R(\vecf_t) = 0$ or there was some $t' < t$ when $(\mathrm{d}/\mathrm{d}t') R(\vecf_{t'}) > - \epsilon^3$ and the conclusion of the theorem holds.
    
    Now, to show \eqref{eq:gt_decreases}, consider the partial derivative
    \begin{align}
        \frac{\partial}{\partial \lambda_i} 2 \left[ R(\vecf_t)^2 - \frac{1}{\sum_{j = 1}^n \alpha_j^2} \sumj \ajs \lambda_j^2 \right] & = 2 \left[ \frac{2 \alpha_i^2}{\sum_{j = 1}^n \alpha_j^2} R(\vecf_t) - \frac{2 \alpha_i^2}{\sum_{j = 1}^n \alpha_j^2} \lambda_i \right] \nonumber \\
        & = \frac{4 \alpha_i^2}{\sum_{j = 1}^n \alpha_j^2} (R(\vecf_t) - \lambda_i), \label{eq:partial}
    \end{align}
    where we use the fact that $R(\vecf_t) = (\sum_{j = 1}^n \alpha_j^2 \lambda_j) / (\sum_{j = 1}^n \alpha_j^2)$.
    Notice that
    the derivative in \eqref{eq:partial} is greater than $0$
    if $\lambda_i < R(\vecf_t)$
    and
    the derivative is less than $0$
    if $\lambda_i > R(\vecf_t)$.
    This means that in order to establish an upper-bound for $(\mathrm{d}/\mathrm{d} t) R(\vecf_t)$, we can assume that the eigenvalues $\lambda_j$ are as close to the value of $R(\vecf_t)$ as possible.
    
    Now, we assume that at time $t$, the conclusion of the theorem does not hold. Then, one of the following cases must hold: 
    \begin{enumerate}
        \item $(\sum_{j : \lambda_j > R(\vecf_t) + \epsilon} \alpha_j^2) / (\sum_{j = 1}^n \alpha_j^2) > \epsilon / 2$;
        \item $(\sum_{j : \lambda_j < R(\vecf_t) - \epsilon} \alpha_j^2) / (\sum_{j = 1}^n \alpha_j^2) > \epsilon / 2$.
    \end{enumerate}
    Suppose the first case holds.
    By the conclusions we draw from \eqref{eq:partial}, we can assume that there is an eigenvalue $\lambda_i = R(\vecf_t) + \epsilon$ such that $\alpha_i^2 / (\sum_{j = 1}^n \alpha_j^2) = \epsilon / 2$ and that $(\sum_{j: \lambda_j < R(\vecf_t)} \alpha_j^2) / (\sum_{j = 1}^n \alpha_j^2) = 1 - \epsilon / 2$.
    Then, since $R(\vecf_t) = (\sum_{j = 1}^n \alpha_j^2 \lambda_j) / (\sum_{j = 1}^n \alpha_j^2)$, we have
    \begin{align*}
        \frac{1}{\sum_{j = 1}^n \alpha_j^2} \sum_{j: \lambda_j < R(\vecf_t)} \alpha_j^2 \lambda_j = R(\vecf_t) - \frac{\epsilon}{2}(R(\vecf_t) + \epsilon),
    \end{align*}
    which is equivalent to
    \begin{equation} \label{eq:ljminusgt}
        \frac{1}{\sumj \ajs} \sum_{j: \lambda_j < R(\vecf_t)} \ajs (\lambda_j - R(\vecf_t)) = \frac{\epsilon}{2}\cdot R(\vecf_t) - \frac{\epsilon}{2}\cdot (R(\vecf_t) + \epsilon) = - \frac{\epsilon^2}{2}.
    \end{equation}
    Now, notice that for any $\lambda_j < R(\vecf_t)$ we have
    \begin{align*}
        (\lambda_j^2 - R(\vecf_t)^2) & = (\lambda_j + R(\vecf_t)) (\lambda_j - R(\vecf_t)) \\
        & \geq 2 R(\vecf_t) (\lambda_j - R(\vecf_t)),
    \end{align*}
    since $\lambda_j - R(\vecf_t) < 0$.
    As such, we have
    \begin{align*}
        \frac{\mathrm{d}}{\mathrm{d}t} R(\vecf_t) & = 2 \left[ R(\vecf_t)^2 - \frac{1}{\sum_{j = 1}^n \alpha_j^2} \sumj \ajs \lambda_j^2\right] \\
        & = - \frac{2}{\sum_{j = 1}^n \alpha_j^2} \sumj \ajs (\lambda_j^2 - R(\vecf_t)^2) \\
        & = - \epsilon \left((R(\vecf_t) + \epsilon)^2 - R(\vecf_t)^2\right) - \frac{2}{\sumj \ajs} \sum_{j: \lambda_j < R(\vecf_t)} \alpha_j^2 (\lambda_j^2 - R(\vecf_t)^2) \\
        & \leq - \epsilon \left(2 \epsilon R(\vecf_t) + \epsilon^2\right) - \frac{2}{\sumj \ajs} \sum_{j: \lambda_j < R(\vecf_t)} 2 \ajs R(\vecf_t) (\lambda_j - R(\vecf_t)) \\
        & = - 2 \epsilon^2 R(\vecf_t) - \epsilon^3 + 2 \epsilon^2 R(\vecf_t) \\
        & = - \epsilon^3
    \end{align*}
    where the fifth line follows by \eqref{eq:ljminusgt}.
    We now consider the second case.
    We can assume that there is an eigenvalue $\lambda_i = R(\vecf_t) - \epsilon$ such that $\alpha_i^2 / (\sumj \ajs) = \epsilon / 2$ and that $(\sum_{j: \lambda_j > R(\vecf_t)} \ajs) / (\sumj \ajs) = 1 - \epsilon / 2$.
    Then, we have
    \begin{align*}
        \frac{1}{\sum_{j = 1}^n \alpha_j^2} \sum_{j: \lambda_j > R(\vecf_t)} \alpha_j^2 \lambda_j = R(\vecf_t) - \frac{\epsilon}{2}\left(R(\vecf_t) - \epsilon\right),
    \end{align*}
    which is equivalent to
    \begin{equation} \label{eq:ljminusgt2}
        \frac{1}{\sumj \ajs} \sum_{j: \lambda_j > R(\vecf_t)} \ajs (\lambda_j - R(\vecf_t)) = \frac{\epsilon}{2} R(\vecf_t) - \frac{\epsilon}{2}(R(\vecf_t) - \epsilon) = \frac{\epsilon^2}{2}.
    \end{equation}
    Now, notice that for any $\lambda_j > R(\vecf_t)$ we have
    \begin{align*}
        \lambda_j^2 - R(\vecf_t)^2 & = (\lambda_j + R(\vecf_t)) (\lambda_j - R(\vecf_t)) \\
        & \geq 2 R(\vecf_t) \cdot (\lambda_j - R(\vecf_t)).
    \end{align*}
    As such, we have
    \begin{align*}
        \frac{\mathrm{d}}{\mathrm{d} t} R(\vecf_t) & = - \frac{2}{\sum_{j = 1}^n \alpha_j^2} \sumj \ajs (\lambda_j^2 - R(\vecf_t)^2) \\
        & = - \epsilon \left((R(\vecf_t) - \epsilon)^2 - R(\vecf_t)^2\right) - \frac{2}{\sumj \ajs} \sum_{j: \lambda_j > R(\vecf_t)} \alpha_j^2 (\lambda_j^2 - R(\vecf_t)^2) \\
        & \leq - \epsilon \left(\epsilon^2 - 2 \epsilon R(\vecf_t)\right) - \frac{2}{\sumj \ajs} \sum_{j: \lambda_j < R(\vecf_t)} 2 \ajs R(\vecf_t) (\lambda_j - R(\vecf_t)) \\
        & = 2 \epsilon^2 g(t) - \epsilon^3 - 2 \epsilon^2 R(\vecf_t) \\
        & = - \epsilon^3
    \end{align*}
    where the fourth line follows by \eqref{eq:ljminusgt2}.
    These two cases establish \eqref{eq:gt_decreases} and complete the proof of the theorem.
\end{proof}

\paragraph{The eigenvector is at most the minimum eigenvector of the clique graph.}
 We now show that the eigenvector to which the algorithm converges is at most the minimum eigenvector of $\signlapg$ where $\graphg$ is the clique reduction of $\graphh$.
 We start by showing the following technical lemma.
\begin{lemma} \label{lem:cliqueedge}
    For any hypergraph $\graphh = (\vertexset_\graphh, \edgeset_\graphh, \weight)$, vector $\vecf \in \R^n$, and edge $e \in \edgeset_\graphh$, it holds that
    \[
        \left(\max_{u \in e} \vecf(u) + \min_{v \in e} \vecf(v)\right)^2 \leq \sum_{u, v \in e} \frac{1}{\rank(e) - 1} (\vecf(u) + \vecf(v))^2.
    \]
    The equality holds iff there is exactly one vertex $v \in e$ with $\vecf(v) \neq 0$ or $\rank(e) = 2$.
 
\end{lemma}
\begin{proof}
    We will consider some ordering of the vertices in $e$,
    \[
    u_1, u_2, \ldots, u_{\rank(e)},
    \]
    such that $u_1 = \argmax_{u \in e} \vecf(u)$ and $u_2 = \argmin_{u \in e} \vecf(u)$ and the remaining vertices are ordered arbitrarily. 
    Then, for any $2 \leq k \leq \rank(e)$, we define
    \[
        C_k = \sum_{u, v \in \{u_1, \ldots, u_k\}} \frac{1}{k - 1} (\vecf(u) + \vecf(v))^2
    \]
    and we will show by induction on $k$ that 
    \begin{equation} \label{eq:induction}
    C_k \geq \left(\max_{u \in e} \vecf(u) + \min_{v \in e} \vecf(v)\right)^2
    \end{equation}
    for all $2 \leq k \leq \rank(e)$ with equality iff $k = 2$ or there is exactly one vertex $u_i \in e$ with $\vecf(v) \neq 0$.
    The lemma follows by setting $k = \rank(e)$.

    The base case when $k = 2$ follows trivially by the definitions and the choice of $u_1$ and $u_2$.

    For the inductive step, we assume that~\eqref{eq:induction} holds for some $k$ and will show that it holds for $k + 1$. 
    We have that
    \[
        C_{k + 1} = \sum_{u, v \in \{u_1, \ldots, u_{k+1}\}} \frac{1}{k}\cdot (\vecf(u) + \vecf(v))^2,
    \]
    which is equivalent to
    \begin{align*}
        C_{k + 1} & = \frac{1}{k} \sum_{i = 1}^k (\vecf(u_i) + \vecf(u_{k+1}))^2 + \frac{1}{k} \sum_{u, v \in \{u_1, \ldots, u_k\}} (\vecf(u) + \vecf(v))^2 \\
        & = \frac{1}{k} \sum_{i = 1}^k (\vecf(u_i) + \vecf(u_{k + 1}))^2 + \frac{k - 1}{k} C_k \\
        & \geq \left(1 - \frac{1}{k}\right)\left(\max_{u \in e}\vecf(u) + \min_{v \in e} \vecf(v)\right)^2 + \frac{1}{k} \sum_{i = 1}^k (\vecf(u_1) + \vecf(u_{k + 1}))^2
    \end{align*}
    where the final inequality holds by the induction hypothesis.
    Therefore, it is sufficient to show that
    \[
        \sum_{i = 1}^k \left(\vecf(u_i) + \vecf(u_{k + 1})\right)^2 \geq \left(\max_{u \in e} \vecf(u) + \min_{v \in e} \vecf(v)\right)^2.
    \]
    We will instead show the stronger fact that
    \begin{equation} \label{eq:stronginequality}
        \left(\max_{v \in e}\vecf(v) + \vecf(u_{k + 1})\right)^2 +
        \left(\min_{v \in e}\vecf(v) + \vecf(u_{k + 1})\right)^2 
        \geq
        \left(\max_{u \in e}\vecf(u) + \min_{v \in e}\vecf(v)\right)^2.
    \end{equation}
    The proof is by case distinction.
    The first case is when $\sign(\max_{u \in e}\vecf(v)) = \sign(\min_{u \in e}\vecf(u))$.
    Assume w.l.o.g.\ that the sign is positive.
    Then, since $\vecf(u_{k+1}) \geq \min_{v \in e} \vecf(v)$, we have
    \[
        \left(\max_{v \in e} \vecf(v) + \vecf(u_{k + 1})\right)^2 \geq \left(\max_{v \in e} \vecf(v) + \min_{u \in e} \vecf(v)\right)^2
    \]
    and~\eqref{eq:stronginequality} holds.
    Moreover, the inequality is strict if $\abs{\min_{u \in e}\vecf(u)} > 0$ or $\abs{\vecf(u_{k+1})} > 0$.

    For the second case, we assume that $\sign(\min_{u \in e}\vecf(u)) \neq \sign(\max_{v \in e}\vecf(v))$.
    Expanding~\eqref{eq:stronginequality}, we would like to show
    \begin{align*}
        & \left(\max_{u \in e} \vecf(u)\right)^2 + \left(\min_{v \in e}\vecf(v)\right)^2 + 2 \vecf(u_{k + 1})\left(\max_{u \in e}\vecf(u)\right) + 2 \vecf(u_{k + 1})\left(\min_{v \in e}\vecf(v)\right) + 2 \vecf(u_{k + 1})^2 \\
        & \geq \left(\max_{u \in e}\vecf(u)\right)^2 + \left(\min_{u \in e} \vecf(u)\right)^2 - 2\left(\max_{u \in e}\vecf(u)\right)\abs{\min_{u \in e}\vecf(u)},
    \end{align*}
    which is equivalent to
    \[
        2 \vecf(u_{k + 1})^2 + 2 \vecf(u_{k + 1})\left(\max_{u \in e}\vecf(u)\right) + 2 \vecf(u_{k + 1})\left(\min_{v \in e}\vecf(v)\right) \geq -2 \left(\max_{u \in e}\vecf(u)\right) \abs{\min_{v \in e}\vecf(v)}.
    \]
    Notice that exactly one of the terms on the left hand side is negative.
    Recalling that $\min_{u \in e}\vecf(u) \leq \vecf(u_{k + 1}) \leq \max_{v \in e} \vecf(v)$, it is clear that
    \begin{itemize}
        \item if $\vecf(u_{k+1}) < 0$, then $-2\left(\max_{u \in e}\vecf(u)\right)\abs{\min_{v \in e}\vecf(v)} \leq 2 \vecf(u_{k + 1})\left(\max_{v \in e}\vecf(v)\right) \leq 0$ and the inequality holds.
        \item if $\vecf(u_{k+1}) \geq 0$, then $-2\left(\max_{u \in e}\vecf(u)\right)\abs{\min_{v \in e}\vecf(v)} \leq 2 \vecf(u_{k + 1})\left(\min_{u \in e}\vecf(u)\right) \leq 0$ and the inequality holds.
    \end{itemize}
    Moreover, in both cases the inequality is strict if $-2 \left(\max_{v \in e}\vecf(v)\right) \abs{\min_{u \in e}\vecf(u)}< 0$ or $\abs{\vecf(u_{k_1})} > 0$.
\end{proof}

 Now, we can show that we always find an eigenvector which is at most the minimum eigenvector of the clique reduction.

\begin{lemma} \label{lem:bettereigenvalue}
        For any hypergraph $\graphh = (\vertexset_\graphh, \edgeset_\graphh, \weight)$ with clique reduction $\graphg$, if $\vecf$ is the eigenvector corresponding to $\lambda_1(\degm_\graphg^{-1} \signlapg)$, then
        \[
            \frac{\vecf^\transpose \signlaph \vecf}{\vecf^\transpose \degm_\graphh \vecf} \leq \lambda_1(\degm_\graphg^{-1} \signlapg)
        \]
        and the inequality is strict if $\min_{e \in \edgeset_\graphh} \rank(e) > 2$.
\end{lemma}
\begin{proof}
    Since $\lambda_1(\degm_\graphg^{-1} \signlapg) = (\vecf^\transpose \signlapg \vecf) / (\vecf^\transpose \degm_\graphg \vecf)$ and $(\vecf^\transpose \degm_\graphg \vecf) = (\vecf^\transpose \degm_\graphh \vecf)$ by the construction of the clique graph, it suffices to show that
    \[
        \vecf^\transpose \signlaph \vecf \leq \vecf^\transpose \signlapg \vecf,
    \]
    which is equivalent to
    \begin{align*}
        \sum_{e \in \edgeset_\graphh} \weight(e) \left(\max_{v \in e} \vecf(v) + \min_{u \in e}\vecf(u)\right)^2 & \leq \sum_{(u, v) \in \edgeset_\graphg} \weight_\graphg(u, v) (\vecf(u) + \vecf(v))^2 \\
        & = \sum_{e \in \edgeset_\graphh} \weight(e) \sum_{u, v \in e} \frac{1}{\rank(e) - 1} (\vecf(u) + \vecf(v))^2
    \end{align*}
    which holds by Lemma~\ref{lem:cliqueedge}.
    
    Furthermore, if $\min_{e \in \edgeset_\graphh} \rank(e) > 2$, then by Lemma~\ref{lem:cliqueedge} the inequality is strict unless every edge $e \in \edgeset_\graphh$ contains at most one $v \in e$ with $\vecf(v) \neq 0$.
    Suppose the inequality is not strict, then it must be that
    \begin{align*}
        \lambda_1(\degm_\graphg^{-1} \signlapg) & = \frac{\sum_{(u, v) \in \edgeset_\graphg} \weight_\graphg(u, v)(\vecf(v) + \vecf(u))^2}{\sum_{v \in \vertexset_\graphg} \deg_\graphg(v) \vecf(v)^2} \\
        & =  \frac{\sum_{v \in \vertexset_\graphg} \deg_\graphg(v) \vecf(v)^2}{\sum_{v \in \vertexset_\graphg} \deg_\graphg(v) \vecf(v)^2} \\
        & = 1,
    \end{align*}
    since for every edge $(u, v) \in \edgeset_\graphg$, at most one of $\vecf(u)$ or $\vecf(v)$ is not equal to $0$.
    This cannot be the case, since it is a well known fact that the maximum eigenvalue $\lambda_n(\degm_\graphg^{-1} \signlapg) = 2$ and so $\sum_{i = 1}^n \lambda_i(\degm_\graphg^{-1} \signlapg) \geq (n - 1) + 2 = n + 1$ which contradicts the fact that the trace $\tr(\degm_\graphg^{-1} \signlapg)$ is equal to $n$. This proves the final statement of the lemma.
\end{proof}

\subsection{Cheeger-type inequality for hypergraph bipartiteness}
Before proving Lemma~\ref{lem:rulesimplydiffusion}, we will prove some intermediate facts about the new hypergraph operator $\signlaph$ which will allow us to show that the operator $\signlaph$ has a well-defined minimum eigenvector.
Given a hypergraph $\graphh = (\vertexset_\graphh, \edgeset_\graphh, \weight)$, for any edge $e \in \edgeset_\graphh$ and weighted measure vector $\vecf_t$, let $r_e^{\sets} \triangleq \max_{v \in \sfte} \{\vecr(v)\}$ and $r_e^{\set{I}} \triangleq \min_{v \in \ifte} \{\vecr(v)\}$,  and recall that $c_{\vecf_t}(e) = \weight(e) \abs{\discrep_{\vecf_t}(e)}$.
\begin{lemma} \label{lem:rnorm}
Given a hypergraph $\graphh = (\vertexset_\graphh, \edgeset_\graphh, \weight)$ and normalised measure vector $\vecf_t$, let
\[
    \vecr = \dfdt = - \degm_\graphh^{-1} \signlaph \vecf_t.
\]
Then, it holds that 
    \[
        \norm{\vecr}_\weight^2 = - \sum_{e \in \edgeset} \weight(e) \discrep_{\vecf_t}(e) \left(r_e^{\sets} + r_e^{\set{I}}\right).
    \]
\end{lemma}
\begin{proof}
Let $\setp \subset \vertexset$ be one of the densest vertex sets defined with respect to the solution of the linear program~\eqref{eq:lp}. By the description of Algorithm~\ref{algo:computechangerate}, we have $\vecr(u)=\delta(\setp)$ for every $u\in \setp$, and therefore
\begin{align*}
    \sum_{u \in \setp} \deg(u) \vecr(u)^2 & = \vol(\setp)\cdot \delta(\setp)^2 \\
    & = \left(c_{\vecf_t}\left(\spp{U}{P}\right) + c_{\vecf_t}\left(\ipp{U}{P}\right) - c_{\vecf_t}\left(\spm{U}{P}\right) - c_{\vecf_t}\left(\ipm{U}{P}\right)\right)\cdot \delta(\setp) \\
     & = \left( \sum_{e\in\spp{U}{P}} c_{\vecf_t}(e) + \sum_{e\in\ipp{U}{P}} c_{\vecf_t}(e) - \sum_{e\in\spm{U}{P}} c_{\vecf_t}(e) - \sum_{e\in\ipm{U}{P}} c_{\vecf_t}(e) \right)\cdot\delta(\setp)\\
     & =   \sum_{e\in\spp{U}{P}} c_{\vecf_t}(e) \cdot r_e^{\sets} + \sum_{e\in\ipp{U}{P}} c_{\vecf_t}(e)\cdot r_e^{\set{I}} - \sum_{e\in\spm{U}{P}} c_{\vecf_t}(e)\cdot r_e^{\sets} - \sum_{e\in\ipm{U}{P}} c_{\vecf_t}(e)\cdot r_e^{\set{I}}.
\end{align*}
Since each vertex is included in exactly one set $\setp$ and each edge will appear either in one each of $\spp{U}{P}$ and $\ipp{U}{P}$ or in one each of $\spm{U}{P}$ and $\ipm{U}{P}$, it holds that 
    \begin{align*}
    \|\vecr\|^2_\weight & = 
        \sum_{v \in \vertexset} \deg(v) \vecr(v)^2\\
        & = \sum_{\setp} \sum_{v\in \setp} \deg(v) \vecr(v)^2\\
        & = \sum_{\setp} \left( \sum_{e\in\spp{U}{P}} c_{\vecf_t}(e) \cdot r_e^\sets + \sum_{e\in\ipp{U}{P}} c_{\vecf_t}(e)\cdot r_e^{\set{I}} - \sum_{e\in\spm{U}{P}} c_{\vecf_t}(e)\cdot r_e^\sets - \sum_{e\in\ipm{U}{P}} c_{\vecf_t}(e)\cdot r_e^{\set{I}}\right) \\
         & = - \sum_{e \in \edgeset} \weight(e) \discrep_{\vecf_t}(e) \left(r_e^\sets + r_e^{\set{I}}\right),
    \end{align*}
    which proves the lemma.
\end{proof}

Next, we define $\gamma_1 = \min_\vecf D(\vecf)$ and show that any vector $\vecf$ that satisfies $\gamma_1 = D(\vecf)$  is an eigenvector of $\signlaph$ with eigenvalue $\gamma_1$. 
We start by showing that the Raleigh quotient of the new operator is equivalent to the discrepancy ratio of the hypergraph.
\begin{lemma} \label{lem:rqdisc}
    For any hypergraph $\graphh$ and vector $\vecf_t \in \R^n$, it holds that $D(\vecf_t) = R_{\signlaph}(\vecf_t)$.
\end{lemma}
\begin{proof}
    Since $\vecf_t^\transpose \degm_\graphh \vecf_t = \sum_{v \in \vertexset} \deg(v) \vecf_t(v)^2$, it is sufficient to show that
    \[
        \vecf_t^\transpose \signlaph \vecf_t = \sum_{e \in \edgeset_\graphh} \weight(e) \left(\max_{u \in e} \vecf_t(u) + \min_{v \in e} \vecf_t(v)\right)^2.
    \]
    Recall that for some graph $\graphg_t$, $\signlaph = \degm_{\graphg_t} + \adj_{\graphg_t}$.
    Then, we have that
    \begin{align*}
        \vecf_t^\transpose \signlaph \vecf_t & = \vecf_t^\transpose (\degm_{\graphg_t} + \adj_{\graphg_t}) \vecf_t \\
        & = \sum_{(u, v) \in \edgeset_\graphg} \weight_\graphg(u, v) (\vecf_t(u) + \vecf_t(v))^2 \\
        & = \sum_{e \in \edgeset_\graphh} \left(\sum_{(u, v) \in \sfte \times \ifte} \weight_{\graphg_t}(u, v) (\vecf_t(u) + \vecf_t(v))^2\right) \\
        & = \sum_{e \in \edgeset_\graphh} \weight(e) \left(\max_{u \in e} \vecf_t(u) + \min_{v \in e} \vecf_t(v)\right)^2, 
    \end{align*}
    which follows since the graph $\graphg_{t}$ is constructed by splitting the weight of each hyperedge $e \in \edgeset_\graphh$ between the edges $\sfte \times \ifte$.
\end{proof}

\begin{lemma} \label{lem:derivatives}
    For a hypergraph $\graphh$, operator $\signlaph$, and vector $\vecf_t$, the following statements hold:
    \begin{enumerate}
        \item $\frac{\mathrm{d}}{\mathrm{d} t} \norm{\vecf_t}_\weight^2 = -2 \vecf_t^\transpose \signlaph \vecf_t$;
        \item $\frac{\mathrm{d}}{\mathrm{d} t} (\vecf_t^\transpose \signlaph \vecf_t) = -2 \norm{\degm_\graphh^{-1} \signlaph \vecf_t}_\weight^2$;
        \item $\frac{\mathrm{d}}{\mathrm{d} t} R(\vecf_t) \leq 0$ with equality if and only if $\degm_\graphh^{-1} \signlaph \vecf_t \in \mathrm{span}(\vecf_t)$.
    \end{enumerate}
\end{lemma}
\begin{proof}
 
By definition, we have that 
\begin{align*}
    \frac{\mathrm{d} \norm{\vecf_t}_\weight^2}{\mathrm{d}t} & = \frac{\mathrm{d}}{\mathrm{d}t} \sum_{v \in \vertexset} \deg(v) \vecf_t(v)^2 \\
    & = \sum_{v \in \vertexset} \deg(v)\cdot \frac{\mathrm{d} \vecf_t(v)^2}{\mathrm{d} t} \\
    & = \sum_{v \in \vertexset} \deg(v)\cdot \frac{\mathrm{d} \vecf_t(v)^2}{\mathrm{d} \vecf_t(v)} \frac{\mathrm{d} \vecf_t(v)}{\mathrm{d}t} \\
    & = 2 \sum_{v \in \vertexset} \deg(v) \vecf_t(v)\cdot \frac{\mathrm{d} \vecf_t(v)}{\mathrm{d}t} \\
    & = 2 \inner{\vecf_t}{\dfdt}_\weight = -2 \inner{\vecf_t}{\degm_\graphh^{-1} \signlaph \vecf_t}_\weight,
\end{align*}
which proves the first statement.

For the second statement, by Lemma~\ref{lem:rqdisc} we have
\[
\vecf_t^\transpose \signlaph \vecf_t = \sum_{e \in \edgeset} \weight(e) \left(\max_{u \in e} \vecf_t(u) + \min_{v \in e} \vecf_t(v)\right)^2,% = \sum_{e \in E} w(e) \Delta_{f_t}(e)^2,
\]
and therefore 
\begin{align}
    \frac{\mathrm{d}}{\mathrm{d}t} \vecf_t^\transpose \signlaph \vecf_t & = \frac{\mathrm{d}}{\mathrm{d}t} \sum_{e \in \edgeset} \weight(e) \left(\max_{u \in e} \vecf_t(u) + \min_{v \in e} \vecf_t(v) \right)^2 \nonumber\\
    & = 2 \sum_{e \in \edgeset} \weight(e) \discrep_{\vecf_t}(e)\cdot \frac{\mathrm{d}}{\mathrm{d}t} \left(\max_{u \in e} \vecf_t(u) + \min_{v \in e} \vecf_t(v)\right)\nonumber\\
    &= 2 \sum_{e \in \edgeset} \weight(e) \discrep_{\vecf_t}(e)\cdot  \left(r_e^{\sets} + r_e^{\set{I}} \right), \label{eq:calculation1}
\end{align}
where the last equality holds by the way that all the vertices receive their $\vecr$-values by the algorithm and the definitions of $r_e^{\sets}$ and $r_e^{\set{I}}$.  On the other side, 
by definition $\vecr=- \degm_\graphh^{-1} \signlaph \vecf_t$ and so by Lemma~\ref{lem:rnorm},
\begin{equation}\label{eq:calculation2}
\norm{\degm_\graphh^{-1} \signlaph \vecf_t}_\weight^2 = 
\norm{\vecr}_\weight^2 = - \sum_{e\in \edgeset} \weight(e)\discrep_{\vecf_t}(e) \left( r_e^{\sets}+r_e^{\set{I}}
\right).
\end{equation}
By combining \eqref{eq:calculation1} with \eqref{eq:calculation2}, we have the second statement.  

 For the third statement, notice that we can write $\vecf_t^\transpose \signlaph \vecf_t$ as $\inner{\vecf_t}{\degm_\graphh^{-1} \signlaph \vecf_t}_\weight$. Then, we have that
\begin{align*}
\frac{\mathrm{d}}{\mathrm{d}t} \frac{ \langle \vecf_t, \degm_\graphh^{-1} \signlaph \vecf_t\rangle_\weight}{ \norm{\vecf_t}_\weight^2} & = \frac{1}{ \norm{\vecf_t}_\weight^2} \cdot \frac{\mathrm{d}~\langle \vecf_t, \degm_\graphh^{-1} \signlaph \vecf_t\rangle_\weight}{\mathrm{d}t}    -  \langle \vecf_t, \degm_\graphh^{-1} \signlaph \vecf_t\rangle_\weight \cdot   \frac{1}{ \norm{\vecf_t}_\weight^4} \frac{\mathrm{d}~\norm{\vecf_t}^2_\weight}{\mathrm{d}t} \\
& = -\frac{1}{\norm{\vecf_t}_\weight^4}\cdot \left(2\cdot \norm{\vecf_t}_\weight^2\cdot  \norm{\degm_\graphh^{-1} \signlaph \vecf_t}_\weight^2 + \left\langle \vecf_t, \degm_\graphh^{-1} \signlaph \vecf_t\right\rangle_\weight  \cdot \frac{\mathrm{d}~\norm{\vecf_t}^2_\weight}{\mathrm{d}t} \right) \\
& = - \frac{2}{\norm{\vecf_t}_\weight^4}\cdot \left( \norm{\vecf_t}_\weight^2\cdot  \norm{\degm_\graphh^{-1} \signlaph \vecf_t}_\weight^2 -\left\langle \vecf_t, \degm_\graphh^{-1} \signlaph \vecf_t\right\rangle_\weight^2  \right)\\
& \leq 0,
\end{align*}
where the last inequality holds by the Cauchy-Schwarz inequality on the inner product $\inner{\cdot}{\cdot}_\weight$ with the equality if and only if $\degm_\graphh^{-1} \signlaph \vecf_t\in\mathrm{span}(\vecf_t)$.
\end{proof}

This allows us to establish the following key fact.

\begin{lemma} \label{lem:eigenvalueexists}
    For any hypergraph $\graphh$, $\gamma_1 = \min_{\vecf} D(\vecf)$ is an eigenvalue of $\signlaph$ and any minimiser $\vecf$ is its corresponding eigenvector.
\end{lemma}
\begin{proof}
By Lemma~\ref{lem:rqdisc}, it holds that $R(\vecf)= D(\vecf)$ for any $\vecf\in\R^n$. When $\vecf$ is a minimiser of $D(\vecf)$, it must hold that
\[
\frac{\mathrm{d} R(\vecf)}{\mathrm{d}t}=0,
\]
which implies by Lemma~\ref{lem:derivatives} that $\degm_\graphh^{-1} \signlaph \vecf\in\mathrm{span}(\vecf)$ and proves that $\vecf$ is an eigenvector.
\end{proof}

We are now able to prove a Cheeger-type inequality for our operator and the hypergraph bipartiteness.
We first show that the minimum eigenvalue of $\signlaph$ is at most twice the hypergraph bipartiteness.

\begin{lemma} \label{lem:trev_cheeg_easy}
    Given a hypergraph $\graphh = (\vertexset_\graphh, \edgeset_\graphh)$ with sets $\setl, \setr \subset \vertexset_\graphh$ such that $\bipart(\setl, \setr) = \bipart$, it is the case that
    \[
        \gamma_1 \leq 2 \bipart
    \]
    where $\gamma_1$ is the smallest eigenvalue of $\signlaph$.
\end{lemma}
\begin{proof}
    Let $\vec{\chi}_{\setl, \setr}\in\{-1,0,1\}^n$ be the indicator vector of the cut $\setl, \setr$ such that
    \[
        \vec{\chi}_{\setl, \setr}(u) = \threepartdefow{1}{u \in \setl}{-1}{u \in \setr}{0}.
    \]
    Then, by Lemma~\ref{lem:rqdisc}, the Rayleigh quotient is given by
    \begin{align*}
        R_{\signlaph}\left(\vec{\chi}_{\setl,\setr}\right) & = \frac{\sum_{e \in \edgeset} \weight(e) \left(\max_{u \in e} \vec{\chi}_{\setl,\setr}(u) + \min_{v \in e} \vec{\chi}_{\setl,\setr}(v)\right)^2}{\sum_{v \in \vertexset} \deg(v) \vec{\chi}_{\setl,\setr}(v)^2} \\
        & = \frac{4 \weight(\setl | \setcomplement{\setl}) + 4 \weight(\setr | \setcomplement{\setr}) + \weight(\setl, \setcomplement{\lur} | \setr) + \weight(\setr, \setcomplement{\lur} | \setl)}{\vol(\lur)} \\
        & \leq 2 \bipart(\setl, \setr),
    \end{align*}
    which completes the proof since $\gamma_1 = \min_{\vecf} R_{\signlaph}(\vecf)$ by Lemma~\ref{lem:eigenvalueexists}.
\end{proof}

The proof of the other direction of the Cheeger-type inequality is more involved, and forms the basis of the second claim in Theorem~\ref{thm:mainalg}.

\begin{theorem} \label{thm:trev-cheeg}
For any hypergraph $\graphh =(\vertexset,\edgeset,\weight)$, let $\signlaph$ be the operator defined with respect to $\graphh$, and $\gamma_1 = \min_{\vecf} D(\vecf) = \min_{\vecf} R(\vecf)$ be the minimum eigenvalue of $\degm_\graphh^{-1} \signlaph$. Then, there are disjoint sets $\setl,\setr \subset \vertexset$ such that
\[
\bipart(\setl,\setr)\leq\sqrt{2\gamma_1}.
\]
\end{theorem}
\begin{proof}
    \newcommand{\maxfu}{a_e}
    \newcommand{\minfu}{b_e}
    % Start the proof
    Let $\vecf\in\R^n$ be the vector such that $D(\vecf)=\gamma_1$.
    For any threshold $t \in [0, \max_u \vecf(u)^2]$, define $\vecx_t$ such that
    \[
        \vecx_t(u) = \threepartdefow{1}{\vecf(u) \geq \sqrt{t}}{-1}{\vecf(u) \leq -\sqrt{t}}{0}.
    \]
    We will show that choosing $t \in \left[0, \max_u \vecf(u)^2\right]$ uniformly at random gives
    \begin{equation} \label{eq:expectation_goal}
        \E\left[\sum_{e \in \edgeset_\graphh} \weight(e) \deltaeabs{\vecx_t} \right] \leq \sqrt{2 \gamma_1}\cdot \E\left[\sum_{v \in \vertexset_\graphh} \deg(v) \abs{\vecx_t(v)} \right].
    \end{equation}
    Notice that every such $\vecx_t$ defines disjoints vertex sets $\setl$ and $\setr$ that satisfy
     \[
        \bipart(\setl, \setr) = \frac{\sum_{e \in \edgeset_\graphh}\weight(e) \deltaeabs{\vecx_t}}{\sum_{v \in \vertexset_\graphh} \deg(v) \abs{\vecx_t(v)}}.
    \]
    Hence, \eqref{eq:expectation_goal} would imply that there is some $\vecx_t$ such that the disjoint $\setl,\setr$ defined by $\vecx_t$ would satisfy
    \[
     \sum_{e \in \edgeset_\graphh}\weight(e) \deltaeabs{\vecx_t}   \leq \sqrt{2 \gamma_1}\cdot  \sum_{v \in \vertexset_\graphh} \deg(v) \abs{\vecx_t(v)},
    \]
    which implies that
    \[
    \bipart(\setl,\setr)\leq\sqrt{2\gamma_1}.
    \]
    Hence, it suffices to prove \eqref{eq:expectation_goal}.
        We assume without loss of generality that $\max_u\left\{\vecf(u)^2\right\}=1$, so $t$ is chosen uniformly from $[0,1]$.
    First of all, we have that
    \begin{align*}
        \E \left[\sum_{v \in \vertexset_\graphh} \deg(v) \abs{\vecx_t(v)} \right] = \sum_{v \in \vertexset_\graphh} \deg(v) \E\left[\abs{\vecx_t(v)}\right] = \sum_{v \in \vertexset_\graphh} \deg(v) \vecf(v)^2.
    \end{align*}
    To analyse the left-hand side of \eqref{eq:expectation_goal}, we will focus on a particular edge $e \in \edgeset_\graphh$. Let $\maxfu = \max_{u \in e} \vecf(u)$ and $\minfu = \min_{v \in e} \vecf(v)$.
    We will drop the subscript $e$ when it is clear from context.
    We will show that
    \begin{equation} \label{eq:edge_expectation}
        \E\left[\deltaeabs{\vecx_t}\right] \leq \abs{\maxfu + \minfu}(\abs{\maxfu} + \abs{\minfu}).
    \end{equation}
    \begin{enumerate}
        \item Suppose $\sign(\maxfu) = \sign(\minfu)$. Our analysis is by case distinction:
        \begin{itemize}
            \item $\deltaeabs{\vecx_t} = 2$ with probability $\min(\maxfu^2, \minfu^2)$;
            \item $\deltaeabs{\vecx_t} = 1$ with probability $\abs{\maxfu^2 - \minfu^2}$;
            \item $\deltaeabs{\vecx_t} = 0$ with probability $1 - \max(\maxfu^2, \minfu^2)$.
        \end{itemize}
        Assume without loss of generality that $\maxfu^2 = \min\left(\maxfu^2, \minfu^2\right)$. Then, it holds hat 
        \[\E \left[\deltaeabs{\vecx}\right] = 2 \maxfu^2 + \abs{\maxfu^2 - \minfu^2} = \maxfu^2 + \minfu^2 \leq \abs{\maxfu + \minfu}(\abs{\maxfu} + \abs{\minfu}).\]
        \item Suppose $\sign(\maxfu) \neq \sign(\minfu)$. Our analysis is by case distinction: 
        \begin{itemize}
            \item $\deltaeabs{\vecx_t} = 2$ with probability $0$;
            \item $\deltaeabs{\vecx_t} = 1$ with probability $\abs{\maxfu^2 - \minfu^2}$;
            \item $\deltaeabs{\vecx_t} = 0$ with probability $\min(\maxfu^2, \minfu^2)$.
        \end{itemize}
        Assume without loss of generality that $\maxfu^2 = \min\left(\maxfu^2, \minfu^2\right)$. Then, it holds that 
        \[\E \left[\deltaeabs{\vecx_t}\right] = \abs{\maxfu^2 - \minfu^2} = (\abs{\maxfu} - \abs{\minfu})(\abs{\maxfu} + \abs{\minfu}) = \abs{\maxfu + \minfu}(\abs{\maxfu} + \abs{\minfu}),\]
        where the final equality follows because $\maxfu$ and $\minfu$ have different signs.
    \end{enumerate}
    These two cases establish \eqref{eq:edge_expectation}.
    Now, we have that 
    \begin{align*}
        \E & \left[\sum_{e \in \edgeset_\graphh} \weight(e) \deltaeabs{\vecx_t}\right] \leq \sum_{e \in \edgeset_\graphh} \weight(e) \abs{\maxfu + \minfu}(\abs{\maxfu} + \abs{\minfu}) \\
        & \leq \sqrt{\sum_{e \in \edgeset_\graphh}\weight(e)\abs{\maxfu + \minfu}^2} \sqrt{\sum_{e \in \edgeset_\graphh}\weight(e) (\abs{\maxfu} + \abs{\minfu})^2} \\
        & = \sqrt{\sum_{e \in \edgeset_\graphh} \weight(e) \deltaesquare{\vecf}} \sqrt{\sum_{e \in \edgeset_\graphh}\weight(e) (\abs{\maxfu} + \abs{\minfu})^2}.
    \end{align*}
    By our assumption that $\vecf$ is the eigenvector corresponding to the eigenvalue $\gamma_1$, it holds that
    \[
        \sum_{e \in \edgeset_\graphh}\weight(e) \deltaesquare{\vecf} \leq \gamma_1 \sum_{v \in \vertexset_\graphh}\deg(v) \vecf(v)^2.
    \]
    On the other side, we have that 
    \[
        \sum_{e \in \edgeset_\graphh}\weight(e)(\abs{\maxfu} + \abs{\minfu})^2 \leq 2 \sum_{e \in \edgeset_\graphh}\weight(e)(\abs{\maxfu}^2 + \abs{\minfu}^2) \leq 2 \sum_{v \in \vertexset_\graphh}\deg(v) \vecf(v)^2.
    \]
    This gives us that 
    \begin{align*}
        \E \left[\sum_{e \in \edgeset_\graphh}\weight(e) \deltaeabs{\vecx}\right] & \leq \sqrt{2 \gamma_1} \sum_{v \in \vertexset_\graphh}\deg(v) \vecf(v)^2 = \sqrt{2 \gamma_1}\cdot  \E \left[\sum_{v \in \vertexset_\graphh}\deg(v) \abs{\vecx(v)}\right],
    \end{align*}
    which proves the statement.
\end{proof}
 Now, we are able to combine these results to prove Theorem~\ref{thm:mainalg}.
\begin{proof}[Proof of Theorem~\ref{thm:mainalg}.]
    The first statement of the theorem follows by setting the starting vector $\vecf_0$ of the diffusion to be the minimum eigenvector of the clique graph $\graphg$. By Lemma~\ref{lem:bettereigenvalue}, we have that $R_{\signlaph}(\vecf_0) \leq \lambda_1(\degm_\graphg^{-1} \signlapg)$ and the inequality is strict if $\min_{e \in \edgeset_\graphh} \rank(e) > 2$.
    Then, Theorem~\ref{thm:convergence} shows that the diffusion process converges to an eigenvector and that the Rayleigh quotient only decreases during convergence, and so the inequality holds.
    The algorithm runs in polynomial time since, by Theorem~\ref{thm:convergence}, the diffusion process converges in polynomial time, and each step of Algorithm~\ref{algo:main} can be computed in polynomial time using a standard algorithm for solving the linear programs.
 
    The second statement is a restatement of Theorem~\ref{thm:trev-cheeg}. The sweep set algorithm runs in polynomial time since there are $n$ different sweep sets, and computing the hypergraph bipartiteness for each one also takes only polynomial time.
\end{proof}

\subsection{Further discussion on the number of eigenvectors of \texorpdfstring{$\signlaph$}{Jh} and \texorpdfstring{$\lap_\graphh$}{Lh}}
In this subsection, we investigate the spectrum of the non-linear hypergraph Laplacian operator introduced in~\cite{chanSpectralPropertiesHypergraph2018} and our new operator $\signlaph$ by considering some example hypergraphs.
In particular, we show that the hypergraph Laplacian $\lap_\graphh$ can have more than $2$ eigenvalues, which answers an open question of Chan et al.\ \cite{chanSpectralPropertiesHypergraph2018}.
Furthermore, we show that the new operator $\signlaph$ can have an exponentially large number of eigenvectors.

\paragraph{The hypergraph Laplacian $\lap_\graphh$ can have more than $2$ eigenvalues.}
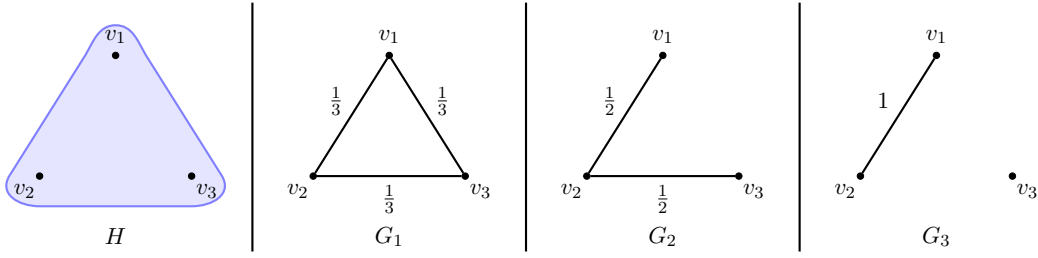
\begin{figure}[t]
    \centering
    \scalebox{0.8}{
    \begin{tikzpicture}
	\begin{pgfonlayer}{nodelayer}
		\node [style=smallblack] (0) at (-9, 2) {};
		\node [style=smallblack] (1) at (-11.5, -2) {};
		\node [style=smallblack] (2) at (-6.5, -2) {};
		\node [style=none] (3) at (-9, 2.625) {$v_1$};
		\node [style=none] (4) at (-12, -2.5) {$v_2$};
		\node [style=none] (5) at (-6, -2.5) {$v_3$};
		\node [style=none] (6) at (-10, 2) {};
		\node [style=none] (7) at (-9, 3) {};
		\node [style=none] (8) at (-8, 2) {};
		\node [style=none] (9) at (-5.5, -2) {};
		\node [style=none] (11) at (-6.5, -3) {};
		\node [style=none] (12) at (-11.5, -3) {};
		\node [style=none] (13) at (-12.5, -2) {};
		\node [style=none] (14) at (-9, -4) {$H$};
		\node [style=smallblack] (15) at (0, 2) {};
		\node [style=smallblack] (16) at (-2.5, -2) {};
		\node [style=smallblack] (17) at (2.5, -2) {};
		\node [style=none] (18) at (0, 2.625) {$v_1$};
		\node [style=none] (19) at (-3, -2.5) {$v_2$};
		\node [style=none] (20) at (3, -2.5) {$v_3$};
		\node [style=none] (28) at (0, -4) {$G_1$};
		\node [style=smallblack] (29) at (9, 2) {};
		\node [style=smallblack] (30) at (6.5, -2) {};
		\node [style=smallblack] (31) at (11.5, -2) {};
		\node [style=none] (32) at (9, 2.625) {$v_1$};
		\node [style=none] (33) at (6, -2.5) {$v_2$};
		\node [style=none] (34) at (12, -2.5) {$v_3$};
		\node [style=none] (42) at (9, -4) {$G_2$};
		\node [style=smallblack] (43) at (18, 2) {};
		\node [style=smallblack] (44) at (15.5, -2) {};
		\node [style=smallblack] (45) at (20.5, -2) {};
		\node [style=none] (46) at (18, 2.625) {$v_1$};
		\node [style=none] (47) at (15, -2.5) {$v_2$};
		\node [style=none] (48) at (21, -2.5) {$v_3$};
		\node [style=none] (56) at (18, -4) {$G_3$};
		\node [style=none] (57) at (-1.75, 0.5) {$\frac{1}{3}$};
		\node [style=none] (58) at (0, -2.75) {$\frac{1}{3}$};
		\node [style=none] (59) at (1.75, 0.5) {$\frac{1}{3}$};
		\node [style=none] (60) at (7.25, 0.5) {$\frac{1}{2}$};
		\node [style=none] (61) at (9, -2.75) {$\frac{1}{2}$};
		\node [style=none] (62) at (16.25, 0.5) {$1$};
		\node [style=none] (63) at (-4.5, 3.75) {};
		\node [style=none] (64) at (-4.5, -4.5) {};
		\node [style=none] (65) at (4.5, -4.5) {};
		\node [style=none] (66) at (4.5, 3.75) {};
		\node [style=none] (67) at (13.5, 3.75) {};
		\node [style=none] (68) at (13.5, -4.5) {};
	\end{pgfonlayer}
	\begin{pgfonlayer}{edgelayer}
		\draw [style=blue filled] (6.center)
			 to (13.center)
			 to [in=180, out=-120, looseness=1.25] (12.center)
			 to (11.center)
			 to [in=-60, out=0, looseness=1.25] (9.center)
			 to (8.center)
			 to [in=0, out=120] (7.center)
			 to [in=60, out=180] cycle;
		\draw [style=undirected] (15) to (16);
		\draw [style=undirected] (16) to (17);
		\draw [style=undirected] (17) to (15);
		\draw [style=undirected] (29) to (30);
		\draw [style=undirected] (30) to (31);
		\draw [style=undirected] (43) to (44);
		\draw [style=undirected] (63.center) to (64.center);
		\draw [style=undirected] (66.center) to (65.center);
		\draw [style=undirected] (67.center) to (68.center);
	\end{pgfonlayer}
\end{tikzpicture}
    }
\caption{\small{Given the hypergraph $\graphh$, there are three graphs $\graphg_1$, $\graphg_2$, and $\graphg_3$ which correspond to eigenvalues of the hypergraph Laplacian $\lap_\graphh$.}
\label{fig:laplacian_eigenvalues}}
\end{figure}
Similar to our new operator $\signlaph$, for some vector $\vecf$, the operator $\lap_\graphh$ behaves like the graph Laplacian $\lap_\graphg = \degm_\graphg - \adj_\graphg$ for some graph $\graphg$ constructed by splitting the weight of each hyperedge $e$ between the edges in $\sfe \times \ife$.
We refer the reader to~\cite{chanSpectralPropertiesHypergraph2018} for the full details of this construction.

Now, with reference to Figure~\ref{fig:laplacian_eigenvalues}, we consider the simple hypergraph $\graphh$ where $\vertexset_\graphh = \{v_1, v_2, v_3\}$ and $\edgeset_\graphh = \{\{v_1, v_2, v_3\}\}$.
Letting $\vecf_1 = [1, 1, 1]^\transpose$, notice that the graph $\graphg_1$ in Figure~\ref{fig:laplacian_eigenvalues} is the graph constructed such that $\lap_\graphh \vecf_1$ is equivalent to $\lap_{\graphg_1} \vecf_1$.
In this case,
\[
\renewcommand\arraystretch{1.6}
    \lap_{\graphg_1} = \begin{bmatrix}
        \frac{2}{3} & -\frac{1}{3} & -\frac{1}{3} \\
        -\frac{1}{3} & \frac{2}{3} & -\frac{1}{3} \\
        -\frac{1}{3} & -\frac{1}{3} & \frac{2}{3}
    \end{bmatrix}
\]
and we have
\[
\renewcommand\arraystretch{1.6}
    \degm_\graphh^{-1} \lap_\graphh \vecf_1 = \lap_{\graphg_1} \vecf_1 = [0, 0, 0]^\transpose 
\]
which shows that $\vecf_1$ is the trivial eigenvector of $\degm_{\graphh}^{-1} \lap_\graphh$ with eigenvalue $0$.

Now, consider $\vecf_2 = [1, -2, 1]^\transpose$. In this case, $\graphg_2$ shown in Figure~\ref{fig:laplacian_eigenvalues} is the graph constructed such that $\lap_\graphh \vecf_2$ is equivalent to $\lap_{\graphg_2} \vecf_2$.
Then, we have
\[
\renewcommand\arraystretch{1.6}
    \lap_{\graphg_2} = \begin{bmatrix}
        \frac{1}{2} & -\frac{1}{2} & 0 \\
        -\frac{1}{2} & 1 & -\frac{1}{2} \\
        0 & -\frac{1}{2} & \frac{1}{2}
    \end{bmatrix}
\]
and
\[
    \degm_{\graphh}^{-1} \lap_\graphh \vecf_2 = \lap_{\graphg_2} \vecf_2 = \left[\frac{3}{2}, -3, \frac{3}{2} \right]^\transpose
\]
and so $\vecf_2$ is an eigenvector of $\degm_\graphh^{-1} \lap_\graphh$ with eigenvalue $3/2$.

Finally, we consider $\vecf_3 = [1, -1, 0]$ and notice that the graph $\graphg_3$ is the constructed graph.
Then,
\[
\renewcommand\arraystretch{1.6}
    \lap_{\graphg_3} = \begin{bmatrix}
        \frac{1}{2} & -\frac{1}{2} & 0 \\
        - \frac{1}{2} & \frac{1}{2} & 0 \\
        0 & 0 & 0
    \end{bmatrix}
\]
and
\[
    \degm_{\graphh}^{-1} \lap_\graphh \vecf_3 = \lap_{\graphg_3} \vecf_3 = \left[2, -2, 0 \right]^\transpose
\]
which shows that $\vecf_3$ is an eigenvector of $\lap_\graphh$ with eigenvalue $2$.

Through an exhaustive search of the other possible constructed graphs on $\{v_1, v_2, v_3\}$, we find that these are the only eigenvalues.
By the symmetries of $\vecf_2$ and $\vecf_3$ this means that the operator $\lap_\graphh$ has a total of $7$ different eigenvectors and $3$ distinct eigenvalues.
It is useful to point out that, since the $\lap_\graphh$ operator is non-linear, a linear combination of eigenvectors with the same eigenvalue is \emph{not}, in general, an eigenvector.
Additionally, notice that the constructions of $\vecf_1$ and $\vecf_2$ generalise to hypergraphs with more than $3$ vertices which shows that the number of eigenvectors of $\lap_\graphh$ can grow faster than the number of vertices.
This example answers an open question in~\cite{chanSpectralPropertiesHypergraph2018} which showed that there are always two eigenvalues and asked whether there can be any more, although further investigation of the spectrum of this operator would be very interesting.

\paragraph{The hypergraph operator $\signlap_\graphh$ can have an exponential number of eigenvectors.}
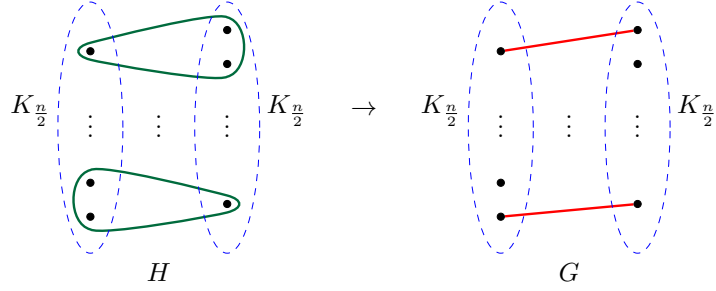
\begin{figure}[t]
    \centering
    \scalebox{0.9}{
    \begin{tikzpicture}
	\begin{pgfonlayer}{nodelayer}
		\node [style=smallblack] (4) at (-5, 2.875) {};
		\node [style=smallblack] (6) at (-5, 1.875) {};
		\node [style=none] (15) at (-5, 1.5) {};
		\node [style=none] (17) at (-5, 3.25) {};
		\node [style=none] (18) at (-9, 2.5) {};
		\node [style=none] (19) at (-9, 2) {};
		\node [style=smallblack] (20) at (-9, 2.25) {};
		\node [style=none] (29) at (-10.75, 0.5) {$K_{\frac{n}{2}}$};
		\node [style=none] (30) at (-3.25, 0.5) {$K_{\frac{n}{2}}$};
		\node [style=none] (31) at (-7, -4.25) {$H$};
		\node [style=smallblack] (90) at (-9, -1.625) {};
		\node [style=smallblack] (91) at (-9, -2.625) {};
		\node [style=none] (92) at (-9, -3) {};
		\node [style=none] (93) at (-9, -1.25) {};
		\node [style=none] (94) at (-5, -2) {};
		\node [style=none] (95) at (-5, -2.5) {};
		\node [style=smallblack] (96) at (-5, -2.25) {};
		\node [style=big vertical ellipse dashed] (101) at (-5, 0) {};
		\node [style=big vertical ellipse dashed] (102) at (-9, 0) {};
		\node [style=none] (103) at (-7, 0.25) {$\vdots$};
		\node [style=smallblack] (104) at (7, 2.875) {};
		\node [style=smallblack] (105) at (7, 1.875) {};
		\node [style=smallblack] (110) at (3, 2.25) {};
		\node [style=none] (111) at (1.25, 0.5) {$K_{\frac{n}{2}}$};
		\node [style=none] (112) at (8.75, 0.5) {$K_{\frac{n}{2}}$};
		\node [style=none] (113) at (5, -4.25) {$G$};
		\node [style=smallblack] (114) at (3, -1.625) {};
		\node [style=smallblack] (115) at (3, -2.625) {};
		\node [style=smallblack] (120) at (7, -2.25) {};
		\node [style=big vertical ellipse dashed] (121) at (7, 0) {};
		\node [style=big vertical ellipse dashed] (122) at (3, 0) {};
		\node [style=none] (123) at (5, 0.25) {$\vdots$};
		\node [style=none] (124) at (-1, 0.5) {$\rightarrow$};
		\node [style=none] (125) at (-9, 0.25) {$\vdots$};
		\node [style=none] (126) at (-5, 0.25) {$\vdots$};
		\node [style=none] (127) at (3, 0.25) {$\vdots$};
		\node [style=none] (128) at (7, 0.25) {$\vdots$};
	\end{pgfonlayer}
	\begin{pgfonlayer}{edgelayer}
		\draw [style=undirected green] (19.center)
			 to [bend left=75, looseness=2.50] (18.center)
			 to [in=150, out=15, looseness=0.25] (17.center)
			 to [bend left=75] (15.center)
			 to [in=-15, out=-150, looseness=0.25] cycle;
		\draw [style=undirected green] (95.center)
			 to [bend right=75, looseness=2.50] (94.center)
			 to [in=30, out=165, looseness=0.25] (93.center)
			 to [bend right=75] (92.center)
			 to [in=-165, out=-30, looseness=0.25] cycle;
		\draw [style=undirected red] (110) to (104);
		\draw [style=undirected red] (115) to (120);
	\end{pgfonlayer}
\end{tikzpicture}
    }
\caption{\small{Given the hypergraph $\graphh$, there are $2^{n/3}$ possible graphs $\graphg$, each of which corresponds to a different eigenvector of the hypergraph operator $\signlaph$.}
\label{fig:new_operator_eigenvalues}}
\end{figure}
To study the spectrum of our new operator $\signlaph$, we construct a hypergraph $\graphh$ in the following way:
\begin{itemize}
    \item There are $n$ vertices split into two clusters $\setl$ and $\setr$ of size $n/2$. There is a clique on each cluster.
    \item The cliques are joined by $n/3$ edges of rank $3$, such that every vertex is a member of exactly one such edge.
\end{itemize}
Now, we can construct a vector $\vecf$ as follows.
For each edge $e$ of rank $3$ in the hypergraph, let $u \in e$ be the vertex alone in one of the cliques, and $v, w \in e$ be the two vertices in the other clique.
Then, we set $\vecf(u) = 1$ and one of $\vecf(v)$ or $\vecf(w)$ to be $-1$ and the other to be $0$.
Notice that there are $2^{n/3}$ such vectors.
Each one corresponds to a different graph $\graphg$, as illustrated in Figure~\ref{fig:new_operator_eigenvalues}, which is the graph constructed such that $\signlaph$ is equivalent to $\signlapg$ when applied to the vector $\vecf$.

Notice that, by the construction of the graph $\graphh$, half of the edges of rank $3$ must have one vertex in $\setl$ and two vertices in $\setr$ and half of the rank-$3$ edges must have two vertices in $\setl$ and one vertex in $\setr$.
This means that, within each cluster $\setl$ or $\setr$, one third of the vertices have $\vecf$-value $1$, one third have $\vecf$-value $-1$ and one third have $\vecf$-value $0$.

Now, we have that
\begin{align*}
    \left(\degm_\graphh^{-1} \signlaph \vecf\right)(u) & = \left(\degm_\graphh^{-1} \signlapg \vecf\right)(u) \\
    & = \frac{2}{n} \sum_{u \sim_\graphg v} (\vecf(u) + \vecf(v))
\end{align*}
where $u \sim_\graphg v$ means that $u$ and $v$ are adjacent in the graph $\graphg$.
Suppose that $\vecf(u) = 0$, meaning that it does not have an adjacent edge in $\graphg$ from its adjacent rank-$3$ edge in $\graphh$.
In this case,
\begin{align*}
    \left(\degm_\graphh^{-1} \signlaph \vecf\right)(u) & = \frac{2}{n} \sum_{u \sim_\graphg v} \vecf(v) \\
    & = \frac{2}{n} \left[ \frac{n}{3 \cdot 2} \cdot 1 + \frac{n}{3 \cdot 2} \cdot (-1) \right] \\
    & = 0.
\end{align*}
Now, suppose that $\vecf(u) = 1$, and so it has an adjacent edge in $\graphg$ from its adjacent rank-$3$ edge in $\graphh$.
Then,
\begin{align*}
    \left(\degm_\graphh^{-1} \signlaph \vecf\right)(u) & = \frac{2}{n} \sum_{u \sim_\graphg v} (1 + \vecf(v)) \\
    & = \frac{2}{n} \left[ \frac{n}{2} + \frac{n}{3 \cdot 2} \cdot (-1) + \left(\frac{n}{3 \cdot 2} - 1\right) \cdot 1 - 1 \right] \\
    & = 1 - \frac{4}{n}.
\end{align*}
Similarly, if $\vecf(u) = -1$, we find that
\[
    \left(\degm_\graphh^{-1} \signlaph \vecf\right)(u) = \frac{4}{n} - 1,
\]
and so we can conclude that $\vecf$ is an eigenvector of the operator $\signlaph$ with eigenvalue $(n - 4)/n$.
Since there are $2^{n/3}$ possible such vectors $\vecf$, there are an exponential number of eigenvectors with this eigenvalue.
Once again, we highlight that due to the non-linearity of $\signlaph$, a linear combination of these eigenvectors is in general not an eigenvector of the operator.

\section{Experiments} \label{sec:experiments}
In this section, we evaluate the performance of our new algorithm on synthetic and real-world datasets. 
All algorithms are implemented in Python 3.6, using the \texttt{scipy} library for sparse matrix representations and linear programs. The experiments are performed using an Intel(R) Core(TM) i5-8500 CPU @ 3.00GHz processor, with 16 GB RAM. Our code can be downloaded from 
\begin{center}
\href{https://github.com/pmacg/hypergraph-bipartite-components}{https://github.com/pmacg/hypergraph-bipartite-components}.
\end{center}

Since ours is the first proposed algorithm for approximating hypergraph bipartiteness, we will compare it to a simple and natural baseline algorithm, which we call \textsc{CliqueCut} (\textsc{CC}).
In this algorithm, we construct the clique reduction of the hypergraph and 
use the two-sided sweep-set algorithm described in \cite{trevisanMaxCutSmallest2012} to find a set with low bipartiteness in the clique reduction.\footnote{We choose to use the algorithm in \cite{trevisanMaxCutSmallest2012} here since, as far as we know, this  is the only non-\textsf{SDP} based algorithm for solving the MAX-CUT problem for graphs. Notice that, although SDP-based algorithms achieve a better approximation ratio for the MAX-CUT problem, they are not practical even for hypergraphs of medium sizes.}
% \he{to check if we need to rewrite this paragraph after defining the clique reduction in Section 2.}

Additionally, we will compare two versions of our proposed algorithm.  \diffalgname\ (\diffalgshortname) is our new algorithm described in Algorithm~\ref{algo:main} and \diffalgapproxname\ (\diffalgapproxshortname) is an approximate version in which we do not solve the linear programs in Lemma~\ref{lem:rulesimplydiffusion} to compute the graph $\graphg$.
Instead, at each step of the algorithm, we construct $\graphg$ by splitting the weight of each hyperedge $e$ evenly between the edges in $\sets(e) \times \set{I}(e)$.

We always set the parameter $\epsilon = 1$ for \diffalgshortname\ and \diffalgapproxshortname, and we set the starting vector $\vecf_0 \in\R^n$ for the diffusion to be the eigenvector corresponding to the minimum eigenvalue of $\signlapg$, where $\graphg$ is the clique reduction of the hypergraph $\graphh$.

\subsection{Synthetic datasets}
We first evaluate the algorithms using a random hypergraph model.
Given the parameters $n$, $r$, $p$, and $q$, we generate an $n$-vertex $r$-uniform hypergraph in the following way:
the vertex set $\vertexset$ is divided into two clusters $\setl$ and $\setr$ of size $n/2$.
For every set $\sets \subset \vertexset$ with $\cardinality{\sets} = r$, if $\sets \subset \setl$ or $\sets \subset \setr$ we add the hyperedge $\sets$ with probability $p$ and otherwise we add the hyperedge with probability $q$.
 We remark that this is a special case of the hypergraph stochastic block model (e.g., \cite{chienCommunityDetectionHypergraphs2018}). We limit the number of free parameters for simplicity while maintaining enough flexibility to generate random hypergraphs with a wide range of optimal $\bipart_\graphh$-values.

% Describe the metrics on which we will compare the algorithms.
We will compare the algorithms' performance using four metrics: the hypergraph bipartiteness ratio $\bipart_\graphh(\setl, \setr)$, the clique graph bipartiteness ratio $\bipart_\graphg(\setl, \setr)$, the $F_1$-score~\cite{vanrijsbergenGeometryInformationRetrieval2004} of the returned clustering, and the runtime of the algorithm.
Throughout this subsection, we always report the average result on $10$ hypergraphs randomly generated with each parameter configuration.

\paragraph{Comparison of \diffalgshortname\ and \diffalgapproxshortname}
We first fix the values $n = 200$, $r = 3$, and $p = 10^{-4}$ and vary the ratio of $q/p$ from $2$ to $6$ which produces hypergraphs with $250$ to $650$ edges.
The performance of each algorithm on these hypergraphs is shown in Figure~\ref{fig:synthetic_rank3} from which we can make the following observations:
\begin{itemize} \itemsep -2pt
    \item From Figure~\ref{fig:synthetic_rank3}~(a) we observe that \diffalgshortname\ and \diffalgapproxshortname\ find sets with very similar bipartiteness and they perform better than the \textsc{CliqueCut} baseline.
    \item From Figure~\ref{fig:synthetic_rank3}~(b) we can see that our proposed algorithms produce output with a lower
    $\beta_G$-value than the output of the \textsc{CliqueCut} algorithm. This is a surprising result given that \textsc{CliqueCut} operates directly on the clique graph.
    \item Figure~\ref{fig:synthetic_rank3}~(c) shows that the \diffalgapproxshortname\ algorithm is much faster than \diffalgshortname.
\end{itemize}
From these observations, we conclude that in practice it is sufficient to use the much faster \diffalgapproxshortname\ algorithm in place of the \diffalgshortname\ algorithm.

\begin{figure}[th]
\captionsetup[subfigure]{justification=centering}
\centering
    \begin{subfigure}{0.32\textwidth}
        \includegraphics[width=\textwidth]{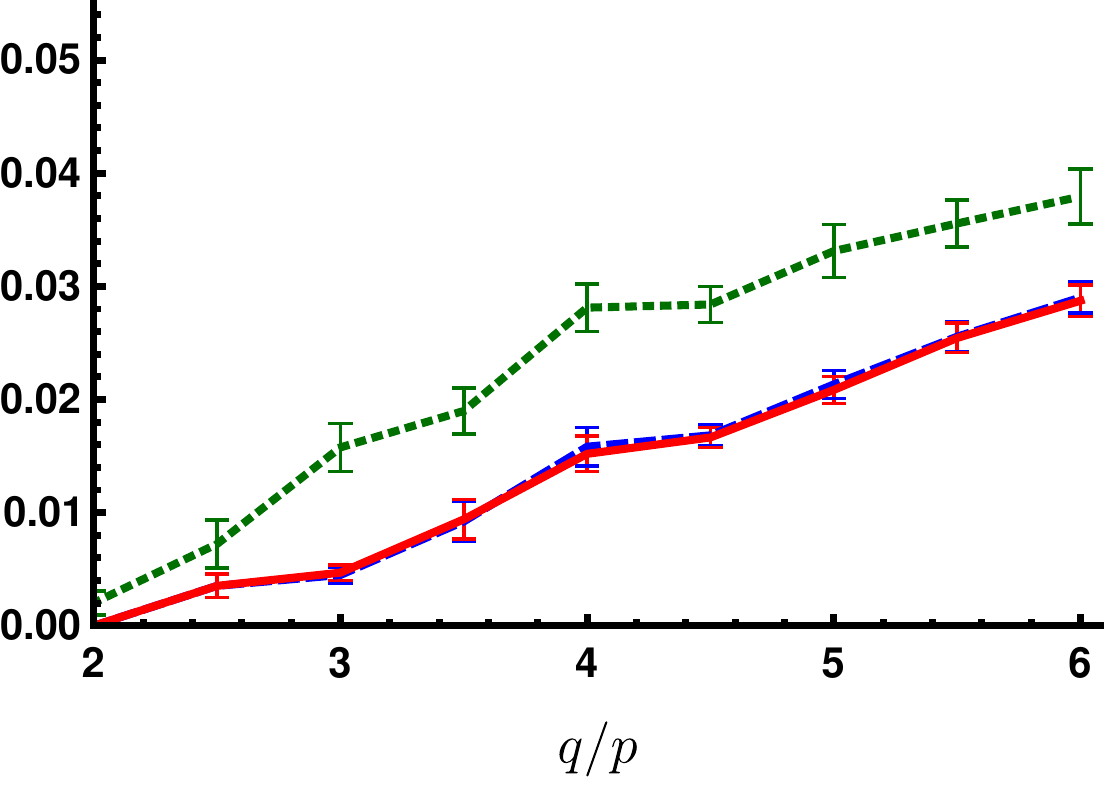}
        \caption{\centering $\bipart_\graphh$-value}
    \end{subfigure}
    \hspace{3pt}
    \begin{subfigure}{0.32\textwidth}
        \includegraphics[width=\textwidth]{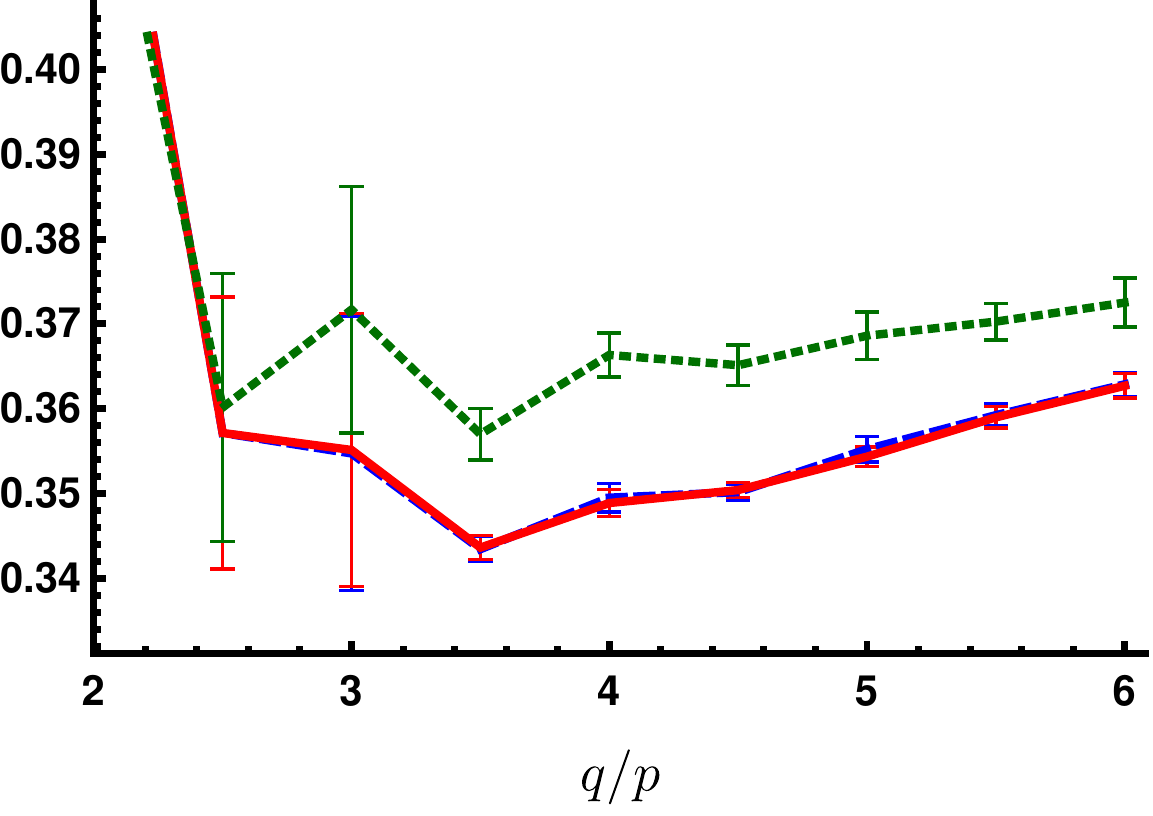}
        \caption{$\bipart_\graphg$-value}
    \end{subfigure}
    \hspace{3pt}
    \begin{subfigure}{0.31\textwidth}
        \includegraphics[width=\textwidth]{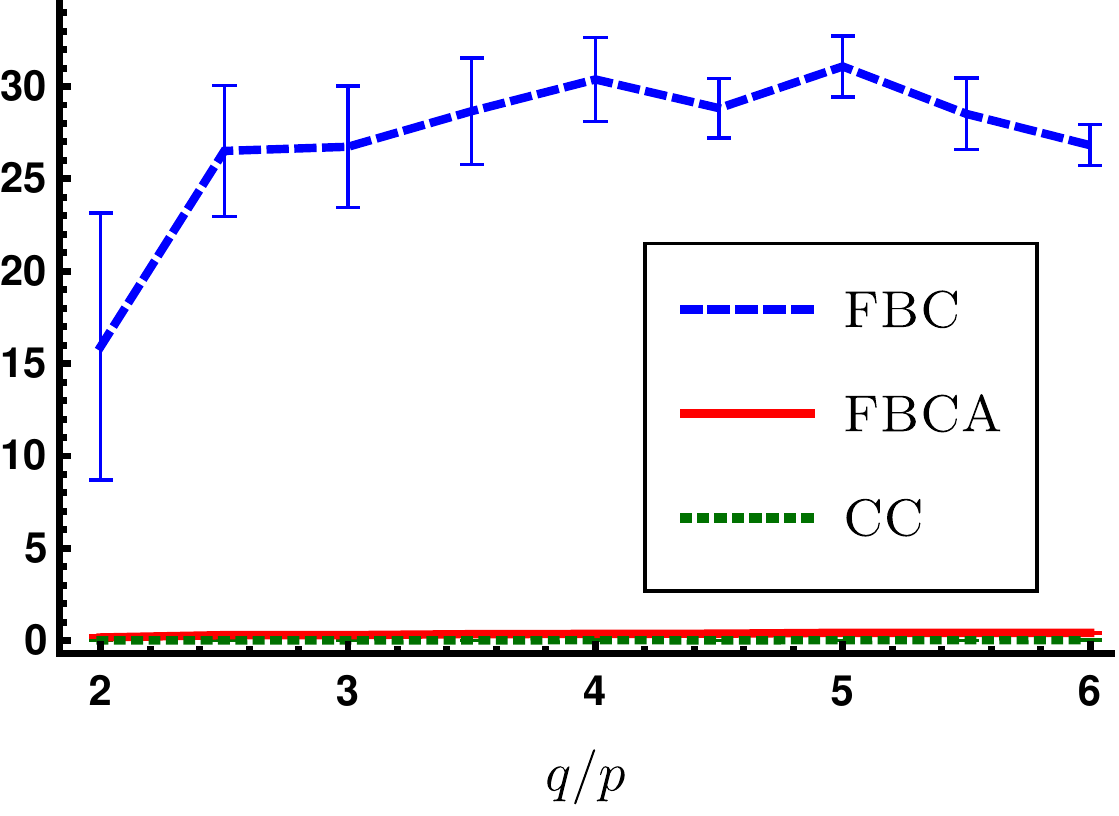}
        \caption{Runtime}
    \end{subfigure}
    \caption{The average performance and standard error of each algorithm when $n = 200$, $r = 3$ and $p = 10^{-4}$. 
    \label{fig:synthetic_rank3}
    }
\end{figure}

\paragraph{Experiments on larger graphs.}

\begin{table} 
  \caption{The runtime in seconds of the \diffalgapproxshortname\ and \textsc{CC} algorithms.}
  \label{tab:runtime}
  \centering
  \begin{tabular}{lllll}
    \toprule
    & & & \multicolumn{2}{c}{Avg.\ Runtime} \\
    \cmidrule(r){4-5}
    $r$ & $p$ & Avg.\ $\cardinality{\edgeset_\graphh}$ & \diffalgapproxshortname\ & \textsc{CC} \\
    \midrule
    & $10^{-9}$ & $1239$ & $1.15$ & $0.12$ \\
    $4$ & $10^{-8}$ & $12479$ & $10.14$ & $0.86$ \\
    & $10^{-7}$ & $124717$ & $89.92$ & $9.08$ \\
    \midrule
    & $10^{-11}$ & $5177$ & $3.99$ & $0.62$ \\
    $5$ & $10^{-10}$ & $51385$ & $44.10$ & $6.50$ \\
    & $10^{-9}$ & $514375$ & $368.48$ & $69.25$ \\
    \bottomrule
  \end{tabular}
\end{table}

We now compare only the \diffalgapproxshortname\ and \textsc{CliqueCut} algorithms, which allows us to run on hypergraphs with higher rank and number of vertices.
We fix the parameters $n = 2000$, $r = 5$, and $p = 10^{-11}$, producing hypergraphs with between $5000$ and $75000$ edges\footnote{In our model, a very small value of $p$ and $q$ is needed since in an $n$-vertex, $r$-uniform hypergraph there are $\binom{n}{r}$ possible edges which can be a very large number. In this case, $\binom{2000}{5} \approx 2.6 \times 10^{14}$.} and show the results in Figure~\ref{fig:synthetic_rank5}.
Our algorithm consistently and significantly outperforms the baseline on every metric and across a wide variety of input hypergraphs.
% This demonstrates that our algorithm outperforms the baseline on a wide variety of hypergraphs
% and we include additional results in Appendix~D to further highlight this.

To compare the algorithms' runtime, we fix the parameter $n = 2000$ and the ratio $q = 2p$, and report the runtime of the \diffalgapproxshortname\ and \textsc{CC} algorithms on a variety of hypergraphs in Table~\ref{tab:runtime}.
Our proposed algorithm takes more time than the baseline \textsc{CC} algorithm but both appear to scale linearly in the size of the input hypergraph\footnote{Although $n$ is fixed, the \textsc{CliqueCut} algorithm's runtime is not constant since the time 
% for both constructing the clique graph and 
to compute an eigenvalue of the sparse adjacency matrix scales with the number and rank of the hyperedges.}
which suggests that our algorithm's runtime is roughly a constant factor multiple of the baseline.

\begin{figure}[th]
\centering
    \begin{subfigure}{0.3\textwidth}
        \includegraphics[width=\textwidth]{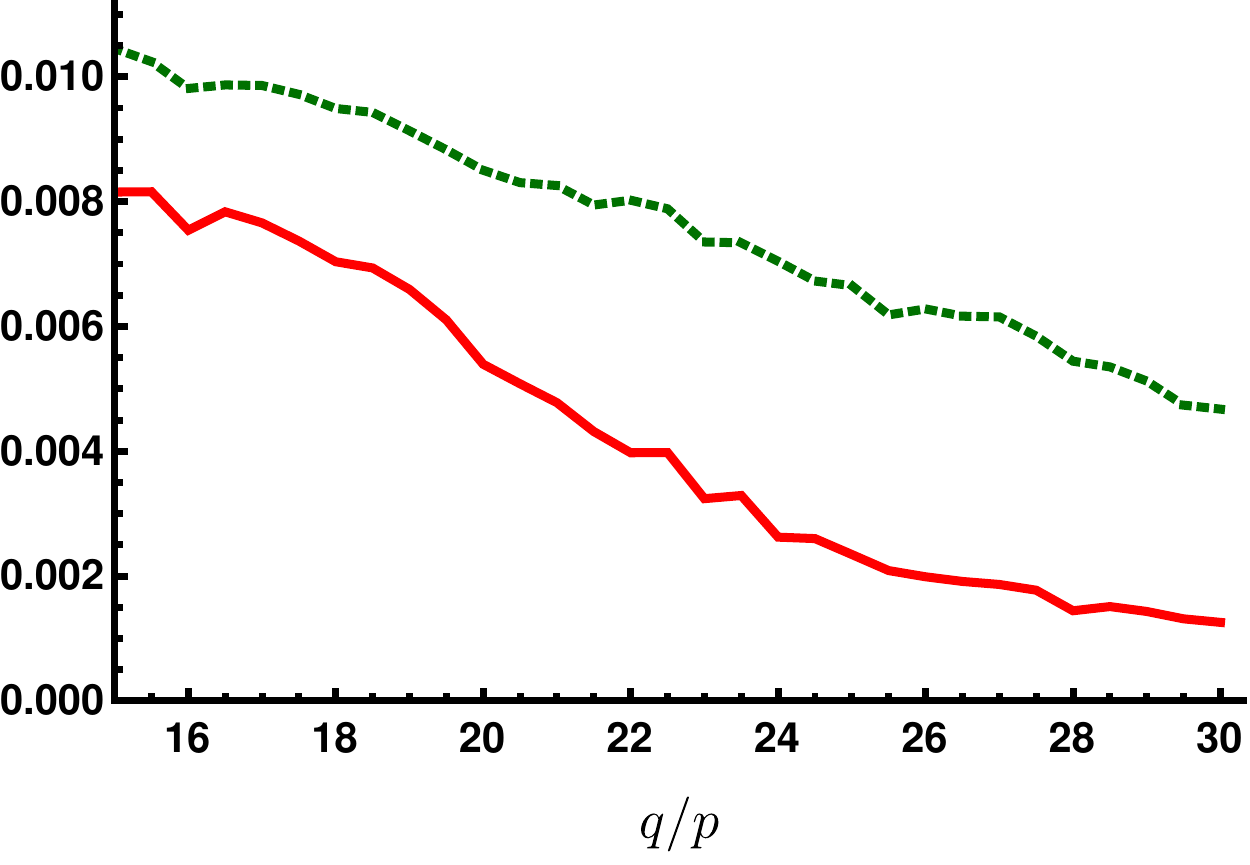}
        \caption{$\beta_H$-value}
    \end{subfigure}
    \hspace{3pt}
    \begin{subfigure}{0.3\textwidth}
        \includegraphics[width=\textwidth]{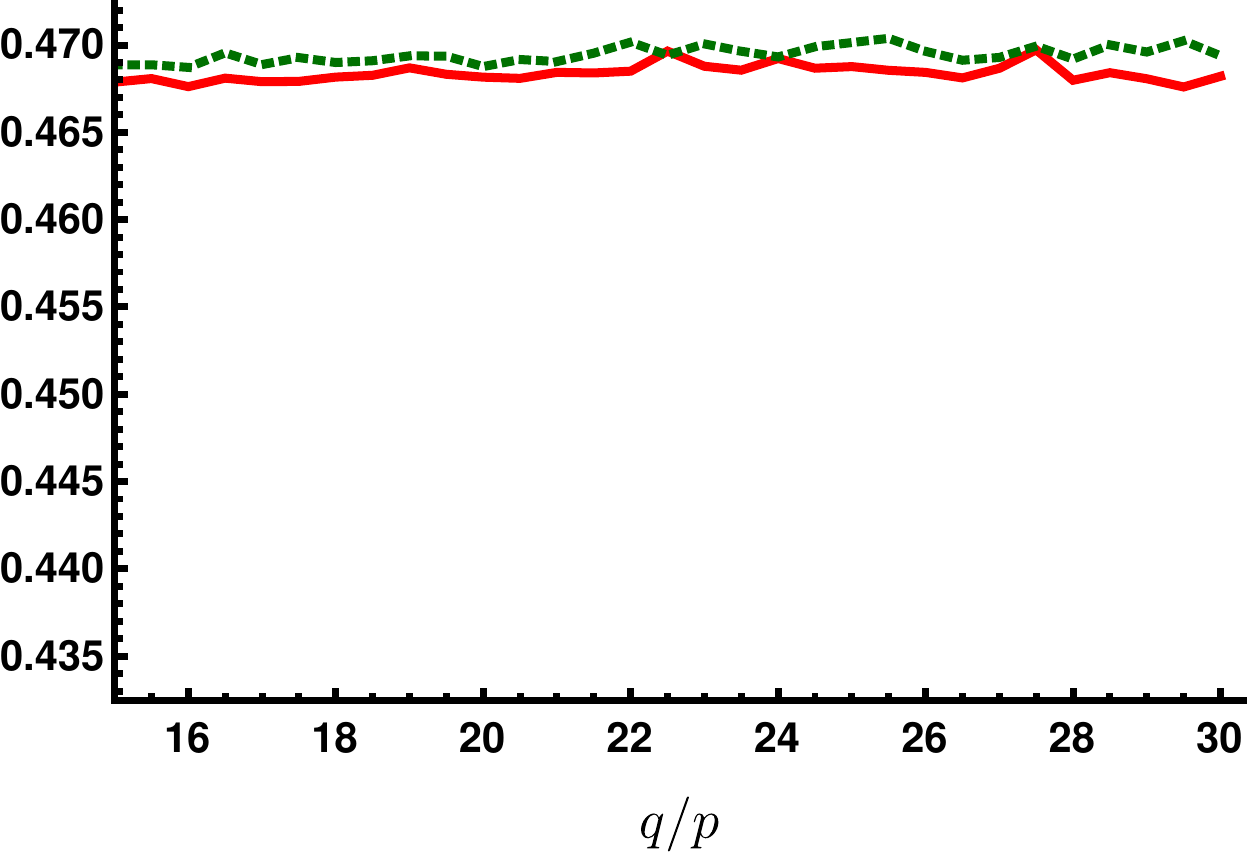}
        \caption{$\beta_G$-value}
    \end{subfigure}
    \hspace{3pt}
    \begin{subfigure}{0.3\textwidth}
        \includegraphics[width=\textwidth]{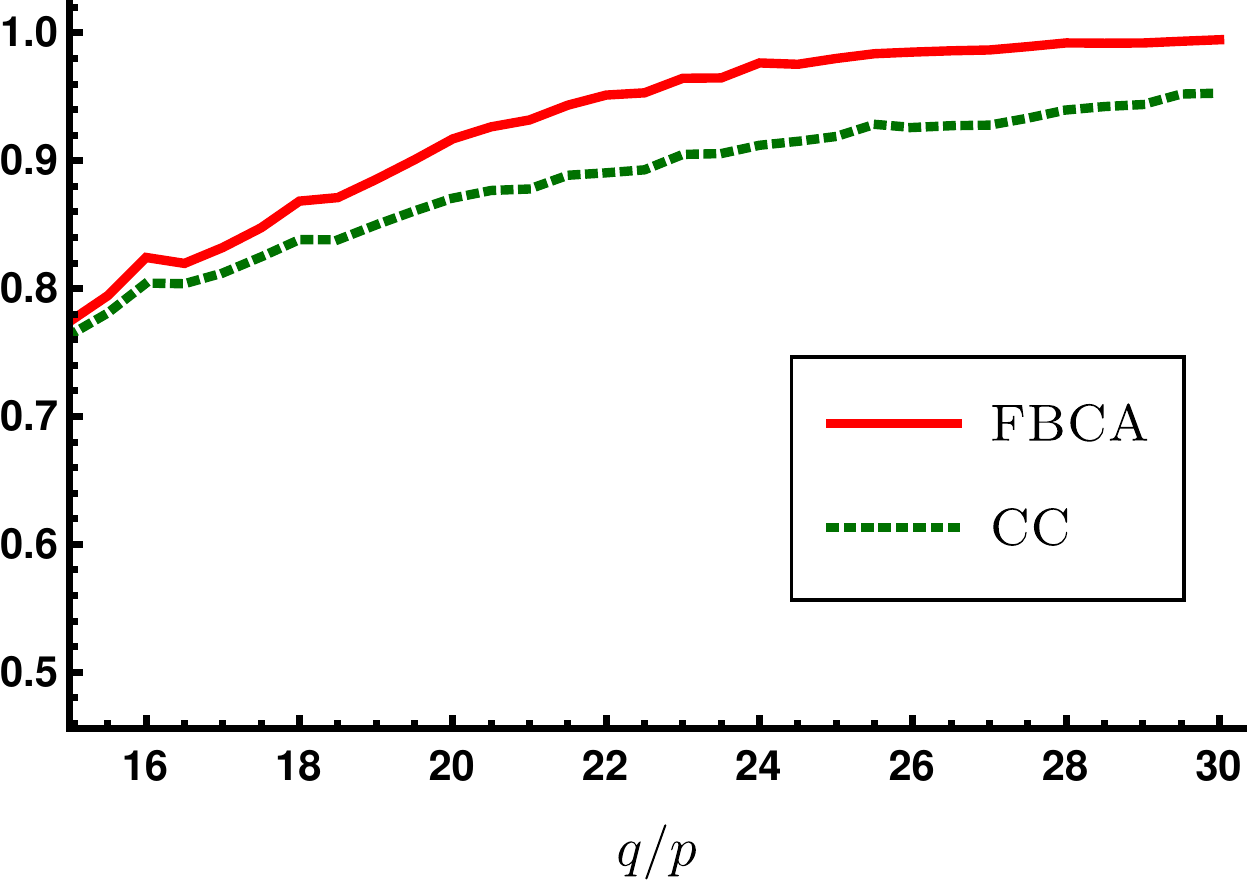}
        \caption{$F_1$-score}
    \end{subfigure}
    \caption{The average performance of each algorithm when $n = 2000$, $r = 5$, and $p = 10^{-11}$. We omit the error bars because they are too small to read.
    \label{fig:synthetic_rank5}
    }
\end{figure}

\subsection{Real-world datasets} \label{sec:real_world_experiments}
% Next, we demonstrate the broad utility of our new algorithm on several real-world datasets.
% \peter{Add more high-level discussion of the problem. Mention data cannot be represented as a graph. Two types of vertices, our algorithm treats them the same.}
 
Next, we demonstrate the broad utility of our algorithm on complex real-world datasets with higher-order relationships which are most naturally represented by hypergraphs.
Moreover, the hypergraphs are \emph{inhomogeneous}, meaning that they contain vertices of different types, although this information is not available to the algorithm and so an algorithm has to treat every vertex identically.
We demonstrate that our algorithm is able to find clusters which correspond to the vertices of different types.
% In each case, the data is represented by a hypergraph with vertices of two different types, forming two clusters which are densely connected to each other.
% Our new algorithm can separate the clusters in every case.
Table~\ref{tab:real_world} shows the $F_1$-score of the clustering produced by our algorithm on each dataset, and demonstrates that it consistently outperforms the \textsc{CliqueCut} algorithm.

\begin{table}
\vspace{-5pt}
  \caption{The performance of the \diffalgapproxshortname\ and \textsc{CC} algorithms on real-world datasets.} 
  \label{tab:real_world}
  \centering
  \begin{tabular}{llll}
    \toprule
    & & \multicolumn{2}{c}{$F_1$-Score} \\
    \cmidrule(r){3-4}
    Dataset & Cluster & \diffalgapproxshortname\ & \textsc{CC} \\
    \midrule
    \multirow{2}{1.6cm}{Penn Treebank} & Verbs & $\mathbf{0.73}$ & $0.69$ \\
      & Non-Verbs & $\mathbf{0.59}$ & $0.56$ \\
      \midrule
    % \multirow{2}{*}{Wikipedia} & Fields & $\mathbf{0.85}$ & $0.38$ \\
    %   & Countries & $\mathbf{0.97}$ & $0.73$ \\
    %   \midrule
    % \multirow{2}{*}{IMDB} & Actors & $\mathbf{0.58}$ & $0.55$ \\
    %   & Directors & $\mathbf{0.64}$ & $0.62$ \\
    %   \midrule
    \multirow{2}{*}{DBLP} & Conferences & $\mathbf{1.00}$ & $0.25$ \\
      & Authors & $\mathbf{1.00}$ & $0.98$ \\
    \bottomrule
  \end{tabular}
% \vspace{-10pt}
\end{table}

\paragraph[Penn Treebank.]{Penn Treebank.}%\footnote{\red{The Penn Treebank corpus has a proprietary license.}}.}
The Penn Treebank dataset is an English-language corpus with examples of written American English from several sources, including fiction and journalism~\cite{marcusBuildingLargeAnnotated1993}.
The dataset contains $49,208$ sentences and over $1$ million words, which are labelled with their part of speech.
We construct a hypergraph in the following way: the vertex set consists of all the  verbs, adverbs, and adjectives which occur at least $10$ times in the corpus, and for every $4$-gram (a sequence of $4$ words) we add a hyperedge containing the co-occurring words.
This results in a hypergraph with $4,686$ vertices and $176,286$ edges. 
The clustering returned by our algorithm correctly distinguishes between verbs and non-verbs with an accuracy of $67$\%.
% Although their approach is not directly comparable with ours, we note the similarities in this experiment with~\cite{mooreActiveLearningNode2011}.
This experiment demonstrates that our \emph{unsupervised} general purpose algorithm is capable of recovering non-trivial structure in a dataset which would ordinarily be clustered using significant domain knowledge, or a complex pre-trained model~\cite{akbikContextualStringEmbeddings2018, goldwaterFullyBayesianApproach2007}.

\paragraph{DBLP.}
We construct a hypergraph from a subset of the DBLP network consisting of $14,376$ papers published in artificial intelligence and machine learning conferences~\cite{DBLPComputerScience, wangHeterogeneousGraphAttention2019}.
For each paper, we include a hyperedge linking the authors of the paper with the conference in which it was published, giving a hypergraph with $14,495$ vertices and $14,376$ edges.  The clusters returned by our algorithm successfully separate the authors from the conferences with an accuracy of $100$\%.

\section{Concluding remarks}
% We further evaluate our algorithm on the Wikipedia dataset~\cite{wikipediaDataset} and the IMDB movie dataset~\cite{imdbDataset, WJS+2019}
% and we refer the reader to Appendix~D for details.
In this paper, we introduce a new hypergraph Laplacian-type operator and
apply this operator to design an algorithm that finds almost bipartite components in hypergraphs. 
Our experimental results demonstrate the potentially wide applications of spectral hypergraph theory, and so we believe that designing faster spectral hypergraph algorithms is an important future research direction in algorithms and machine learning.
This will allow spectral hypergraph techniques to be applied more effectively to analyse the complex datasets which occur with increasing frequency in the real world.

\bibliographystyle{alpha}
\bibliography{arxiv_references}

\end{document}